\documentclass[prd,aps,amsfonts,eqsecnum,superscriptaddress,nofootinbib,notitlepage,longbibliography]{revtex4-1}

\usepackage[a4paper, left=1.5cm, right=1.5cm, top=2cm, bottom=2cm]{geometry}

\usepackage[english]{babel}
\usepackage[utf8]{inputenc}
\usepackage[T1]{fontenc}
\usepackage{amsthm}

\newtheorem{thm}{Theorem}[section]
\newtheorem*{thm*}{Theorem}
\newtheorem{hyp}{Hypothesis}[section]
\newtheorem{assum}{Assumption}[section]
\newtheorem{defn}[thm]{Definition}

\newtheorem{fact}[thm]{Fact}
\newtheorem{lem}[thm]{Lemma}
\newtheorem{cor}{Corollary}[thm]

\newtheorem{prop}{Proposition}[thm]
\newtheorem*{prop*}{Proposition}

\usepackage{comment}

\usepackage{amsmath}
\usepackage{graphicx}
\usepackage[colorinlistoftodos]{todonotes}
\usepackage[colorlinks=true, urlcolor=violet, linkcolor=blue, citecolor=red, hyperindex=true, linktocpage=true]{hyperref}
\usepackage{mathtools}
\mathtoolsset{showonlyrefs=true}
\usepackage{bbm,microtype,mathrsfs,amsmath,amssymb,color,amsthm,graphicx,bm} 
\usepackage{tikz-cd}
\usepackage{xcolor}
\usepackage{braket}
\usepackage{multirow}

\renewcommand{\vec}{\bm}

\newcommand{\CB}{\mathcal{B}}

\newcommand{\BC}{\mathbb{C}}
\newcommand{\CD}{\mathcal{D}}

\newcommand{\BE}{\mathbb{E}}

\newcommand{\CH}{\mathcal{H}}

\newcommand{\CL}{\mathcal{L}}
\newcommand{\CM}{\mathcal{M}}
\newcommand{\CN}{\mathcal{N}}

\newcommand{\CO}{\mathcal{O}}

\newcommand{\CT}{\mathcal{T}}

\newcommand{\BZ}{\mathbb{Z}}

\newcommand{\vA}{\bm{A}}
\newcommand{\va}{\bm{a}}
\newcommand{\vB}{\bm{B}}
\newcommand{\vb}{\bm{b}}

\newcommand{\vD}{\bm{D}}

\newcommand{\vG}{\bm{G}}
\newcommand{\vH}{\bm{H}}
\newcommand{\vI}{\bm{I}}

\newcommand{\vK}{\bm{K}}
\newcommand{\bK}{\bar{K}}
\newcommand{\vL}{\bm{L}}

\newcommand{\vM}{\bm{M}}

\newcommand{\vO}{\bm{O}}
\newcommand{\vP}{\bm{P}}

\newcommand{\vU}{\bm{U}}

\newcommand{\vX}{\bm{X}}
\newcommand{\vY }{\bm{Y }}

\newcommand{\vsigma}{\bm{ \sigma}}
\newcommand{\bvsigma}{\bar{ \vsigma}}
\newcommand{\vrho}{\bm{ \rho}}
\renewcommand{\L}{\left}
\newcommand{\R}{\right}

\newcommand{\bnu}{\bar{\nu}}

\newcommand{\bu}{\bar{u}}
\newcommand{\bmu}{\bar{\mu}}
\newcommand{\bCL}{\bar{\CL}}
\newcommand{\bCD}{\bar{\CD}}
\newcommand{\bvH}{\bar{\vH}}
\newcommand{\bomega}{\bar{\omega}}
\newcommand{\tOmega}{\tilde{\Omega}}
\newcommand{\tCO}{\tilde{\CO}}

\newcommand{\dagg}{\dagger}

\newcommand{\vertiii}[1]{{\left\vert\kern-0.25ex\left\vert\kern-0.25ex\left\vert #1 \right\vert\kern-0.25ex\right\vert\kern-0.25ex\right\vert}}
\newcommand{\norm}[1]{\Vert {#1} \Vert}

\newcommand{\normp}[2]{\norm{#1}_{#2}}
\newcommand{\lnormp}[2]{\lnorm{#1}_{#2}}

\newcommand{\labs}[1]{\left\vert {#1} \right\vert}
\newcommand{\lnorm}[1]{\left\Vert {#1} \right\Vert}
\newcommand{\e}{\mathrm{e}}
\newcommand{\iunit}{\mathrm{i}}
\newcommand{\ri}{\mathrm{i}}
\newcommand*{\tr}{\mathrm{Tr}}
\newcommand*{\poly}{\mathrm{Poly}}

\newcommand*{\Supp}{\mathrm{Supp}}
\newcommand{\indicator}{\mathbbm{1}}

\begin{document}
\title{Fast Thermalization from the Eigenstate Thermalization Hypothesis}
\author{Chi-Fang (Anthony) Chen}
\email{chifang@caltech.edu}
\affiliation{Institute for Quantum Information and Matter,
California Institute of Technology, Pasadena, CA, USA}
\author{Fernando G.S.L. Brand\~ao}
\affiliation{Institute for Quantum Information and Matter,
California Institute of Technology, Pasadena, CA, USA}
\affiliation{AWS Center for Quantum Computing, Pasadena, CA}

\begin{abstract}
The Eigenstate Thermalization Hypothesis (ETH) has played a major role in understanding thermodynamic phenomena in closed quantum systems. However, its connection to the timescale of thermalization for open system dynamics has remained elusive. This paper establishes a rigorous link between ETH and fast thermalization to the global Gibbs state. Specifically, we demonstrate fast thermalization for a system coupled weakly to a bath of quasi-free Fermions that we refresh periodically. To describe the joint evolution, we derive a finite-time version of Davies' generator with explicit error bounds and resource estimates. Our approach exploits a critical feature of ETH: operators in the energy basis can be modeled by independent random matrices in a near-diagonal band. This gives quantum expanders at nearby eigenstates of the Hamiltonian and reduces the problem to a one-dimensional classical random walk on the energy eigenstates. Our results explain finite-time thermalization in chaotic open quantum systems. 

\end{abstract}
\maketitle

\setcounter{tocdepth}{2} 
{ \hypersetup{hidelinks} \tableofcontents } 

\section{Introduction}

Thermodynamic phenomena are ubiquitous in nature but highly non-obvious to analyze from first principles. Therefore, to make progress, we often make assumptions that have strong explanatory power. An influential one in quantum thermodynamics is the Eigenstate Thermalization Hypothesis (ETH)~\cite{Srednicki_1999,ETH_review_2016}. For an observable, it states that thermalization happens at the level of individual eigenstates of the Hamiltonian and that random matrix theory (RMT) can model the transitions between nearby energies. The folklore suggests ETH is generically fulfilled in chaotic systems, although its justification has been primarily numerical in several systems (see, e.g.,~\cite{ETH_review_2016}). 

In its various versions, ETH~\cite{Srednicki_1999} transparently explains static properties (how a closed system appears thermal at equilibrium) and infinite-time properties (fluctuations around time-averaged expectations). However, for open system dynamics, ETH has not yet been connected to the \textit{timescale} of thermalization. Consider a particular thermalization model such as coupling the system weakly to a heat bath or implementing a Gibbs sampler~\cite{Temme_Quantum_metro} on a quantum computer. Does ETH imply rapid convergence, or do we need another hypothesis for that? More generally, do we expect chaotic open systems to always thermalize in a reasonable time?  

This apparent lack of connection signifies the larger open challenge of understanding the approach to thermalization (and equilibration more generally). While there are general arguments for equilibration \cite{linden2009quantum} and thermalization (under ETH) \cite{Srednicki_1999,ETH_review_2016} at infinite times for both closed and open system dynamics, less is known at finite times\footnote{ 
To the best of our knowledge, the only works that study equilibration at finite times are~\cite{dymarsky2018bound,Dymarsky2018_isolated}. For closed systems, they relate the energy scale of RMT (which we denote by $\Delta_{RMT}$ in our work) to the time-scale of equilibration for time-averaged observables $1/T\int_0^{T} \cdot \tr[\rho \vO(t)] dt$. 
}. Even classically, the problem is intricate, as demonstrated by glassy dynamics. Yet, interesting progress was accomplished classically using Glauber dynamics, a family of stochastic processes that model the interaction of the system with a heat bath. There, the rapid convergence to thermality was linked to finite correlation length in the Gibbs measure \cite{martinelli1999lectures}. These results also featured in the development of classical algorithms in areas ranging from approximate counting \cite{jerrum1996markov} to computer vision \cite{li1994markov}, where sampling of Gibbs measures (i.e. Markov random fields \cite{clifford1990markov}) is an important primitive \cite{resnik2010gibbs}.

Addressing the timescale of thermalization for open quantum systems is substantially more challenging. Some progress was achieved for local commuting models where the quantum versions of Glauber dynamics are local themselves~\cite{kastoryano2016commuting,capel2021modified}. In contrast, very few general results~\cite{brandao2019finite_prepare, Temme_2013_gap} are known in the non-commuting case since even defining a realistic thermalization model for quantum systems is non-trivial. Davies derived a Lindbladian that models a system interacting weakly with a heat bath~\cite{davies74,davies76}. Unfortunately, the derivation requires an unphysical infinite-time and weak-coupling limit that allows the Lindbladian to distinguish exponentially close energies. It is still an open question to find a Lindbladian that faithfully models a system coupled weakly to a heat bath~\cite{Trushechkin_2021,REDFIELD1965,2017_Rivas}, under realistic assumptions, and that leads to thermalization.

Another class of processes modeling thermalization is known as quantum Gibbs samplers. These are quantum algorithms for preparing Gibbs states on a future quantum computer; the best-known example is Quantum Metropolis Sampling~\cite{Temme_Quantum_metro,QQmetropolis2012}. Further, similarly to the classical setting, they have shown to be an interesting quantum algorithmic primitive for more general problems (e.g., quantum algorithms for semi-definite programming \cite{brandao2017quantum} and machine learning \cite{amin2018quantum, anschuetz2019realizing}). While Quantum Metropolis should eventually converge to the (approximate) Gibbs state, the convergence at finite times for any non-commuting model has remained an open question. Even though we mainly take the physics angle of this work, one can certainly regard the system-bath joint dynamics as an algorithm (albeit with a potentially excessively large bath). See Section~\ref{sec:discussion} for the interplay between the physics and the algorithm and more recent developments there.

\subsection{Main Results}
This work addresses the timescale of open-system thermalization for Hamiltonians satisfying ETH, considering coupling the system to a refreshable thermal bath (Figure~\ref{fig:refresh}). Before we show convergence, we first generalize Davies' generator without taking the weak-coupling limit and then show ETH implies fast convergence to an approximation of the \textit{global} Gibbs state. 

\subsubsection{Davies' Generator at Finite Resources~{(Section~\ref{sec:Davies})} }

We formalize a version of Davies' generator under physically \textit{realistic} assumptions. We do not take any weak-coupling limit and only allow finite resources, which is different from the Lindbladian literature~\cite{REDFIELD1965,2017_Rivas,Trushechkin_2021}. 
These realistic assumptions mean our results are relevant for modeling thermalization in nature at finite times and Gibbs state preparation in digital or analog quantum computers. The generator we derive approximates the marginal of the joint evolution with an explicit trade-off between resources and accuracy. One crucial technical assumption we make to impose Markovianity is the routine refreshment of the bath.\footnote{
The idea of refreshing the bath also appears in the non-Markovian context~\cite{Purkayastha2020PeriodicallyRB}.
}\footnote{Qualitatively, routine refreshing avoids back-reaction and heating (or cooling) of the bath. Practically, this may not be necessary if the bath temperature changes only slightly, but we leave it for future work. } Practically, the refreshing can be engineered efficiently in a quantum computer. However, for thermalization that occurs in nature, the conditions when such a Markovian model is accurate remains open.

\begin{figure}[t]
    \centering
    \includegraphics[width=0.6
\textwidth]{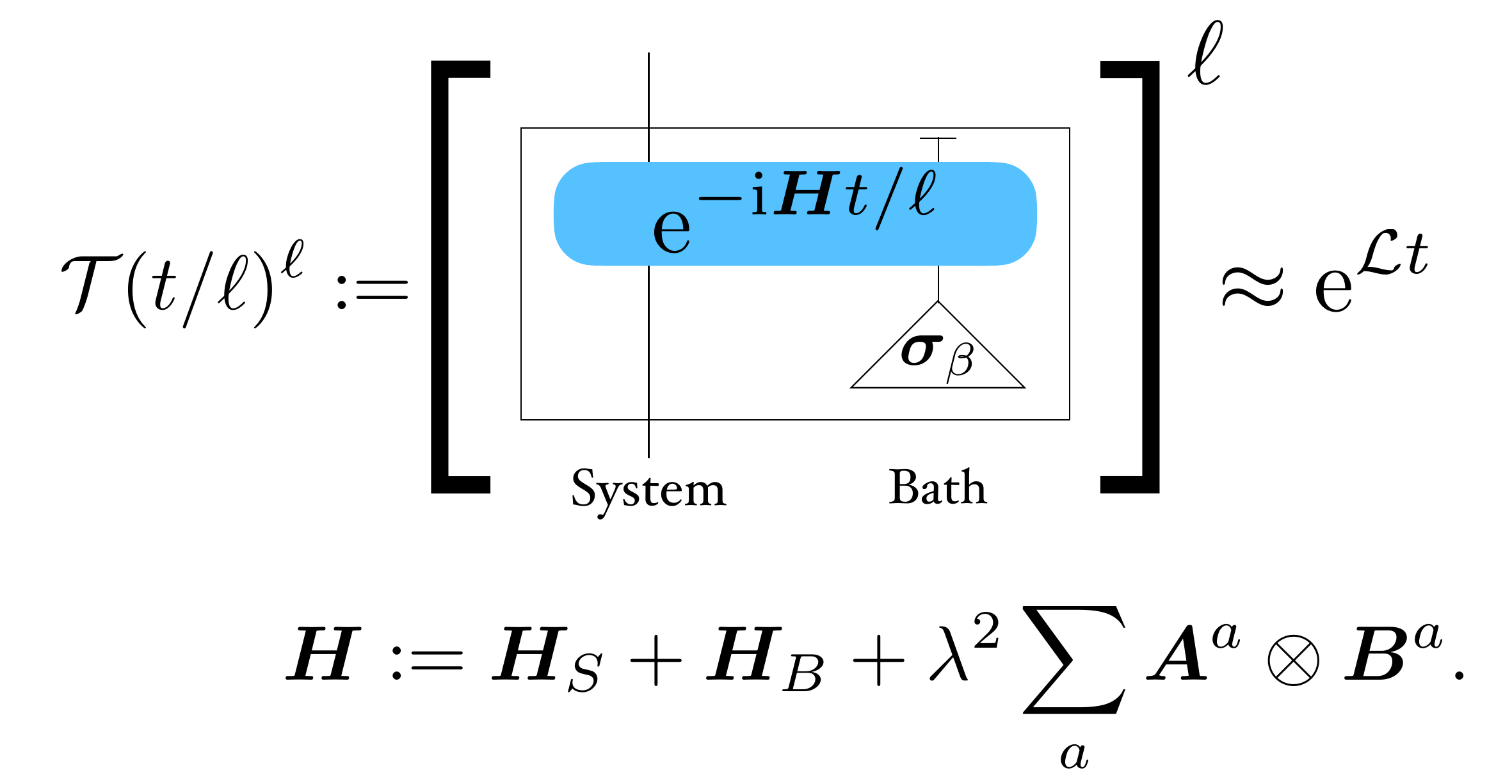}
    \caption{ The model of thermalization we consider in this work. We weakly couple the system to a quasi-free Fermionic bath that we routinely refresh. The process is markovian and parameterized by the Hamiltonian $\vH_S$, the set of interactions $\vA^a$, the inverse temperature $\beta$, and the bath.
    }
    \label{fig:refresh}
\end{figure}

Let us begin with an intuitive argument that will inspire a finite-time version of Davies' generator. Consider a system with Hamiltonian $\vH_S$ coupled weakly with a bath $\vH_B$ via interaction $\vH_I$
\begin{align} \label{eq:fullHam}
        \vH = \vH_S+\vH_B+\vH_{I}\quad\text{where}\quad \vH_{I}:= \lambda \sum_{a} \vA^a\otimes \vB^a,
\end{align}
where $\lambda$ is the strength of the interaction and $\vA^{a}$ are operators acting on the system (operators $\vB^a$ acting on bath). By taking an infinite-time and weak-coupling limit, the Davies' generator~\cite{davies74,davies76} for this Hamiltonian takes the Lindbladian form in the Heisenberg picture
\begin{align}
    \CD_{WCL}[\vX]&:= \sum_{\omega} \sum_{a}\gamma_{a}(\omega) \left( \vA^{a\dagg}(\omega) \vX \vA^a(\omega)-\frac{1}{2}\{\vA^{a\dagg}(\omega)\vA^a(\omega),\vX \} \right)\\
    \text{where}\quad   \vA^{a}(\omega) &:= \sum_{\nu_1-\nu_2=\omega} \vP_{\nu_2} \vA^a\vP_{\nu_1}\label{eq:davies_WCL}.
\end{align}
It implicitly depends on the system Hamiltonian $\vH_S$ via its eigenspace projectors $\vP_{\nu}$, energy label $\nu$, and the Bohr frequencies $\omega = \nu_1-\nu_2$. The Fourier-transformed interactions $\vA^a(\omega)$ drive transitions between energy eigenspaces. The bath dependence only comes in via the scalar function\footnote{Here, we simplify by assuming the function is diagonal $\gamma_{ab}(\omega) = \delta_{ab}\gamma_{a}(\omega)$. See Section~\ref{sec:WCL} for full generality.} $\gamma_{a}(\omega)$. Intuitively, this generates a classical Markov chain\footnote{If there are no degeneracies.} for which we can in principle estimate its gap~\cite{Temme_2013_gap} (Figure~\ref{fig:semi-classical}). 

At finite-times, however, we may effectively consider a \textit{rounded} Hamiltonian with energies $\{\bnu\} = \BZ\cdot \bnu_0$ being integer multiples of some rounding precision $\bnu_0$ 
\begin{align}
    \vH_S  = \sum_{\nu} \nu \vP_{\nu} \approx \sum_{\bnu} \bnu \vP_{\bnu} =: \bar{\vH}_S.
\end{align}
It approximates the original Hamiltonian for a sufficiently small resolution $\bnu_0$. Indeed, this is a manifestation of the energy-time uncertainty principle. However, the Davies' generator of the rounded Hamiltonian $\bvH_S$ is now\footnote{In a slightly different context, \cite{wocjan2021szegedy} also considers the Davies' generator for the rounded Hamiltonian $\bvH_S$. We find this interpretation of $\CD$ an intuitive justification. }
\begin{align}
    \bCD[\vX]&= \sum_{\bomega} \sum_{a}\gamma_{a}(\bomega) \left( \vA^{a\dagg}(\bomega) \vX \vA^a(\bomega)-\frac{1}{2}\{\vA^{a\dagg}(\bomega)\vA^a(\bomega),\vX \} \right) \quad \text{where}\quad \vA^{a}(\bomega) := \sum_{\bnu_1-\bnu_2=\bomega} \vP_{\bnu_2} \vA^a\vP_{\bnu_1}. \label{eq:intro_rounded}
\end{align}
The massive degeneracy in $\bvH_S$ drastically changes the original Davies' generator. The projector $\vP_{\bnu}$ contains the nearby energies and may have exponential rank $\e^{\Omega(n)}$ (Figure~\ref{fig:semi-classical}).   In other words, we see that the weak-coupling limit at infinite times~\eqref{eq:davies_WCL} does not consistently capture the massive coherence between nearby energies, which should be omnipresent at any reasonable run-time. 

\begin{figure}[t]
    \centering
    \includegraphics[width=0.95\textwidth]{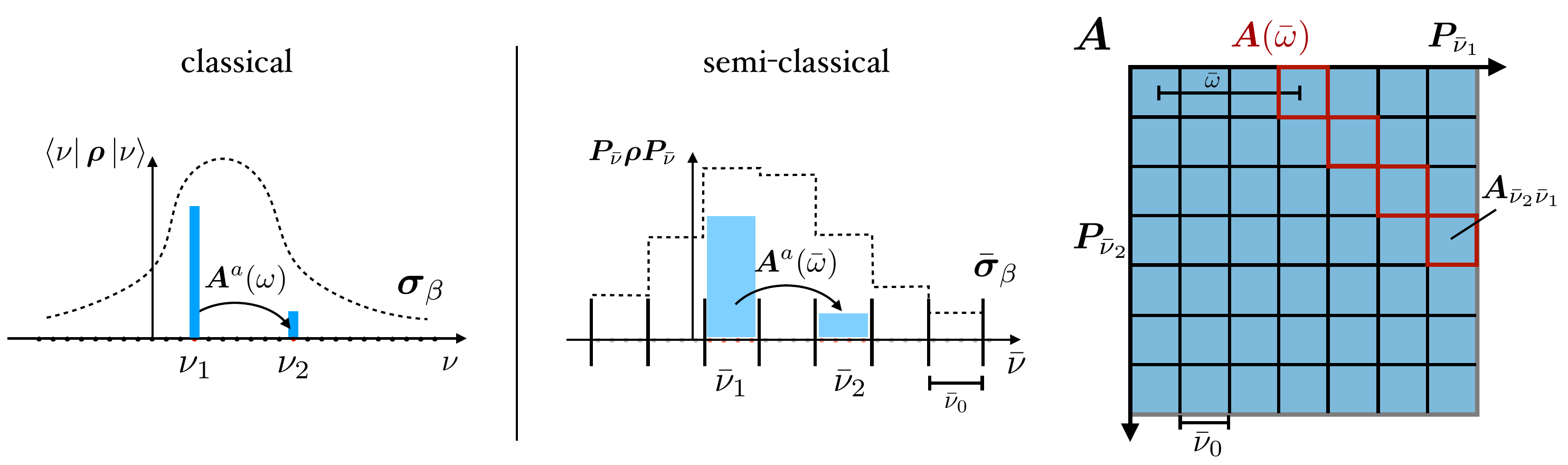}
    \caption{ (Left) The infinite-time limit leads to a classical Markov chain generator (for energy eigenstate inputs). The Gibbs state $\vsigma_{\beta}$ is a fixed point. (Middle) The finite-time evolution leads to a ``semi-classical'' generator that retains coherence within nearby energies ( for block-diagonal inputs in the energy basis). The rounded Gibbs state $\bvsigma_{\beta}$ is a fixed point. (Right) The operator at rounded frequency $\vA(\bomega)$ is dissected by projectors $\vP_{\bnu_1}$ and $\vP_{\bnu_2}$. Coherence remains within each subspace. 
    }
    \label{fig:semi-classical}
\end{figure}

Our first main result formalizes the above intuition. Construct a channel by coupling the \textit{rounded} Hamiltonian $\vH_S$ to a bath: initialize bath in the Gibbs state $\vsigma_B$ at the desired temperature, jointly evolve, and trace out the bath
\begin{align}
    \bar{\CT}(t)[\vrho]:=\tr_B\left[\e^{-\ri \bvH t}[ \vrho\otimes  \vsigma_B]\e^{\ri \bvH t}\right] \quad \text{where} \quad \bvH := \bvH_S + \vH_B + \vH_{I}. 
\end{align}
We show that the rounded generator $\bCD$ in~\eqref{eq:intro_rounded} indeed characterizes the joint-evolution with reasonable iterations and run-times.

\begin{thm}[Implementing the rounded generator, informal]\label{thm:rounded_Davies_informal}
For an $n$-qubit system coupled to a quasi-free Fermionic bath as in Eq. (\eqref{eq:fullHam}), 
the rounded generator $\bCL$ approximates the marginal evolution of the rounded Hamiltonian $\bvH$ with $\ell$ refreshes. For each input $\vrho$, 
\begin{align}
 \lnormp{\bar{\CT}(t/\ell)^\ell[\vrho] - \e^{ \bCL^\dagg t}[ \vrho]}{1} &\le \epsilon 
 \quad \text{where} \quad \bCL:=     \ri[ \bvH_S+\lambda^2\bvH_{LS},\cdot] + \lambda^2 \bCD
\end{align}
whenever certain polynomial constraints are satisfied between $t,\tau, \lambda, \ell,n$, the rounding precision $\bnu_0$, the error $\epsilon$ and the size of the bath.
\end{thm}
For our purposes for thermalization, we can skip the unitary part (the system Hamiltonian $\bvH_S$ and the Lamb-shift term $\lambda^2\bvH_{LS}$) and focus on the dissipative part $\bCD$~\eqref{eq:intro_rounded}. It has no cross term between different Bohr frequency blocks $\vA(\bomega)$ and $\vA(\bomega')$. Intuitively, this is because for a fixed rounding precision $\bnu_0$ and large enough evolution time $t$
\begin{align}
     \bnu_0t \gg 1,
\end{align}
cross terms with large frequency difference $\labs{\bomega - \bomega'}\ge \bnu_0$ decohere. Our rounded generator was inspired by~\cite{Trushechkin_2021}, which (after taking several limits) presented the dissipative part $\bCD$ that served our purposes. What allows us to put together a finite resource estimate is due to refreshing the bath and controlling the secular approximation without taking limits. See Theorem~\ref{thm:rounded_Davies} for the explicit trade-off between parameters. For example, one may set a desired effective time $\tau=\lambda^2 t$, and then estimate the required resources for implementation. 

However, the Hamiltonian $\bvH_S$ rounded at a fixed precision $\bnu_0$ is not physical\footnote{In general, the rounded Hamiltonian is not only unphysical but also difficult to implement in a quantum computer. Coherent phase estimation suffers from rounding errors. Still, we believe the rounded generator $\bCL$ captures the essential finite-time physics, in a way more transparent than the realistic generator $\CL$. } at later times
\begin{align}
    \e^{\ri \vH_S t } \not\approx \e^{\ri \bvH_S t } \quad \text{for} \quad \bnu_0t \gg 1.
\end{align}
For the physical, true evolution $\vH_S$, we also obtain a realistic generator $\CL$ for the marginal-joint evolution
\begin{align}
    \CT(t)[\vrho]:=\tr_B\left[\e^{-\ri \vH t}[ \vrho\otimes  \vsigma_B]\e^{\ri \vH t}\right] \quad \text{where} \quad \vH := \vH_S + \vH_B + \vH_{I}.
\end{align}
\begin{thm}[Implementing the realistic generator, informal]\label{thm:true_Davies_informal}
For a system coupled to a quasi-free Fermionic bath as in Eq. (\eqref{eq:fullHam}), 
the realistic generator $\CL$ approximates the marginal evolution with $\ell$ refreshes. For each input $\vrho$,
\begin{align}
    \lnormp{\CT(t/\ell)^\ell[\vrho] - \e^{ \CL^\dagg t}[ \vrho]}{1} &\le \epsilon 
\end{align}
whenever certain polynomial constraints are satisfied. 
\end{thm}
Unfortunately, the realistic generator $\CL$ takes a complicated, \textit{non-Lindbladian} form that we postpone to Section~\ref{sec:details_of_true_davies}. 

\subsubsection{ETH Implies Fast Convergence of Generators }
\begin{table}[t]
\begin{tabular}{|l|llll|}
\hline
 & \multicolumn{2}{l|}{Weak-coupling limit} & \multicolumn{2}{l|}{Finite-coupling with bath refresh} \\ \hline
Hamiltonian & \multicolumn{1}{l|}{Non-degenerate} & \multicolumn{1}{l|}{Degenerate} & \multicolumn{1}{l|}{Rounded $\bvH_S$} & Any  $\vH_S$ \\ \hline
Effective generator & \multicolumn{2}{l|}{Davies' generator~\cite{davies74,davies76}} & \multicolumn{1}{l|}{$\bCL$ (Theorem~\ref{thm:rounded_Davies_informal})} & $\CL$ (Theorem~\ref{thm:true_Davies_informal}) \\ \hline
\multirow{2}{*}{Technicality} & \multicolumn{1}{l|}{\multirow{2}{*}{Classical RW}} & \multicolumn{3}{l|}{Semi-classical RW} \\ \cline{3-5} 
 & \multicolumn{1}{l|}{} & \multicolumn{2}{l|}{Coherence at eigenspace} & \begin{tabular}[c]{@{}l@{}}Coherence at nearby energies;\\ Bohr frequency cross-terms; \\ non-CP\end{tabular} \\ \hline
Fixed point & \multicolumn{3}{l|}{Gibbs state~\cite{davies74,davies76}} & \begin{tabular}[c]{@{}l@{}}Approximate Gibbs state \\ (Section~\ref{sec:Lamb_error})\end{tabular} \\ \hline
Effective run-time $\tau$ & \multicolumn{1}{l|}{\cite{Temme_2013_gap}} & \multicolumn{3}{l|}{\begin{tabular}[c]{@{}l@{}}Logarithmic, assuming local commuting Hamiltonians~\cite{capel2021modified}\\ Polynomial, assuming ETH (Theorem~\ref{thm:ETH_rounded_Davies_convergence_informal}, Theorem~\ref{thm:ETH_true_Davies_convergence_informal})\end{tabular}} \\ \hline
\end{tabular}
    \caption{ Our contributions to thermalization of a system coupled to a bath. We first need an effective generator of the joint evolution. A well-known result is the Davies' generator~\cite{davies74,davies76}, but it unfortunately invokes an unphysical weak-coupling and infinite-time limit. Convergence was only known when the Hamiltonian is non-degenerate~\cite{Temme_2013_gap}. At finite times and finite refreshment of the bath, we obtain generators (for both the rounded and the true Hamiltonian) that retain coherence at nearby energies. We also prove convergence for both cases assuming ETH. }
    \label{table:system_bath}
\end{table}
Our next main result proves the convergence of the rounded generator $\bCL$. It depends on the Hamiltonian $\vH_S$ and some set of interactions $\{\vA^a\}$, and these are precisely prescribed from the following version of ETH. 

\begin{hyp}[Eigenstate Thermalization Hypothesis, simplified]\label{hyp:ETH_diagonal_informal}
In the energy basis $\{\ket{\nu} \}$ of the Hamitonian $\vH_S$, the operator $\vA$ satisfies
\begin{align}
A_{ij} = \bra{\nu_i}\vA\ket{\nu_j} &= O(\mu)\delta_{ij} +\frac{1}{\sqrt{dim(\vH_S) \cdot D(\mu)}}f_{\vA}(\omega) g_{ij}
\end{align}
for 
\begin{align}
    &f_{\vA}(\omega) &\text{(transition rates)}\\
    &dim(\vH_S) & \text{(dimension of Hilbert space)}\\
    &D(\cdot) &\text{(normalized density of states)}\\
    &\mu := (\nu_i + \nu_j)/2 \\
    &\omega := \nu_j-\nu_i &\text{(change-of-variables)}.
\end{align}
In addition, for energies in the window $\labs{\nu_i-\nu_j} \le \Delta_{RMT}$, the variables $g_{ij}$ are modeled by independent Gaussians $\BE [g_{ij}]=0$ and $\BE [g_{ij}^2]=1$. \end{hyp}

Intuitively, if one applies the operator $\vA$ to an energy eigenstate $\ket{\nu}$, the function $f_{\vA}(\omega)$ governs the transition rate for an energy difference $\omega$. For simplicity, we have assumed\footnote{In practice, it may depend on the energy $\mu$ (or the temperature) of the eigenstates. This simplification does not change the main argument of this work.} that the function $f_{\vA}(\omega)$ depends only on the energy difference $\omega$; some traditional formulations of ETH assume that the \textit{diagonal} matrix elements $O(\mu)$ give the thermal expectation of $\vA$ at the energy $\mu$, but it is irrelevant for this work. 

Instead, the crucial aspect of ETH we utilize is that for transitions below an energy scale $\labs{\omega} \le \Delta_{RMT}$, random matrix theory (RMT) governs the \textit{off-diagonal} entries of the operator $\vA$ in the energy eigenbasis (Figure~\ref{fig:f(omega)}).\footnote{This is reminiscent of Wigner's RMT model of heavy nuclei~\cite{nuclear_random}. } This random matrix energy scale $\Delta_{RMT}$ will come into our convergence rate. Its precise value is known to be non-universal\footnote{We thank Anatoly Dymarksy for clarifying this.} and has connections to the dynamics of the system (such as transport~\cite{dymarsky2018bound}), but all we need for our proof is a polynomial system size dependence~\cite{ETH_review_2016,dymarsky2018bound,2020_ETH_small_omega_Richter,2021_ETH_OTOC_Brenes,wang2021eigenstate}
\begin{align}
    \Delta_{RMT} = \Omega\L(\poly(\frac{1}{n})\R).
\end{align} 

Furthermore, we make the unconventional interpretation of ETH that a few (polynomially many) operators $\vA^a$ can be modeled as \textit{i.i.d.} samples of random matrices (below the scale $\Delta_{RMT}$). Our proof relies on the ``non-commutativity'' between these operators $\vA^a$, unlike the previous usage of ETH in terms of a single operator. On the other hand, our proof does not care about the entries outside the band $\labs{\omega}>\Delta_{RMT}$ .

More carefully, there are other implicit and mundane assumptions behind ETH; let us briefly instantiate them regardless. We assume that the density of states $D(\nu)$ is well-defined and varies slowly such that for any nearby frequencies $\labs{\nu-\nu'} \le \Delta_{RMT}$, the ratio of densities is bounded uniformly:
\begin{align} \label{Ddefinition}
\frac{D(\nu)}{D(\nu')} \le  R.    
\end{align} 
To be free from finite-sized effects, we truncate the Hamiltonian near the edge. See Section~\ref{sec:ETH} for a review of ETH as well as the details in the assumptions. 
\begin{figure}[t]
    \centering
    \includegraphics[width=0.8\textwidth]{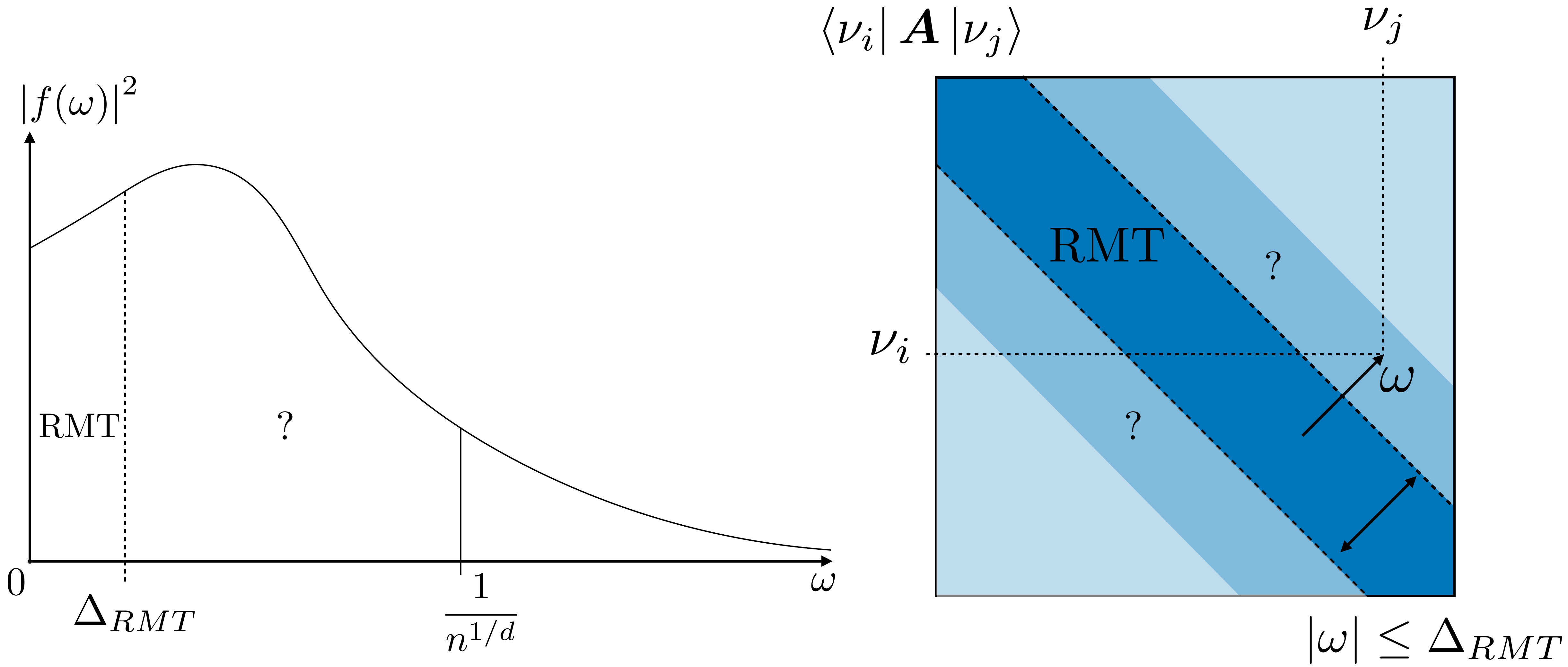}
    \caption{ The function $f_{\omega}$ in the ETH ansatz is expected to have most weight below some scale (e.g., $1/n^{1/d}$ for a d-dimensional lattice). The scale $\Delta_{RMT}$ that random matrix behavior kicks in is believed to be smaller and depends on the dynamics of the system~\cite{dymarsky2018bound}. The energy scales in between (the question mark) may partly exhibit random matrix behavior but retain correlation between entries~\cite{2020_ETH_small_omega_Richter,wang2021eigenstate,2021_ETH_OTOC_Brenes}. 
    }
    \label{fig:f(omega)}
\end{figure}

Now, let us show the above assumptions circling ETH suffice for thermalization at finite times. 
\begin{thm}[Convergence of the rounded generator, informal]\label{thm:ETH_rounded_Davies_convergence_informal}
Consider an n-qubit Hamiltonian $\vH_S$ and interactions $\vA^a$ that satisfy ETH in the sense the above. (In particular, the interactions $\vA^a$ are modeled by i.i.d.samples of random matrices within the band $\labs{\omega}\le \Delta_{RMT}$.) Assume
\begin{align}
    R&=\CO(1) &\text{(small relative ratio of DoS)}\\
    \beta\Delta_{RMT} &= \CO(1) &\text{(small ETH window)}\\
    \bnu_0 & \le 2\Delta_{RMT} &\text{(high rounding precision)}.
\end{align}
Then, with high probability (w.r.t to the randomness of ETH), running the rounded generator $\bCL$ for 
a few interactions and effective time 
\begin{align}
 \labs{a} =\Omega(1)\quad\text{and}\quad \tau &= \poly\L( n, \beta,\frac{1}{\labs{a}}, \frac{1}{\Delta_{RMT}},\frac{1}{\lambda_{RW}},\log(\frac{1}{\epsilon}) \R) \quad \text{ensures}\quad \lnormp{\e^{\bCL^\dagg t}[\vrho]- \bvsigma}{1} \le \epsilon
\end{align}
for 
\begin{align}
    &\bvsigma \propto \e^{-\beta {\bar{\vH}}_S} &\text{(rounded Gibbs state)}\\
    &\lambda_{RW}   &\text{(certain classical random walk gap)}.
\end{align}
\end{thm}

Our run-time\footnote{The variable $\tau = \lambda^2 t$ (instead of the physical time $t$) is the effective time for the dissipative part $\bCD$.} $\tau$ depends on many parameters. See Theorem~\ref{thm:ETH_rounded_Davies_convergence} for further details. Intuitively, the polynomial dependence on temperature $\beta$ means that low temperature Gibbs states can be efficiently prepared whenever ETH holds; this justifies the Gibbs state as a meaningful thermodynamic notion in open quantum systems. Of course, ETH does not apply to systems in the glassy phase or many-body-localized phase (see, e.g.,~\cite{2015MBL_Rahul}); there, the thermal state is unphysical anyway.

The run-time depends on the random matrix energy scale $\Delta_{RMT}$ because we only utilize the tiny RMT energy band (Figure~\ref{fig:f(omega)}) for the proof.
What's more mysterious is the appearance of a classical random walk. It is defined on the energy basis (weighted by the Gibbs state) with step size $\sim\Delta_{RMT}$. Roughly speaking, if the Gibbs state is a Gaussian with variance $\Delta_{Gibbs}^2 = \CO(\poly(n))$, then the gap is
\begin{align}
    \lambda_{RW} = \Omega\L( \frac{\Delta_{RMT}^2}{\Delta_{Gibbs}^2}\R) = \Omega\L( \poly( \frac{1}{n}) \R).
\end{align}
Therefore, the total physical run-time $t$ is also polynomial (By Theorem~\ref{thm:rounded_Davies_informal}, the coupling strength is polynomial $\lambda = \Omega(\poly(1/n))$) 
\begin{align}
    t = \frac{\tau}{\lambda^2} = \CO\L(\poly(n,\beta) \R). 
\end{align}
 
 Note that at low-temperatures $\beta \gg \Delta_{RMT}$, we would restrict to a smaller window $\Delta'_{RMT}= \CO(\frac{1}{\beta})$ to comply with the assumption $\beta \Delta_{RMT}' = \CO(1)$. This only polynomially impacts the run time (analogously the window $\Delta'_{RMT}$ should also ensure the ratio of densities is small $R=\CO(1)$).

Also, we present a similar convergence result for the realistic generator $\CL$. 
\begin{thm}[Convergence of the realistic generator, informal]\label{thm:ETH_true_Davies_convergence_informal}
Consider an n-qubit Hamiltonian $\vH_S$ and interactions $\vA^a$ that satisfy ETH in the sense the above. (In particular, the interactions $\vA^a$ are modeled by i.i.d.samples of random matrices within the band $\labs{\omega}\le \Delta_{RMT}$.) Assume
\begin{align}
    R&=\CO(1) &\text{(small relative ratio of DoS)}\\
    \beta\Delta_{RMT} &= \CO(1) &\text{(small ETH window)}.
\end{align}
Then, with high probability (w.r.t to the randomness of ETH), running the realistic generator $\CL$ for many interactions and effective time 
\begin{align}
    \labs{a} = \Omega(\frac{1}{\lambda_{RW}^2}) \quad\text{and}\quad \tau = \poly\L( n, \beta,\frac{1}{\labs{a}}, \frac{1}{\Delta_{RMT}},\frac{1}{\lambda_{RW}} ,\log(\frac{1}{\epsilon}) \R) \quad \text{ensures}\quad \lnormp{\e^{\CL^\dagg t}[\vrho]- \vsigma}{1} \le \epsilon.
\end{align}

\end{thm}
See Theorem~\ref{thm:ETH_true_Davies_convergence} for further details. The main qualitative differences (Table~\ref{table:system_bath}) between the two convergence results are that the rounded generator (1) requires fewer interactions $\labs{a}$, (2) has the rounding precision $\bnu_0$ given as an extra assumption, and (3) is a Lindbladian with nice properties.

\subsection{Discussion}\label{sec:discussion}

\begin{figure}[t]
    \centering
    \includegraphics[width=1.0\textwidth]{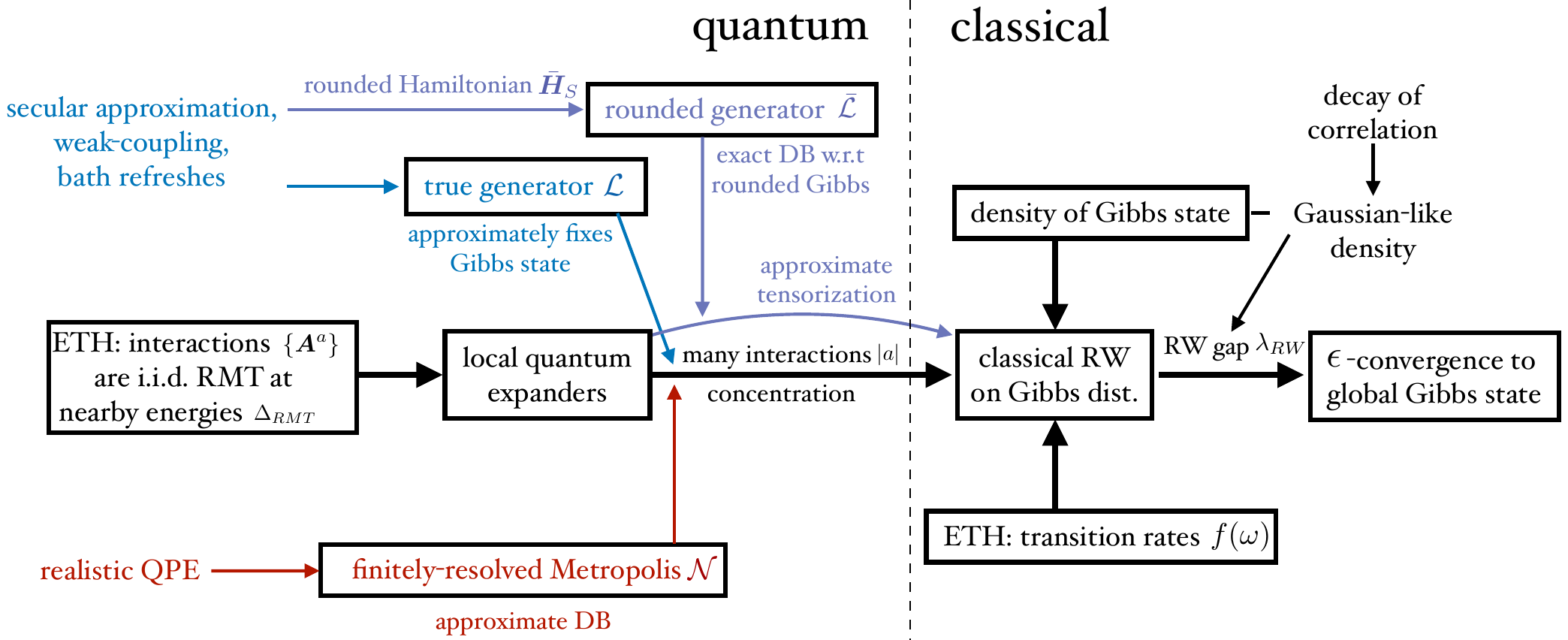}
    \caption{ The interdependence of the concepts and the parameters in this work. The colored texts and arrows distinguish the notions for the rounded generator and the realistic generator and those in common are in black. This manifests the flexibility of the presented arguments. The main assumption is ETH that the interaction terms are prescribed by i.i.d. random matrices. This gives quantum expanders at local energies and reduces the calculation to a classical random walk on the energy eigenbasis. Interestingly, the decay of correlation in Gibbs state, which featured in the classical~\cite{martinelli1999lectures} and commuting Hamiltonian literature~\cite{kastoryano2016commuting,capel2021modified}, is now a replaceable component. Here, it serves the only purpose that the density of states is Gaussian-like~\cite{brandao2015equivalence}, which, through standard conductance calculation, implies the random walk on the Gibbs distribution mixes rapidly. 
    }
    \label{fig:flow}
\end{figure}
Let us discuss a few noteworthy points of our main results.
\subsubsection{Removal of Quantum Metropolis Sampling discussions and related works}
In an earlier draft, we also studied Quantum Metropolis Algorithm~\cite{Temme_Quantum_metro} and presented analogous results; unfortunately, we recently realized the particular "shift-invariant boosted" phase estimation subroutine, which is crucial for the proof, is provably impossible (see~\cite{efficent_gibbs}). We do not have a simple fix and have completely removed the related discussions (including preliminary versions of approximate detailed balance that is later fully developed in~\cite{efficent_gibbs}). 

In follow-up works, the open system approach proves more elegant to analyze~\cite{efficent_gibbs} and natural to implement in digital or in analog~\cite{efficent_gibbs,Shtanko2021AlgorithmsforGibbs}. In particular, the work~\cite{efficent_gibbs} crystalizes the algorithmic essence of this work and further develops the analytic ideas circling finite energy resolution (the secular approximation). Technically, our paper contributed to control system-bath interactions at finite times but got stuck at a non-CPTP map (the realistic generator). The work~\cite{efficent_gibbs} substantially simplifies the picture by defining the appropriate Lindbladian, the appropriate notion of approximate detailed balance, and presenting quantum algorithms for simulating them. In other words, the essential functionality of the bath can be imitated using much fewer controlled registers. 

Let us also comment on the related work~\cite{Shtanko2021AlgorithmsforGibbs} that appeared soon after our draft became public. The high-level ideas are similar in spirit: construct a Gibbs sampler from a controllable bath and then show convergence assuming ETH. Our work focused on proving technical results: correctness of the fixed point and convergence assuming ETH from any initial state; the work~\cite{Shtanko2021AlgorithmsforGibbs} seems to focus on near-term applicability (whose bath construction appears to be conceptually simpler than ours), but its guarantee for Gibbs state relies simultaneously on three assumptions: the initial state is diagonal on the energy basis, ETH, and the mixing time of the classical walk on the energy basis. Indeed, their analysis is restricted to the diagonals and does not apply to the general entangled inputs state.
In contrast, our fixed point correctness does not require ETH but only the mixing time, and ETH is one of the ways of controlling the mixing time. We also provide the classical walk conductance calculation (Section~\ref{sec:conductance_1d}).

\subsubsection{Convergence to the Global Gibbs State}
One broader conceptual message is the strong notion of the convergence results: whenever ETH holds, \textit{every} initial state converges to a good approximation of the \textit{global} Gibbs state at reasonable times. This justifies the notion of the Gibbs state (for "chaotic" systems) in quantum thermodynamics. 

In quantum thermodynamics, it is unclear whether thermalization refers to convergence to the global Gibbs state or just to another state with similar properties (e.g., in terms of local marginals). People traditionally use the diagonals of ETH (which we do not use) to explain \textit{static} thermodynamics of closed systems: equilibrium at infinite times~\cite{Srednicki_1999} and/or of small subsystems~\cite{subsystem_ETH} (Section~\ref{sec:ETH}). Without an external bath, indeed, thermalization of small subsystems is the best we could hope for. Our results extend the applicability of ETH to open quantum systems using the off-diagonal RMT prescription.\footnote{Interestingly, ~\cite{Dymarsky2018_isolated} also used RMT to explain equilibration at finite times in closed systems. }

Technically, the convergence of Lindbladians (or channels), parameterized by arbitrary Hamiltonian $\vH$ and interactions $\vA^a$, seems intractable.  While there are well-established tools and examples in the classical Markov chain literature (see, e.g.,~\cite{Markovchain_mixing}), the quantum analogs are very much in their cradle. This is why some of our analysis appears adhoc (Section~\ref{sec:global_from_RW}) and draws heavily from RMT and matrix concentration.

\subsubsection{ETH Gives Quantum Expanders at Nearby Energies}\label{sec:intro_expander}

Technically, the proof idea is the close link between the RMT prescription of ETH and the rapid, thorough decoherence at local energies: random matrix theory gives \textit{quantum expanders}. Historically, a quantum expander refers to an ensemble of unitaries that efficiently generates 1-design (the maximally mixed state) with a gap between the largest ($\lambda_1$) and the second-largest ($\lambda_2$) eigenvalues (Figure~\ref{fig:expander}). Many constructions of quantum expanders are associated with random matrices. For example, the quantum channel composed of $\labs{a}$ Haar random unitaries and their adjoints has a gap $1-\CO(\frac{1}{\sqrt{\labs{a}}})$~\cite{Hastings_2007,hastings2008classical}\footnote{There are also notions of quantum tensor product expanders~\cite{hastings2008classical,brandao_local_tdesign} that generates t-designs.}. 

In our case, the random matrices arise from the small ETH window $\Delta_{RMT}$ and also give quantum expander. Let us illustrate this for the dissipative part $\bCD$ of the rounded generator $\bCL$. Suppose our input has only two energy blocks $ \vP_{\bnu_1}\vX\vP_{\bnu_1} + \vP_{\bnu_2}\vX \vP_{\bnu_2}$ (in the Heisenberg picture)
\begin{align}
   \CL_{\bnu_1,\bnu_2}[ \vX_{\bnu_1\bnu_1}+ \vX_{\bnu_2\bnu_2}] &:= \sum_{a}\bigg[\frac{\gamma(\bnu_1-\bnu_2)}{2} \left( \vA^{a}_{\bnu_1\bnu_2}  \vX  \vA^{a}_{\bnu_2\bnu_1}-\frac{1}{2}\{\vA^{a}_{\bnu_1\bnu_2} \vA^a_{\bnu_2\bnu_1},\vX_{\bnu_1\bnu_1} \} \right) + (\bnu_1\leftrightarrow\bnu_2) \bigg],
\end{align}
we show that the Lindbladian rapidly mixes the nearby energies.\footnote{The leading eigenvalue on the diagonal block is zero $\lambda_1=0$. } For all $\bnu_1,\bnu_2$, with high probability w.r.t to the randomness in ETH,
\begin{align}
     \lambda_2\L( \CL_{\bnu_1,\bnu_2} \R ) 
     &= - \Omega\L( \big(1 - \CO(\frac{1}{\sqrt{\labs{a}}})\big) \cdot \labs{a} \cdot \gamma(-\bomega)\labs{f(\bomega)}^2 \bnu_0 \R). \label{eq:maintext_expander}
\end{align}
Roughly speaking, using a few interactions $\labs{a}=\Omega(1)$ ensures the Lindbladian converges quickly.
See Section~\ref{sec:expanders} for a formal definition that also considers the off-diagonal inputs.\footnote{As a technical note, the above definition for quantum expander is intended for the cases when the density ratio between $ \bnu_1, \bnu_2$ and the Boltzmann factor $\e^{\beta \Delta_{RMT}}$ are $\CO(1)$. }
Our proof for quantum expanders relies on the ETH ansatz that the interactions $\vA^{a}_{\bnu_1\bnu_2}$, for close enough energies $\bnu_1-\bnu_2 \le \Delta_{RMT}$, are Gaussian matrices. We then go through concentration inequalities for a sum of tensor product of Gaussian matrices $\sum_a \vG_a \otimes \vG^*_a$, which follows~\cite{pisier2013rand_mat_operator}. 

In short, being a quantum expander in the sense above is how ETH (open) systems thermalize; at the same time, quantum expanders may be an alternative definition of ETH that is more precise and checkable than invoking RMT at nearby energies as a black box. Indeed, the gap~\eqref{eq:maintext_expander} does not refer to any randomness, and that is all we need for the proof. In small-scale numerics (a chaotic spin chain on $12$ qubits, Appendix~\ref{sec:numerics}), we observe the expander behavior as predicted by RMT ansatz.

\subsubsection{Global Convergence From a Random Walk on the Spectrum}\label{sec:global_from_RW}

Interestingly, once we obtain local convergence (roughly speaking, the expander property) using ETH, global convergence is controlled by a \textit{classical} random walk gap on the Gibbs distribution (Figure~\ref{fig:flow}). This is consistent with the known results~\cite{Temme_2013_gap} for infinite run-time where coherence is not an issue.

For the rounded generator $\bCL$ , the local-to-global lift is done via Modified-Log-Sobolev inequalities using a very recent development called \textit{approximate tensorization}~\cite{gao2021complete,laracuente2021quasifactorization}\footnote{This was developed for tensor product Hilbert spaces, but we applied it to the energy spectrum, which is more like a single particle Hilbert space. It is possible that an elementary derivation without approximate tensorization exists. Indeed, we could have used a similar argument as in the realistic generator $\CL$, but that requires quantum expander to hold at a large number of interactions $\labs{a}$ that scales with the inverse random walk gap squared $1/\lambda_{RW}^2$. }. The fact that our rounded generator satisfies exact detailed balance (w.r.t. the rounded Gibbs state) has also saved us from potential issues with approximate detailed balance.

For the realistic generator $\CL$, the proof of global convergence is more ad-hoc as it is not a Lindbladian. The dissipative part $\CD'$ is trace-preserving, detailed balanced, but not completely positive; the Lamb-shift term $\CL_{LS}$ only approximately fixes the Gibbs state. To ensure convergence to an approximate Gibbs state, we need to show the Lamb-shift term only incurs a small error to the Gibbs state. Despite the technical complication, the proof idea behind convergence again uses concentration for random matrices and gives an alternative interpretation of quantum expander. We use a large number of interactions $\labs{a}$ to ensure fluctuation in the dissipative part $\CD'$ concentrates around the expected map
\begin{align}
    \CD' \sim \BE\CD'
\end{align}
in the spectral norm \footnote{Technically, the concentration inequality tells us that the two generators $\CD'$ and $\BE\CD'$ share similar spectral properties, but not that they are close in the $1-1$ super-operator norm. In other words, the generator $\CD'$ is far from being classical, even though they share similar convergence rates and fixed points. }. Intriguingly, the expected map is essentially the \textit{infinite-time} Davies' generator, which generates a classical random walk whose gap we can calculate. To reiterate, even though our finite-time generator retains massive coherence at nearby energies, random matrix theory gives us access to the infinite-time object. 

See Section~\ref{sec:comment_optimal} for comments on the optimality of the Lindbladian results (Theorem~\ref{thm:true_Davies_informal}, Theorem~\ref{thm:ETH_rounded_Davies_convergence_informal}) and a comparison with the case of commuting Hamiltonians. 
 \begin{figure}[t]
    \centering
    \includegraphics[width=0.25\textwidth]{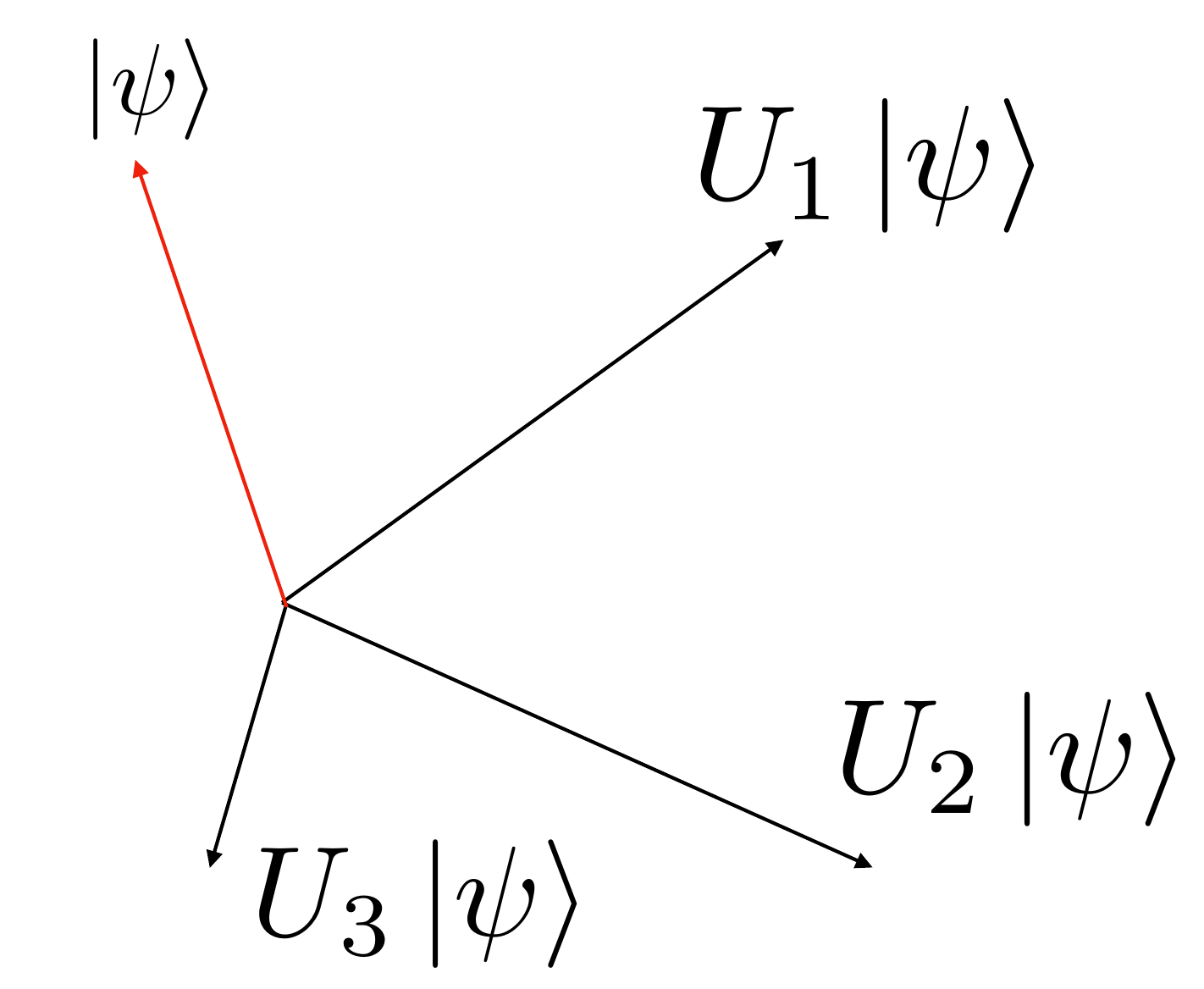}
    \caption{ A quantum expander rapidly mixes the inputs. An example is the channel with a few i.i.d. Haar random Kraus operators
    $\CN: = \frac{1}{2\labs{a}} \sum_a  \vU_a [ \cdot ] \vU_a^\dagger +  \vU_a^\dagger [ \cdot ] \vU_a$. Intuitively, the i.i.d. Haar random unitaries are ``pointing at different directions'' so that any input state $\ket{\psi}$ maps to nearly orthogonal states $\{\vU_a\ket{\psi}\}$. Roughly speaking, this connection between random matrices and quantum expander is why ETH leads to thermalization.
    }
    \label{fig:expander}
\end{figure}
\subsubsection{Implications to Quantum Algorithms}
Simulating thermal properties of physical systems is expected to be one important application of quantum computers. This paper shows that for the broad class of Hamiltonians satisfying ETH, we can prepare the associated quantum Gibbs states efficiently on a quantum computer (which is unexpected for classical methods). Our work suggests a class of physically relevant problems where we may expect a quantum advantage.

Quantum Gibbs sampling has been used as a primitive in many quantum algorithms (see, e.g., \cite{brandao2017quantum, anschuetz2019realizing, amin2018quantum, van2018improvements}). For example, one can solve semi-definite programs (SDPs) on a quantum computer \cite{brandao2017quantum, van2018improvements} in roughly the time required to prepare Gibbs states  (of Hamiltonians given by linear combinations of the input matrices of the instance). Therefore, whenever those matrices ``satisfy ETH'' (with suitable choices of interactions and suitable input model), the associated SDPs could potentially be solved in poly-log-dimension time; we leave end-to-end analysis for future work.

\subsubsection{Open Questions}
Even though we have shown ETH implies fast thermalization, many conceptual questions point towards a better understanding of ETH itself. First, there is a big gap (Table~\ref{table:system_bath}) between the convergence rate we obtained and the convergence rates obtained for the commuting case (assuming a finite correlation length). Resolving it would require a re-examination of the ETH parameters and a more stringent analysis. It will also be curious to compare the proof methods for the commuting case~\cite{kastoryano2016commuting,capel2021modified} (the locality of Lindbladian) with the random walk picture in this work. 

Second, a crucial assumption in our proofs is to interpret ETH for different operators $\{\vA^a\}$ as i.i.d. random matrices. It would be interesting to check this numerically and perhaps look for a more realistic or de-randomized formulation of ETH with many operators. A candidate we give is quantum expanders. We may ask whether this replaces the traditional ETH and check the expander property in larger-scale numerics.

Third, we assume the bath to be solvable and refreshable. We have fundamentally relied on the resulting Markovianity for our convergence results. For a more realistic model of nature, one should really study subsystem thermalization without refreshment. This seems much harder. It at least requires understanding the back-reaction from the bath beyond the leading order interaction picture and dealing with recurrences.

\subsection{Acknowledgments}
We thank Charles Xu for early discussions on the topic of this paper. We thank Li Gao, Ángela Capel, Cambyse Rouz\'e, and Daniel~Stilck França for discussions on approximate tensorization and comparison of our methods with their works~\cite{gao2021complete,capel2021modified}. We thank Patrick Rall, Pawel Wocjan, Sam McArdle, Alexander Dalzell, and Mario Berta for the discussion about phase estimation at a finite resolution. We thank András Gilyén for pointing out that the shift-invariant boosted phase estimation as given by~\cite{Temme_Quantum_metro} is impossible and collaboration on relevant topics. We thank Robert (Hsin-Yuan) Huang for suggesting checking quantum expander properties numerically. We thank Anatoly Dymarsky for discussions about the validity of random matrix prescriptions of ETH. CFC is supported by the Caltech RA fellowship and the Eddleman Fellowship. 

\section{Davies' Generator}\label{sec:Davies}
Consider our system of interest $\vH_S$ coupled to a heat bath. We can ask several basic questions: (1) when does the subsystem $S$ permit an effective Markovian description (i.e., Lindbladians)? (2) Is its fixed point the Gibbs state? (3) How fast does it converge?

\subsection{The Weak-Coupling Limit with Infinite Time and Bath}\label{sec:WCL}
The general question is hard; however, a Lindbladian can be derived by tuning the interaction by parameter $\lambda$ 
\begin{align}
    \vH = \vH_S+\vH_B+\vH_{I}\quad \text{where} \quad \vH_{I}:= \lambda \sum_{a} \vA^a\otimes \vB^a
\end{align}
and taking the \textit{weak-coupling limit} by 
\begin{align}
 \lambda\rightarrow 0\quad \text{fixing}\quad \lambda^2t = \tau <\infty.   
\end{align}
Technically, we also assume we start with a state tensored with the bath Gibbs state $\vrho \otimes \vsigma_B$ and that the bath are quasi-free Fermions (Section~\ref{sec:rounded_Davies}). Then, Davies showed that the marginal of the joint-evolution is effectively described by a Lindbladian $\CL$ acting only on the system S
\begin{align}
    \CT(t)[\vrho] = \tr_B\L[\e^{-\ri\vH t}( \vrho\otimes \vsigma_B )\e^{\ri\vH t}\R] \stackrel{\lambda\rightarrow0}{=} \e^{\CL^\dagg t}[\vrho_S].
\end{align}
The resulting Lindbladian, \textit{the Davies generators}~\cite{davies76, Rivas_2012_open_systems}, in the Heisenberg picture has the following form
\begin{align}
    \CL_{WCL}[\vX] = \ri[\vH_S +\lambda^2\vH_{LS},\vX]+ \lambda^2 \CD_{WCL}[\vX] \label{eq:WCL_LS_D}
\end{align}
with the dissipative part
\begin{align}
    \CD_{WCL}[\vX]= \sum_{\omega} \CL_{\omega} = \sum_{\omega, ab}\gamma_{ab}(\omega) \left( \vA^{a\dagg}(\omega) \vX \vA^b(\omega)-\frac{1}{2}\{\vA^{a\dagg}(\omega)\vA^b(\omega),\vX \} \right).
\end{align}
The Kraus operators $\vA^a(\omega)$ are the Fourier-tranformed interactions
 \begin{align}
     \e^{\ri \vH t}\vA^ae^{-\ri \vH t}=\sum_\omega \vA^a(\omega)e^{\ri\omega t} \quad \text{or} \quad \vA^a(\omega) = \sum_{\nu_1-\nu_2 =\omega} \vP_{\nu_2}\vA^a \vP_{\nu_1}.
 \end{align}
Implicitly, we refer to the Hamiltonian $\vH_S$ by 
 \begin{align}
     &\nu &\text{(energies of Hamiltonian $\vH_S$)}\\
     &\vP_{\nu} &\text{(energy projectors)}\\
     &\omega &\text{(Bohr frequencies)}.
 \end{align}
 Let us define the remaining variables  
\begin{align}
    \vH_{LS} &= \sum_{\omega} \sum_{ab} S_{ab}(\omega)  \vA^{a\dagg}(\omega)\vA^b(\omega) &\text{(the Lamb-shift term)}\\
    \gamma_{ab}(\omega) &= \Gamma_{ab}(\omega)+\Gamma_{ba}^*(\omega)\label{eq:gamma_ab}\\
    S_{ab}(\omega) &= \frac{1}{2\iunit}(\Gamma_{ab}(\omega)-\Gamma_{ba}^*(\omega))\label{eq:S_ab}\\
    \Gamma_{ab}(\omega) &= \int_0^\infty ds \e^{\iunit \omega s} \tr\L[\e^{\ri\vH_Bs}\vB^{a\dagg} \e^{-\ri\vH_Bs}\vB^b \vsigma\R]\\
    &:= \int_0^\infty ds \e^{\iunit \omega s} \braket{ \vB^{a\dagg}(s) \vB^b }_{\vsigma_B} &\text{(certain bath correlators)}.
\end{align}
Remarkably, the bath only comes in via the functions $\gamma_{ab}(\omega)$ and $S_{ab}(\omega)$.
From the above abstract forms, there are already unconditional properties for Davies' generator. It generates a CPTP map (i.e., a Lindbladian).
\begin{fact}[Trace-preserving]
The Davies' generator generates a trace-preserving map
\begin{align*}
    \CL_{WCL}[\vI]=0 \quad \text{and}\quad \tr[\e^{\CL_{WCL}^\dagg t}[ \vrho]] = \tr[ \vrho].
\end{align*}
\end{fact}

\begin{fact}[Completely-Positive]
The functions $\gamma_{ab}(\omega)$ give a positive-semidefinite matrix over $a,b$ for a fixed Bohr frequency $\omega$.
Consequently, the Davies' generator $\CL^\dagg_{WCL}$ generates a completely-positive map.
\end{fact}
 Further, it satisfies detailed balance. Together with trace-preserving, these imply the Gibbs state $ \vsigma \propto \e^{-\beta \vH}$ is stationary $(\CL_\omega+\CL_{-\omega})[ \vsigma] = 0$. 
\begin{fact}[Detailed balance]
The Davies' generator satisfies the KMS condition 
\begin{align}
    \gamma_{ba}(-\omega)=e^{-\beta \omega}\gamma_{ab}(\omega). \label{eq:KMS}
\end{align}
Consequently, for each frequency $\omega$, the symmetrized super-operator $\CL_\omega+\CL_{-\omega}$ satisfies the \textit{detailed balance condition} 
\begin{align*}
    \sqrt{ \vsigma}(\CL_\omega+\CL_{-\omega})[\vX]\sqrt{ \vsigma}=(\CL_\omega+\CL_{-\omega})^\dagg[\sqrt{ \vsigma}\vX\sqrt{ \vsigma}].
\end{align*}
\end{fact}

\subsection{The Advertised Form of the realistic generator}\label{sec:details_of_true_davies}
Given the above form, \cite{Temme_2013_gap} has shown the gap of Davies' generator (and setting the function to be diagonal $\gamma_{ab}(\omega) = \delta_{ab}\gamma(\omega)$) can be calculated given the description of the interactions $\vA^a$. However, it is not clear whether Davies' generator provides an accurate description of finite-time physics since the weak-coupling limit requires (1) an (unphysical) infinite time limit and (2) the bath to have a continuous spectrum and thus infinite-dimensional.

Our finite-time discussion highlights the following form of the realistic generator $\CL$. There are several related proposals~\cite{REDFIELD1965,2017_Rivas}, and the argument in~\cite{Trushechkin_2021} is a prototype of our finite-times results. Nevertheless, the crucial difference is that we adhere to finite resources and avoid taking any limits. This makes our result realistic but more complicated (see Section~\ref{sec:rounded_Davies} for a simplified version when the Hamiltonian is rounded). Ultimately, the infinitesimal generator $\CL$, like in~\eqref{eq:WCL_LS_D}, features the evolution term $\CL_S$, the Lamb-shift term $\lambda^2 \CL_{LS}$, and the dissipative part $\lambda^2\CD'$
\begin{align}
    \CL &:= \CL_S+\lambda^2 (\CL_{LS}+\CD' ).
\end{align}
The evolution term is the unitary evolution of the Hamiltonian \textit{rounded} to precision $\bnu_0$
\begin{align}
    \CL_S &:= \ri[\bvH_{S}, \cdot ] \quad \text{where}\quad \bar{\vH}_S = \sum_{\bnu} \bnu \vP_{\bnu}.
\end{align}
It differs slightly from the original Hamiltonian by $\norm{\vH_S-\bar{\vH}_S}\le \bnu_0/2$ and should be indistinguishable\footnote{This differs from the rounded generator $\bCL$, where the time should be large  $t\bnu_0 \gg 1$ with a fixed precision $\nu_0$. } at the desired time $t$
 \begin{align}
     \frac{1}{t} \gtrsim \bnu_0. 
 \end{align}
The rounding precision $\bnu_0$ is less physically meaningful but helps us simplify the expression by discretizing the energy labels.\footnote{Technically, this also serves as a regulator for the fourier series.}  The dissipative part is more complicated than Davies' generator 
\begin{align}
    \CD^{'}[\vX] &:= \sum_{ \labs{\bomega-\bomega'} \le \bmu_0 } \sum_{ab}\gamma_{ab}(\frac{\bomega+\bomega'}{2}) \left(  \vA^{a\dagg}(\bomega')\vX \vA^b(\bomega) -\frac{\e^{\beta \bomega_- }}{1+\e^{\beta \bomega_- }} \vX\vA^{a\dagg}(\bomega')\vA^b(\bomega) - \frac{1}{1+\e^{\beta \bomega_- }} \vA^{a\dagg}(\bomega')\vA^b(\bomega)\vX \right)\\
    & =:  \CD^{'}_{\vA\otimes \vA}[\vX] + \CD^{'}_{\vI \otimes \vA\vA}[\vX] + \CD^{'}_{\vA\vA\otimes \vI}[\vX].\label{eq:true_Davies}
\end{align}
Let us unpack the notations. The Fourier tranform of operators $\vA^{a}(\bomega)$ are parameterized by discrete frequencies $ \{\bnu\}, \{ \bomega\} , \{ \bomega'\} = \{ \BZ \bnu_0\}$ as multiples of the rounding precision $\bnu_0$
 \begin{align}
     \vA^{a}(\bomega)
     = \sum_{\bnu_1-\bnu_2=\bomega} \vP_{\bnu_2}\vA^a\vP_{\bnu_1}=:\sum_{\bnu_1-\bnu_2=\bomega} \vA^a_{\bnu_2\bnu_1}.
 \end{align}
The important frequency scale in~\eqref{eq:true_Davies} is the \textit{coherence width}
\begin{align}
    \bmu_0 := m\bnu_0.
\end{align}
The tunable integer $m$ allows this frequency scale $\bmu_0$ to differ from the rounding precision $\bnu_0$ (which can be arbitrarily small). Intuitively, the coherence width $\bmu_0$ sets the energy scale beyond which different Bohr frequencies $\bomega$ and $\bomega'$ are \textit{incoherent}. Providing a rigorous bound on the coherence width $\bmu_0$ marks the major technical difference from the many other candidates for approximating the evolution at finite times~\cite{REDFIELD1965,2017_Rivas}. Conceptually, this scale is rooted in the energy-time uncertainty principle
\begin{align}
    \bmu_0 \gtrsim  \frac{1}{t}, 
\end{align}
and a longer run time $t$ allows for a smaller coherence width $\bmu_0$.

Since there are now two Bohr frequencies $\bomega$ and $\bomega'$, the bath function depends on both frequencies $\frac{\bomega+\bomega'}{2}$. The anti-commutator~\eqref{eq:true_Davies} now has different prefactors to ensure detailed balance and trace-preserving (see below). It depends on the Bohr frequency difference
\begin{align}
    \bomega_- &:= \frac{\bomega -\bomega'}{2}.
\end{align}
The structure of the Lamb-shift term is less important 
\begin{align}
    \CL_{LS}:=\ri[\vH_{LS}, \cdot ]\quad \text{where} \quad \vH_{LS}&:= \sum_{ \labs{\bomega-\bomega'} \le \bmu_0} \sum_{a} S_{a}(\bomega,\bomega')  \vA^{a\dagg}(\bomega')\vA^a(\bomega),\\
    \quad\text{and}\quad S_{ab}(\bomega,\bomega') &= \frac{1}{2\iunit}\L(\Gamma_{ab}(\bomega)-\Gamma_{ba}^*(\bomega')\R).
\end{align}

 Technically, the realistic generator is less nice than the Davies' generator. The dissipative part $\CD'$ is only approximately completely positive. The Lamb-shift term $\CL_{LS}$ unfortunately only approximately preserves the Gibbs state. 
 
\subsubsection{The fixed point of the dissipative part}
The immediate question before diving into the proof is about the fixed point. The dissipative part $\CD'$ is nice as it satisfies detailed balance for the \textit{rounded} Gibbs state. 
\begin{prop}[The dissipative part is detailed balanced]\label{prop:D'_DB}
For the rounded Gibbs state $ \bvsigma\propto \e^{-\beta \bar{\vH_S}}= \sum_{\bnu}\e^{-\beta \bnu}\vP_{\bnu}$,
\begin{align}
    \CD^{'\dagg} = \sqrt{\bvsigma}\CD^{'} \L[\frac{1}{\sqrt{\bvsigma}}(\cdot)\frac{1}{\sqrt{\bvsigma}} \R] \sqrt{\bvsigma}.
\end{align}
\end{prop}
\begin{proof}[Proof of Proposition~\ref{prop:D'_DB}]
Observe 
\begin{align*}
    \sqrt{ \bvsigma}\vA^{a\dagg}(\bomega)=e^{-\beta \bar{\vH}/2}\vA^{a\dagg}(\bomega)  &= \e^{-\beta \bomega/2} \vA^{a\dagg}(\bomega)e^{-\beta \bar{\vH}/2}\\
    &=e^{-\beta \bomega/2} \vA^{a\dagg}(\bomega)\sqrt{ \bvsigma}.
\end{align*}
And similarly for the other term $ \vA^{a}(\bomega)\sqrt{ \bvsigma}=e^{-\beta \bomega/2}\sqrt{ \bvsigma}\vA^{a}(\bomega) $.
Hence, for the first term $\CD^{'}_{\vA\otimes \vA}$,
\begin{align}
     \sqrt{ \bvsigma}  \CD^{'}_{\vA\otimes \vA}[\vX]\sqrt{ \bvsigma} &= \sum_{ \labs{\bomega-\bomega'} \le \bmu_0 } \sum_{ab}\gamma_{ab}(\frac{\bomega+\bomega'}{2})   \sqrt{ \bvsigma} \vA^{a\dagg}(\bomega')\vX \vA^b(\bomega) \sqrt{ \bvsigma} \\
     &=  \sum_{ \labs{\bomega-\bomega'} \le \bmu_0 } \sum_{ab}\gamma_{ab}(\frac{\bomega+\bomega'}{2})\e^{-\beta(\bomega+\bomega')/2}   \vA^{a\dagg}(\bomega')\sqrt{ \bvsigma}\vX \sqrt{ \bvsigma}\vA^b(\bomega) \\
     &= \sum_{ \labs{\bomega-\bomega'} \le \bmu_0 } \sum_{ab}\gamma_{ba}(-\frac{\bomega+\bomega'}{2})   \vA^{a\dagg}(-\bomega')\sqrt{ \bvsigma} \vX \sqrt{ \bvsigma}\vA^b(-\bomega) \\
     & =   \CD^{'\dagg}_{\vA\otimes \vA}[\sqrt{ \bvsigma}\vX\sqrt{ \bvsigma}].
\end{align}

The second equality commutes the operator $\sqrt{ \bvsigma}$ through the interaction $\vA$ and the third equality uses the KMS condition~\eqref{eq:KMS} $\gamma_{ba}(-\omega)=e^{-\beta \omega}\gamma_{ab}(\omega)$.
For the remaining term $\CD^{'}_{\vA\vA\otimes \vI}[\vX]+ \CD^{'}_{\vI \otimes \vA\vA}[\vX]$, we calculate for one of them 
\begin{align}
    \sqrt{ \bvsigma}\CD^{'}_{\vA\vA\otimes \vI}[\vX]\sqrt{ \bvsigma} &=\sum_{ \labs{\bomega-\bomega'} \le \bmu_0 } \sum_{ab}\gamma_{ab}(\frac{\bomega+\bomega'}{2}) \left(  -\frac{\e^{\beta \bomega_- }}{1+\e^{\beta \bomega_- }} \sqrt{ \bvsigma}\vX\vA^{a\dagg}(\bomega')\vA^b(\bomega)\sqrt{ \bvsigma} \right) \\
    &=\sum_{ \labs{\bomega-\bomega'} \le \bmu_0 } \sum_{ab}\gamma_{ab}(\frac{\bomega+\bomega'}{2}) \left(  -\frac{1}{1+\e^{\beta \bomega_- }} \sqrt{ \bvsigma}\vX\sqrt{ \bvsigma} \vA^{a\dagg}(\bomega')\vA^b(\bomega) \right) \\
    &= \CD^{'\dagg}_{\vI \otimes \vA\vA}[\sqrt{ \bvsigma}\vX\sqrt{ \bvsigma}].
\end{align}
Repeat for the other term $\CD^{'}_{\vI \otimes \vA\vA}[\vX]$ and combine to obtain the advertised result.
\end{proof}
Further, we can calculate that the dissipative part $\CD'^{\dagg}$ generates a trace-preserving map in the Schrodinger picture (but not necessarily completely positive!).
\begin{prop}[Trace preserving]\label{prop:D'_TP}
\begin{align}
\CD'[\vI] = 0.
\end{align}
\end{prop}
Importantly, trace-preserving and detailed balance together imply that the rounded Gibbs state is a fixed point in the Schrodinger picture.
\begin{cor}[Gibbs fixed point]\label{cor:D'_fixed}
\begin{align}
    \CD'^{\dagg}[\bvsigma] = 0.
\end{align}
\end{cor}
Unfortunately, the Lamb-shift term $\CL_{LS}$ is trace-preserving but does not preserve the rounded Gibbs state (Section~\ref{sec:Lamb_error}).

\subsection{The Quasi-free Fermionic Bath}\label{sec:free_bath}
In the finite resource mindset, one must also choose a bath. Here, we collect the specifications, finite or infinite, completing the missing details in Section~\ref{sec:WCL}. Again, consider the total Hamiltonian
 \begin{align}
     \vH &=(\vH_S+\vH_B)+\lambda \sum_{a} \vA^a\otimes \vB^a\\
      &=: \vH_0+\vH_{I}
 \end{align}
where system $B$ consists of a direct sum of quasi-free Fermions with Hamiltonian
\begin{align}
    \vH_B=\sum_p^{n_B} \omega(p) \va^\dagg_{p}\va_p,
\end{align}
where $n_B$ is the dimension of \textit{single-particle} Hilbert space and note that $p$ may have energy degeneracy. Quasi-free refers to the (Wick-like) factorization of multipoint correlation into two-point correlation and is a consequence of the Hamiltonian $\vH_B$ being quadratic. 
\begin{align}
    \langle \va^\dagg(g_m)\cdots \va^\dagg(g_1) \va(f_1)\cdots \va(f_n)  \rangle_{\vsigma} = \delta_{mn} \textrm{det}(\langle \va(f_i)\va^\dagg(g_j)\rangle_{\vsigma} ) = \delta_{mn} \sum \epsilon_{\vec{i}\vec{j}}\prod \braket{f_{i_k}|g_{j_k}}
\end{align}
where $\epsilon_{\vec{i}\vec{j}}=\pm 1$ accounts for the signs (which thankfully we do not need to keep track of) 
\begin{align}
    \va(f)&:= \sum_p^{n_B} f^*(p)\va_p,\ \  \va^\dagg(g):= \sum_p^{n_B}g(p)\va^\dagg_p.\\
    \e^{\ri \vH_B t}\va(f)\e^{-\ri \vH_B t}&= \va(\e^{\ri h t}f) =\sum_p^{n_B} \e^{-\ri \omega(p)t}f^*(p)\va_p.
\end{align}

\subsubsection{Choosing a bath}
In this section, we will make a simple choice of bath(s). The finite-sized bath is what we can implement, while the infinite limit will simplify the calculation.  Label the single Fermion Hilbert space by
\begin{align*}
    p = \{\bu\} \times \{a\}  \stackrel{n_B=\infty}{\rightarrow} p = \{u\} \times \{a\} 
\end{align*} 
which means the Hamiltonian takes the form
\begin{align}
    \vH_B=\sum_{a, \bu} \bu \va^\dagg_{\bu, a}\va_{\bu,a} \stackrel{n_B=\infty}{\rightarrow} \vH_B=\sum_a \int  u \va^\dagg_{u, a}\va_{u,a} du.
\end{align}
In other words, for each discrete energy $\bu$ we introduce degeneracies per interaction term $\vA^a\otimes \vB^a$. Correspondingly, each slot $a$ is allocated for modes $\vB^a$
\begin{align}
    \vB_a = \sum_{\bu} f^*(\bu)\va_{\bu,a}+ \sum_{\bu} f(\bu)\va_{\bu,a}^\dagg \stackrel{n_B=\infty}{\rightarrow} \vB_a = \int f(u)\va_{u,a} du + \int f^*(u) \va^\dagg_{u,a} du. 
\end{align}
We obtain correlation functions
\begin{align}
    \braket{\vB^{a}(t)}_{\vsigma_B} &= 0\\
    \braket{\vB^{a}(t')\vB^{b}}_{\vsigma_B} &=  \delta_{ab} \sum_{\bu} \tr\L[ (\e^{\ri \bu t} \va^\dagg_{\bu,a} \va_{\bu} +\e^{-\ri \bu t}  \va_{\bu}\va^\dagg_{\bu,a} )\labs{f_{\bu}}^2 \frac{ \exp( -\beta \bu \va^{\dagg}_{\bu,a} \va_{\bu,a}) }{1+\e^{ -\beta \bu}} \R]\notag\\
    &= \delta_{ab} \sum_{\bu} (\e^{\ri \bu t-\beta \bu/2}+\e^{-\ri \bu t+ \beta \bu/2}  )\labs{f_{\bu}}^2 \frac{ \exp( -\beta \bu /2) }{1+\e^{ -\beta \bu}}\\
     \stackrel{n_B=\infty}{\rightarrow} \braket{\vB^{a}(t')\vB^{b}}_{\vsigma_{B'}} &= \delta_{ab} \int (\e^{\ri u t-\beta u/2}+\e^{-\ri u t+ \beta u/2}  )\labs{f_{u}}^2 \frac{ \exp( -\beta u /2) }{1+\e^{ -\beta u}} du.
\end{align}

where the cross term vanishes because different $\vB^a \vB^b$ acts on factorized Hilbert spaces, and each single term correlator vanishes. We also recalled the free Fermionic Gibbs state 
\begin{align}
    \vsigma_B = \prod_{\bu, a} \frac{ \exp( -\beta \bu \va^{\dagg}_{\bu,a} \va_{\bu,a}) }{1+\e^{ -\beta \bu}}.
\end{align}

The remaining parameters of bath are the functions $f_{\bu}, f'_{u}$:
\begin{align}
\labs{f'_{u}}^2 \frac{ \exp( -\frac{ \beta u}{2} ) }{1+\e^{ -\beta u}} &:=  \frac{\e^{-\frac{\beta^2\Delta^2_{B}}{8}}}{\sqrt{2\pi \Delta^2_{B}}} \exp( \frac{-u^2}{2\Delta^2_{B}} ) \\
\labs{f_{\bu}}^2 \frac{ \exp( - \frac{ \beta u}{2} ) }{1+\e^{ -\beta \bu}} &:= \begin{cases}
\displaystyle 0 &\textrm{if $\labs{\bu}\ge \bu_{max}$ }\\
\displaystyle\int_{\bu_-}^{\bu_+} \labs{f'_{u}}^2 \frac{ \exp( - \frac{ \beta u}{2} ) }{1+\e^{ -\beta u}} du.\ &\textrm{else}.
\end{cases}
\end{align}
The scale $\Delta_B$ is the variation of the energy of the bath, and through Fourier transform implies a timescale $\sim 1/\Delta_{B}$ of bath decay of correlation. 
\begin{prop}[Correlators]\label{prop:correlator}
\begin{align*}
    \braket{\vB^{a\dagg}(t')\vB^{a}}_{\vsigma_{B'}} &=\exp( \frac{\ri \beta \Delta^2_{Ba} t}{2}) \exp(-\frac{\Delta^2_{B}t^2}{2} ) +(h.c.) .
\end{align*}
\end{prop}
\begin{proof}
\begin{align}
    \braket{\vB^{a\dagg}(t')\vB^{a}}_{\vsigma_{B'}} &= \frac{\e^{-\frac{\beta^2\Delta^2_{B}}{8}}}{\sqrt{2\pi \Delta^2_{B}}} \int_{-\infty}^{\infty}\exp( \frac{-u^2}{2\Delta^2_{B}}-\frac{\beta u}{2}+\ri u t ) du+ \frac{\e^{-\frac{\beta^2\Delta^2_{B}}{8}}}{\sqrt{2\pi \Delta^2_{B}}} \int_{-\infty}^{\infty}\exp( \frac{-u^2}{2\Delta^2_{B}}+\frac{\beta u}{2}-\ri u t ) du\notag\\
    &= \e^{-\frac{\beta^2\Delta^2_{B}}{8}} \exp( \frac{\Delta^2_B(\beta - 2\ri t)^2 }{8} ) + (h.c.).
\end{align}
In the first line, the terms correspond to $\va^{\dagg}_{\bu,a} \va_{\bu,a}$ and $\va_{\bu,a} \va^{\dagg}_{\bu,a} $, respectively. The second equality is a Gaussian integral by completing the square, and our choice of normalization precisely cancels out the $\e^{\beta^2 \Delta^2_{B}/8}$ term. This is the advertised result.
\end{proof}
We will also need the integral of the (absolute) two-point correlator.

\begin{prop}\label{prop:abs_correlator}
\begin{align}
    c' &:= \int_0^\infty\sum_{a}\norm{\vA^{a}}^2 \labs{\braket{\vB^{a\dagg}(t')\vB^{a}}_{\vsigma_{B'}} } dt' \le 
    \labs{a} \CO \L( \frac{1}{\Delta_{B}} \R)\notag\\
    c(t) &:= \int_0^t\sum_{a}\norm{\vA^{a}}^2 \labs{\braket{\vB^{a\dagg}(t')\vB^{a}}_{\vsigma_B} } dt' \le 
    \labs{a} \CO \L( \frac{t(t +\beta)\bu_{max}\labs{a}}{n_B}+ t \exp( -\frac{(\bu_{max} - \beta V/2)^2}{2V})+ \frac{1}{\Delta_{B}} \R).\notag
\end{align}
\end{prop}
The correlator with finite bath is actually dominated by $\frac{\labs{a}}{\Delta_{B}}$ for all our purposes, and the reader should not be distracted by other terms.

\begin{proof}
We first evaluate the Fourier transform at an infinite system size limit 
\begin{align}
\int_0^\infty \exp(-\frac{\Delta^2_{B}t^2}{2}) dt =\CO(\frac{1}{\Delta_{B}}).     
\end{align}
Next, we compute the finite system size error in the integrand
\begin{align}
    &\braket{\vB^{a}(t)\vB^{a}}_{\vsigma_B}- \braket{\vB^{a}(t)\vB^{a}}_{\vsigma_{B'}} \\
    &=  \sum_{\bu} (\e^{\ri \bu t-\beta \bu/2}+\e^{-\ri \bu t+ \beta \bu/2} )\labs{f_{\bu}}^2 \frac{ \exp( -\beta \bu/2 ) }{1+\e^{ -\beta \bu}} - \int (\e^{\ri u t-\beta \bu/2}+\e^{-\ri u t+\beta \bu/2} ) \labs{f_{u}}^2 \frac{ \exp( -\frac{\beta u}{2} ) }{1+\e^{ -\beta u}} du\\
    &=\int_{\labs{\bu}\le \bu_{max}}\L[ (\e^{\ri \bu(u) t-\beta \bu(u)/2}+\e^{-\ri \bu(u) t+ \beta \bu(u)/2}  ) -  (\e^{\ri u t-\beta u/2}+\e^{-\ri u t+ \beta u/2} ) \R] \labs{f_{u}}^2 \frac{ \exp( -\beta u/2 ) }{1+\e^{ -\beta u}} du\notag\\
    &\hspace{2cm}+ \int_{\labs{\bu}>\bu_{max} }  (\e^{\ri u t-\beta u/2}+\e^{-\ri u t+ \beta u/2}  ) \labs{f_{u}}^2 \frac{ \exp( -\beta u/2 ) }{1+\e^{ -\beta u}} du \\
    &=\pm \CO\L( \frac{(t+\beta) \bu_{max}\labs{a}}{n_B}+ \exp( -\frac{(\bu_{max} - \beta \Delta^2_{B}/2)^2}{2\Delta^2_{B}})\R).
\end{align}
Together they yield the RHS.
\end{proof}
Similarly, we can calculate explicitly the functions $\gamma_{aa}(\omega)$, and indeed they satisfy the KMS condition
\begin{prop} \label{prop:gamma_profile}
\begin{align}
    \gamma_{ab}(\omega) = \delta_{ab}\frac{1}{\sqrt{2\pi \Delta^2_{B}}}  \exp (\frac{-(\omega-\beta \Delta^2_{B}/2)^2}{2\Delta^2_{B}})  = \gamma_{ba}^*(-\omega) \e^{\beta \omega}.\notag
\end{align}
\end{prop}
In other words, $\gamma_{ab}(\omega)$ has its weight at $\omega \sim [ \beta \Delta^2_{B} - \Delta_{B}, \beta \Delta^2_{B}+\Delta_{B}]$(Fig~\ref{fig:gamma}). (Later, we will make it overlap with $\Delta_{RMT}$. )
\begin{proof} Recall the definition~\eqref{eq:gamma_ab}
\begin{align}
        \gamma_{ab}(\omega)&= \int_0^\infty ds \e^{\iunit \omega s} \braket{ \vB^{a\dagg}(s) \vB^b }_{\vsigma_B}+\int_0^\infty ds \e^{-\iunit \omega s} \braket{ \vB^{a\dagg}(-s) \vB^b }_{\vsigma_B}\\
        &=\int_{-\infty}^\infty ds \e^{\iunit \omega s} \braket{ \vB^{a\dagg}(s) \vB^b }_{\vsigma_B}\\
        & = \delta_{ab}\frac{\e^{-\frac{\beta^2\Delta^2_{B}}{8}} }{\sqrt{2\pi \Delta^2_{B}}}  \exp (\frac{-\omega^2}{2\Delta^2_{B}}+\frac{\beta \omega}{2}).
\end{align}
Complete the square to obtain the advertised result.
\end{proof}
\begin{figure}[t]
    \centering
    \includegraphics[width=0.6\textwidth]{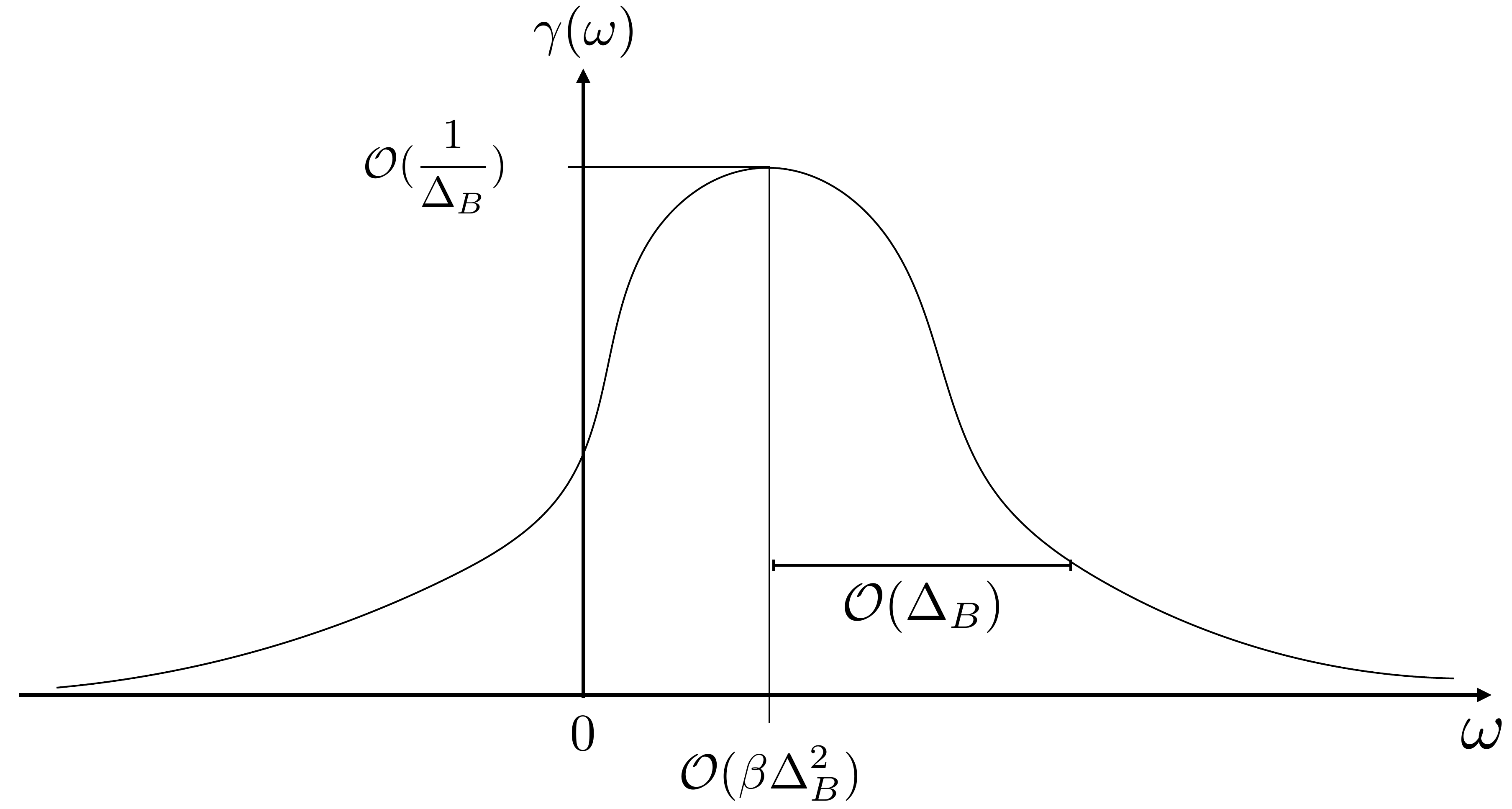}
    \caption{ The shape of function $\gamma$ with a tunable parameter $\Delta_{B}$. Later, for proving convergence with ETH, we can choose $\Delta_{B}$ such that the Gaussian aligns with the ETH transitions $\pm \Delta_{RMT}$. 
    }
    \label{fig:gamma}
\end{figure}

\section{Main Result I: Implementing Davies with Finite Resources}\label{sec:implement_true_Davies}

Deriving the original Davies' generator requires taking the weak-coupling limit $\lambda \rightarrow 0$, i.e., an infinite run time $t\rightarrow \infty$. There has been a surge of interest in proving non-asymptotic bounds~\cite{Mozgunov2020completelypositive,Nathan2020UniversalLE}, and existing results largely impose a Markovian assumption of the bath. While we could have directly begun our discussion with their results, to obtain an end-to-end quantum algorithm, we decided to explicitly construct a finite bath. The argument closely follows the usual open system derivation but with error bounds.

In this section, we show that the advertised realistic generator $\CL$ approximates iterations of the marginal joint-evolution 
 \begin{align}
     \CT(t)[\vrho]:=\tr_B\left[\e^{ (\CL^\dagg_0+\lambda \CL^\dagg_I)t}[ \vrho\otimes  \vsigma_B]\right] \quad \text{where}\quad \CL_0 = \ri [\vH_S +\vH_B, \cdot] \quad \text{and}\quad \lambda\CL_I = \ri [\vH_I, \cdot].
 \end{align}

\begin{thm}[Weak-coupling at finite times.]\label{thm:true_Davies}
 Assume there are $\labs{a}$ interaction terms in Lindbladian $\vH_{I}=\sum_a \lambda \vA^a\otimes \vB^a $.
 With a quasi-free Fermionic bath, the realistic generator $\CL$ at given effective time $\tau=\lambda^2t$ can be implemented with accuracy $\epsilon$:
\begin{align}
    \forall \vrho, \ \lnormp{\CT(t/\ell)^\ell[\vrho] - \e^{ \CL^\dagg t}[ \vrho]}{1} \le \epsilon, 
\end{align}
for the bath function
\begin{align}
    \gamma_{ab}(\omega) = \delta_{ab}\frac{1}{\sqrt{2\pi \Delta^2_{B}}}  \exp (\frac{-(\omega-\beta \Delta^2_{B}/2)^2}{2\Delta^2_{B}}),
\end{align}
whenever 
\begin{align*}
    \ell&=\theta\L( \frac{\labs{a}^2\tau^2}{\epsilon \Delta_{B}^2}\R)&(\textrm{bath refreshes}),\\
    \bmu_0 &:= m\bnu_0 = \CO\L(\frac{\epsilon^{11}\Delta^9_B}{\labs{a}^8\tau^{10}(\beta^2+\frac{1}{\Delta^2_B})} \R) &(\textrm{coherence width}),\\
    n_B&= \tilde{\theta}\L( \labs{a}^2\tau \frac{t}{\ell} (\frac{t}{\ell}+\beta)\frac{ (\beta \Delta^2_{B}+\Delta_B)}{\epsilon} \R)  &(\textrm{size of bath}),\\
    t&= \Omega\L( \frac{\labs{a}\tau \ell}{\epsilon \Delta_{B}} + \frac{\labs{a}^4\tau^4}{\bmu_0\epsilon^4\Delta_B^4}\R)&(\textrm{total physical run-time}),\\
    \bnu_0 &:= \CO\L(\frac{\epsilon^{16}\Delta_B^{13}}{\labs{a}^{12}\tau^{14}(\beta^2+\frac{1}{\Delta^2_B})} \R) &(\textrm{rounding precision}).
\end{align*}
\end{thm}

The theorem presents the physical resources (total time $t$, bath refreshes $\ell$, and size of bath $n_B$) needed to implement for a desired effective time $\tau = \lambda^2 t$ and a coherent width $\bmu_0$. The implicit parameters (weak-coupling strength $\lambda $, integer $m$) have been suitably chosen. Nevertheless, the approximation holds whenever certain inequalities are satisfied and see the proof for the complete dependence (Section~\ref{sec:proof_true_Davies}). For example, if one fixes a time $t$ and coupling strength $\lambda$, one may infer the appropriate coherence width $\bmu_0$ for which an $\epsilon$ approximation holds. This flexibility is the merit of not taking limits. See Section~\ref{sec:comment_optimal} for comments on optimality.
\subsection{Proof of the realistic generator (Theorem~\ref{thm:true_Davies})}\label{sec:proof_true_Davies}

\begin{proof}[Proof of Theorem~\ref{thm:true_Davies}]

First, we round the Hamiltonian at the resolution $\bnu_0$
\begin{align}
\bar{\vH}_S = \sum_{\bnu} \bnu \vP_{\bnu},
\end{align}
which differs slightly from the original Hamiltonian $\norm{\vH_S-\bar{\vH}_S}\le \bnu_0/2$. Of course, the resolution $\bnu_0$ needs to be small enough that the rounded Hamiltonian accurately approximate the Hamiltonian evolution up to the total evolution time $t$ 
\begin{align}
    \norm{ \vH_S - \bvH_S} t = \frac{\bnu_0}{2} t = \CO(\epsilon).
\end{align}
This discretization simplifies the presentation; it is technically important but less physically meaningful. For the rest of the proof, we will never refer to the unrounded Hamiltonian $\vH_S$. To simplify notations, we only keep the bars on the Hamiltonian and denote the rounded Liouvilian as $\CL_S:= \ri [\bvH_S,\cdot]$ and $\CL_0:= \CL_S+ \CL_B $.
The proof composes of several simple estimates, and we outline the flow of equations.
\begin{align*}
    \CT_{t}:=\tr_B\left[\e^{ (\CL^\dagg_0+\lambda \CL^\dagg_I)t}[ \vrho\otimes  \vsigma_B]\right]
    &\stackrel{1}{\approx} \tr_B\bigg[\e^{ \CL^\dagg_0t}\bigg(1+\lambda^2\int_0^t\int_0^{t_1} \CL^\dagg_I(t_1)\CL^\dagg_I(t_2)dt_2dt_1[ \vrho\otimes  \vsigma_B]\bigg)\bigg] \\
    &\stackrel{2}{\approx} \tr_{B'}\bigg[\e^{ \CL^\dagg_0t}\bigg(1+\lambda^2\int_0^t\int_0^{t_1} \CL^\dagg_I(t_1)\CL^\dagg_I(t_2)dt_2dt_1[ \vrho\otimes  \vsigma_{B'}]\bigg)\bigg] \\
    &\stackrel{2.5}{\approx} \tr_{B'}\bigg[\e^{ \CL^\dagg_0t}\bigg(1+ \lambda^2\int_0^{t}\e^{{\CL^\dagg_0}t_2}\left( \int_{t_2}^t \CL^\dagg_I(t_1-t_2)\CL^\dagg_Idt_1  \right)\e^{-\CL^\dagg_0t_2} dt_2\bigg)[ \vrho\otimes  \vsigma_{B'}]\bigg] \\
    &\stackrel{3}{\approx} \tr_{B'}\bigg[\e^{ \CL^\dagg_0t}\bigg(1+ \lambda^2\int_0^{t}\e^{{\CL^\dagg_0}t_2}\left( \int_{0}^{\infty} \CL^\dagg_I(t')\CL^\dagg_Idt_1  \right)\e^{-\CL^\dagg_0t_2} dt_2\bigg)[ \vrho\otimes  \vsigma_{B'}]\bigg] \\
    &\stackrel{3.5}{\approx} \e^{ \CL^\dagg_S t+\lambda^2 K^\dagg t}[ \vrho] \\
    &\stackrel{4}{\approx}  \e^{ \CL^\dagg_S t+\lambda^2 K^\dagg_{sec} t}[ \vrho] \\
    &\stackrel{5}{\approx}  \e^{ \CL^\dagg_S t+\lambda^2 \CD' t}[ \vrho] =: \e^{\CL^\dagg t}.
\end{align*}
The steps are, in order, 
\begin{itemize}
    \item (1)The interaction picture (Lemma~\ref{lem:half-integral}),
    \item (2)From a finite to an infinite bath ($B\rightarrow B'$), 
    \item (2.5)Rearrangement~\eqref{eq:rearrange_integral_order} ,
    \item (3)A finite-to-infinite time integral (Proposition~\ref{prop:int_infinite}),
    \item (3.5)Completing the interaction picture~\eqref{eq:complete_Kbar},
    \item (4)The secular approximation (Lemma~\ref{lem:true_davies_secular}).
    \item (5)Modifying the generator (Lemma~\ref{lem:modifying_D})
\end{itemize}
Let us briefly highlight the intuition for these steps. The idea of weak-coupling, without taking limit $\lambda\rightarrow \infty$, is really a leading-order approximation in the interaction picture (Step 1); the weak dissipative term $\lambda^2 K^\dagg$ in the presence of rapid rotation $\CL_S= \ri[\bar{\vH_S},\cdot]$ loses coherence. This is the secular approximation (Step 4), with an error depending on the time $t$ and coherence width $\bmu_0$ (and the rounding precision $\bnu_0$ and integer $m$); we modify the generator so that it preserves trace and satisfies exact detailed balance (Step 5).
See below for the accounting for each. Note that (1) the interaction picture and (4) the secular approximation is only valid for short times, so we have to additionally do a telescoping sum.  
The total error combines to

\begin{align}
    \lnormp{ (\CT_{t/\ell})^\ell - \e^{\bCL^\dagg t}}{1-1} & \le \lnormp{ (\CT_{t/\ell})^\ell - \e^{\ri\CL^\dagg_S+\lambda^2 K^\dagg t}}{1-1}+\lnormp{ \e^{\ri\CL^\dagg_S+\lambda^2 K^\dagg t}- \e^{\ri\CL^\dagg_S+\lambda^2 K^\dagg_{sec} t}}{1-1}+  \lnormp{  \e^{\ri\CL^\dagg_S+\lambda^2 K_{sec}^{\dagg} t} - \e^{\CL^\dagg t}  }{1-1}\\
    &\le \CO \Bigg[   \frac{\labs{a}^2\tau^2}{\Delta^2_{B}\ell}+ \labs{a}\tau \frac{t}{\ell} \L( (\frac{t}{\ell}+\beta)\frac{ \bu_{max}\labs{a}}{n_B}+ \exp( -\frac{(\bu_{max}-\beta \Delta^2_{B}/2 )^2}{2\Delta^2_{B}}) \R)+  \frac{\labs{a}\tau  \ell }{\Delta_{B}t} \notag \\
    & \hspace{3cm} +  \L( \frac{\tau^2\labs{a}^{2}\log(m)^2}{\Delta_B^2\sqrt{m}\bnu_0t}\R)^{2/3}+\tau\sqrt{\frac{m^3 \bnu_0}{\Delta_{B}}(\beta^2+\frac{1}{\Delta_B^2})} \Bigg].
\end{align}
 The integer $\ell$ counts the number of bath refreshes, $\bnu_0$ is the energy resolution of the final Gibbs state, there are $\labs{a}$ interacting terms, and $\Delta_{B}$ is the bath width
 \begin{align}
     \labs{\braket{\vB'^{a\dagg}(t)\vB'^{a}}_{B'}} = \CO( \e^{-\Delta^2_{B}t^2/2}).
 \end{align}
 The terms in the second and third lines, in order, are (1), (2), (3), (4), (5). 

For a total accuracy $\epsilon$ in trace distance, it suffices to choose
\begin{align*}
    \ell&=\theta\L( \frac{\labs{a}^2\tau^2}{\epsilon \Delta_{B}^2}\R)\\
    \bmu_0 &:= m\bnu_0 = \CO\L(\frac{\epsilon^{11}\Delta_B^9}{\labs{a}^8\tau^{10}(\beta^2+\frac{1}{\Delta^2_B})} \R)\\
    t&= \Omega\L( \frac{\labs{a}\tau \ell}{\epsilon \Delta_{B}} + \frac{\labs{a}^4\tau^4}{\bmu_0\epsilon^4\Delta_B^4}\R)\\
    n_B&= \tilde{\theta}\L( \labs{a}^2\tau \frac{t}{\ell} (\frac{t}{\ell}+\beta)\frac{ \beta \Delta^2_{B}+\Delta_B}{\epsilon} \R).  
\end{align*}
Note that in calculating time $t$ and bath size $n_B$, we chose the implicit variables
\begin{align}
    m &= \tilde{\theta}\L(\frac{\labs{a}^4\tau^4}{\epsilon^5\Delta^4_B} \R)\\
    \bnu_0 & = \frac{\bmu_0}{m}\\
    \bu_{max} &= \tilde{\theta} \L( \beta \Delta^2_{B}+\Delta_B\R)\\
    \lambda^2 & = \frac{\tau}{t}
\end{align}
and $\tilde{\theta}$ supresses a poly-logarithmic dependence (due to the $\log(m)$ for $m$ and the Gaussian decay for $\bu_{max}$) on all other parameters $\sqrt{\log(\cdots)}$. 
\end{proof}
Here are the error bounds for (1)-(5). The (1) interaction picture, (4) the secular approximation, and (5) the modification require more work.
\subsubsection{The interaction picture}

We start with the interaction picture
\begin{align}
    \tr_B\left[\e^{ (\CL^{\dagg} _0+\lambda \CL^{\dagg} _I)t}[ \vrho\otimes  \vsigma_B]\right] &= \tr_B\bigg[\e^{ \CL^{\dagg} _0t}\bigg(1+\lambda\int_0^t \CL^{\dagg} _I(t_1)dt_1+ \lambda^2\int_0^t\int_0^{t_1} \CL^{\dagg} _I(t_1)\CL^{\dagg} _I(t_2)dt_2dt_1+\cdots\notag\\ &\hspace{3cm}+\lambda^{m}\int_0^t\cdots\int_0^{t_{m-1}}\CL^{\dagg} _I(t_1)\cdots\CL^{\dagg} _I(t_m)dt_m\cdots dt_1+\cdots\bigg)[ \vrho\otimes  \vsigma_B]\bigg].
\end{align}
We keep only the second-order term, as the first-order term vanishes and the higher-order terms are subleading in $\tau:=\lambda^2t$.
\begin{fact}\label{fact:odd_orders_vanish} If $\tr_B[\vB^a \vsigma_B]=0, \forall a$, then the odd orders vanish, in particular the first order
\begin{align*}
\tr_B\left[\e^{ \CL^{\dagg} _0t}\int_0^t \CL^{\dagg} _I(t_1)dt_1[ \vrho\otimes  \vsigma_B] \right] =0.
\end{align*}
\end{fact}
\begin{lem}[Davies~\cite{davies74}]\label{lem:half-integral}
\begin{align*}
\lambda^m\lnormp{\tr_B\left[\e^{ \CL^{\dagg} _0t} \int_0^t\cdots\int_0^{t_{m-1}} \CL^{\dagg} _I(t_1)\cdots\CL^{\dagg} _I(t_m)dt_m\cdots dt_1 [ \vrho\otimes  \vsigma_B] \right]}{1}\le  \frac{(c(t)\tau)^{m/2}}{(m/2)!}
\end{align*}
where
\begin{align*}
    c(t) :=4\int_0^t\sum_{a,b}\norm{\vA^{a}}\norm{\vA^{b}}\labs{\braket{\vB^{a\dagg}(t)\vB^{b}}} dt.
\end{align*}
\end{lem}
\begin{proof}
Using quasi-freeness, triangle, and Holder's inequality, we obtain a sum over pairings
\begin{align}
    &\lnormp{\tr_B\left[\e^{ \CL^{\dagg} _0t} \int_0^t\cdots\int_0^{t_{m-1}} \CL^{\dagg} _I(t_1)\cdots\CL^{\dagg} _I(t_m)dt_m\cdots dt_1 [ \vrho\otimes  \vsigma_B] \right]}{1}\\
    &\le \sum_{(\ell_i,r_i)} \int_0^t\cdots\int_0^{t_{m-1}} \prod_i^{m/2}\sum_{a_i,b_i}2^{m}\norm{\vA^{a_i}}\norm{\vA^{b_i}}\labs{\braket{\vB^{a_i\dagg}(t_{\ell_i})\vB^{b_i}(t_{r_i})}} dt_m\cdots dt_1\\
    &= \frac{1}{2^{m/2}(m/2)!} \sum_{\pi \in S_{m}} \int_0^t\cdots\int_0^{t_{m-1}} \prod_i^{m/2}\sum_{a_i,b_i}2^{m}\norm{\vA^{a_i}}\norm{\vA^{b_i}}\labs{\braket{\vB^{a_i\dagg}(t_{\pi(2i)})\vB^{b_i}(t_{\pi(2i-1)})}} dt_m\cdots dt_1\\
    &= \frac{1}{2^{m/2}(m/2)!}\int_0^t\cdots\int_0^{t} \prod_i^{m/2}\sum_{a_i,b_i}2^{m}\norm{\vA^{a_i}}\norm{\vA^{b_i}} \labs{\braket{\vB^{a_i\dagg}(t_{2i})\vB^{b_i}(t_{2i-1})}} dt_m\cdots dt_1\\
    &\le \frac{2^m}{2^{m/2}(m/2)!}  \left( t \cdot 2\int_0^t\sum_{a_i,b_i}\norm{\vA^{a_i}}\norm{\vA^{b_i}}\labs{\braket{\vB^{a_i \dagg}(t)\vB^{b_i}}}dt\right)^{m/2}.
\end{align}
The first equality over-counts the number of pairings by permutation, normalized by the multiplicity. Note that the correlator in absolute value is even in time. The second equality combines the permuted time variables with the time-ordered integral. In the last inequality, the integral factorizes, and the factor of 2 comes from a crude estimate. This is the advertised result.
\end{proof}
The above lemma is at the heart of Davies' derivation of the weak-coupling limit. Crucially, it absorbs half of the t-integral and enables the change of variable $\tau=\lambda^2t$. However, the higher-order terms behave badly at larger values of $\tau$, and Davies managed to absorb an extra $t^\epsilon$ into the correlator. To suppress the higher-order terms at late times, Davies requires the coupling strength $\lambda\rightarrow 0$ to vanish very rapidly. Here, since we are interested in finite values of $\lambda$, we must avoid that by keeping $\tau$ small. This is why we allow ourselves to refresh the bath. We have not attempted to obtain better control over the higher-order terms, and see Section~\ref{sec:comment_optimal} for further discussion on potential improvements.  
 
We move on to massaging the (leading) second-order $\CO(\lambda^2)$ term. 
\subsubsection{From a finite to an infinite bath}
Next, we replace the finite bath with an infinite bath.

\begin{prop}[Errors from a finite bath]\label{prop:finite_bath}
\begin{align*}
\lambda^2\labs{\tr_{B}\left[\e^{ \CL^{\dagg} _0t}\int_0^t\int_0^{t_1} \CL^{\dagg} _I(t_1)\CL^{\dagg} _I(t_2)dt_2dt_1\cdot[ \vrho\otimes  \vsigma_B] \right]- \tr_{B'}\left[\e^{ \CL^{\dagg} _0t}\int_0^t\int_0^{t_1} \CL^{\dagg} _I(t_1)\CL^{\dagg} _I(t_2)dt_2dt_1\cdot \vrho\otimes \vsigma_{B'} \right] }\notag\\
\le \CO\L( \tau t \sum_a \norm{\vA^a}^2  \L( \frac{t \bu_{max}\labs{a}}{n_B}+ \exp( -\frac{(\bu_{max}-\beta \Delta^2_{B}/2)^2}{2\Delta^2_{B}})\R)\R).
\end{align*}

\end{prop} 
\begin{proof}
Expand the commutator
\begin{align}
    &\tr_{B}\left[\e^{ \CL^{\dagg} _0t}\int_0^t\int_0^{t_1} \CL^{\dagg}_I(t_1)\CL^{\dagg}_I(t_2)dt_2dt_1\cdot[ \vrho\otimes  \vsigma_B] \right]=\notag\\
&\sum_{a,b}\int_0^t\int_0^{t_1} \bigg( [\vA^b(t_1),\vA^a(t_2) \vrho]\cdot \braket{\vB^{b}(t_1)\vB^{a}(t_2)}_{\vsigma_B}-[\vA^b(t_1), \vrho \vA^a(t_2)]\cdot \braket{\vB^{a}(t_2)\vB^{b}(t_1)}_{\vsigma_B} \bigg)dt_2dt_1,\notag
\end{align}
And recall the finite system size error (Proposition~\ref{prop:correlator}) to obtain the advertised result.
\end{proof}

A great convenience due to an infinite bath is replacing finite time integral to infinite. Rearranging the integration order
\begin{align}
    \lambda^2\int_0^t\int_0^{t_1} {\CL^{\dagg}}_I(t_1)\CL^{\dagg}_I(t_2)dt_2dt_1 =\lambda^2\int_0^{t}\e^{{\CL^{\dagg}_0}t_2}\left( \int_{t_2}^t \CL^{\dagg}_I(t_1-t_2)\CL^{\dagg}_Idt_1  \right)\e^{-\CL^{\dagg}_0t_2} dt_2,\label{eq:rearrange_integral_order}
\end{align}
we will soon extend the inner integral to infinite time.

\subsubsection{Infinite time limit}
We now extend the inner integral till infinite time. 
\begin{prop}[Integral till infinite time.]\label{prop:int_infinite}

\begin{align*}
\lambda^2\bigg| \tr_{B'}\left[\e^{ \CL^{\dagg}_0t}    \int_0^{t}\e^{\CL^{\dagg}_0t_2}\left( \int_{0}^{t-t_2} \CL^{\dagg}_I(t')\CL^{\dagg}_Idt'  \right)\e^{-\CL^{\dagg}_0t_2} dt_2 \right] &- \tr_{B'}\left[\e^{\CL^{\dagg}_0t}     \int_0^{t}\e^{\CL^{\dagg}_0t_2}\left( \int_{0}^\infty \CL^{\dagg}_I(t')\CL^{\dagg}_Idt'  \right)\e^{-\CL^{\dagg}_0t_2} dt_2 \right] \bigg| \\
&\le \lambda^2\sum_{a,b}4\norm{\vA^{a}}\norm{\vA^{b}} \int^t_0 \left(\int_{t_2}^\infty\labs{\braket{\vB^{a}(t')\vB^{b}}}dt' \right)dt_2 \\
&\le \CO( \frac{\labs{a}\tau}{\Delta_{B}^2 t}  ).
\end{align*}
\end{prop}
\begin{proof}
Expand the commutator and change variable $t-t_2\rightarrow t_2$ to simplify the integral. Note the above estimate does not depend on $\bar{\vH}_S$, which is rounded. 
\end{proof}

Now, the above allows us to compress the dependence on bath into a super-operator $K^\dagg$
\begin{align}
    \tr_{B'}\left[\e^{ (\CL^{\dagg}_0+\lambda \CL^{\dagg}_I)t}[ \vrho\otimes  \vsigma_{B'}]\right] &\approx     \e^{ \CL^{\dagg}_S t}\bigg(1+\lambda^2\int_0^t K^{\dagg}(s)ds\bigg)[ \vrho]
\end{align}
with the notation $K^{\dagg}(s):=\e^{ -\CL^{\dagg}_S s}K^{\dagg}\e^{\CL^{\dagg}_S s}$, and 
\begin{align}
    K^{\dagg} [ \vrho] &:=-\int^\infty_0 \tr_{B'}\L[\e^{\ri(\bvH_S+\vH_{B'})s'}\vH_{I}\e^{-\ri(\bvH_S+\vH_{B'})s'}, [\vH_{I} , \vrho\otimes  \vsigma_{B'} ]\R] ds'.
\end{align}
There are four terms in $K^{\dagg}$ from taking two commutators, and we evaluate each of them as follows.
For example, 
\begin{align}
     - \int^\infty_0 \vA^a(s')\vA^b \vrho  \braket{ \vB^{a}(s') \vB^b }_{\vsigma_{B'}} ds' 
     = -  \sum_{\bomega} \Gamma_{ab}(\bomega)  \vA^a(\bomega)  \vA^b \vrho 
\end{align}
where we used the Fourier transforms 
\begin{align}
    \vA(s) &= \sum_{\bomega} \vA^a(\bomega)\e^{-\ri\bomega s},\\
    \braket{ \vB^{a}(s) \vB^b }_{\vsigma_{B'}} \indicator(s \ge 0 ) &=\frac{1}{2 \pi} \int_{-\infty}^\infty  \e^{-\ri \omega s} \Gamma_{ab}(\omega) d\omega.
\end{align}
Analogously, we calculate the remaining terms
\begin{align}
     - \int^\infty_0 \vrho \vA^b \vA^a(s') \braket{  \vB^b \vB^{a}(s')}_{\vsigma_{B'}} ds' 
     &= -  \sum_{\bomega} \Gamma_{ab}(-\bomega)    \vrho \vA^b\vA^a(\bomega), \\
      \int^\infty_0 \vA^a(s') \vrho \vA^b \braket{ \vB^b \vB^{a}(s') }_{\vsigma_{B'}} ds' 
     &=  \sum_{\bomega} \Gamma_{ab}(-\bomega) \vA^a(\bomega)   \vrho \vA^b, \\
      \int^\infty_0 \vA^b \vrho \vA^a(s')  \braket{ \vB^{a}(s')\vB^b  }_{\vsigma_{B'}} ds' 
     &=  \sum_{\bomega} \Gamma_{ab}(\bomega) \vA^b \vrho \vA^a(\bomega). 
\end{align}
We use that $\braket{ \vB^b \vB^{a}(s') }_{\vsigma_{B'}} = \braket{ \vB^{b}(-s') \vB^a  }_{\vsigma_{B'}} $. 
We also quickly complete the interaction picture
\begin{align}
\lnormp{\e^{ \CL^{\dagg}_S t}\bigg(1+\int_0^t K^{\dagg}(t_1)dt_1\bigg)[ \vrho] - \e^{ \CL^{\dagg}_S t+\lambda^2 K^{\dagg}t}[ \vrho] }{1-1} \le \CO\L(\normp{K^{\dagg}}{1-1}^2 \tau^2\R)  \label{eq:complete_Kbar}
\end{align}
where
\begin{align}
    \lnormp{ \bar{K}^{\dagg}}{1-1} \le 4\int_0^\infty\sum_{a,b}\norm{\vA^{a}}\norm{\vA^{b}}\labs{\braket{\vB'^{a}(t')\vB'^{b}}} dt' \le \CO(\frac{\labs{a}}{\Delta_{B}}).
\end{align}

\subsubsection{The secular approximation}\label{lem:true_davies_secular}
Next, we apply the secular approximation to the integral. We present it in terms of operators for clarity, but the identical proof applies to super-operators.
\begin{lem}[The secular approximation]
Consider an operator $\vK$ in the interaction picture 
\begin{align}
    \vK(s):=\e^{ -\ri \vL^{\dagg} s} \vK \e^{\ri \vL^{\dagg} s}
\end{align}
where the Hermitian operator $\vL$ has eigenvalues being integer multiples of the frequency $\bnu_0$. 
Then, for any unitarily invariant norm $\normp{\cdot}{*}$ and a tunable integer $m$
\begin{align}
     \lnormp{ \int_0^t \vK(s)ds - \int_0^t \vK_{sec}(s)ds}{*} \le \CO\L(   \frac{\normp{\vK}{*}}{\sqrt{m}\bnu_0}\R).
\end{align}
The secular-approximated $\vK_{sec}$ in the $\vL$ eigenbasis is defined by
 \begin{align}
     (\vK_{sec})_{ij} := \sum_{ij} \vK_{ij} \cdot \indicator\L(|n_i - n_j|\le m\R).
 \end{align}
\end{lem}
Intuitively, the time average $\int_0^t \vK(s)ds $ weakens the off-diagonal entries (in the $\vL$ eigenbasis) with a large eigenvalue difference. Dropping them incurs an error depending on the truncation value $m \bnu_0$.

\begin{proof}
In the eigenbasis of operator $\vL$ labeled by $i,j$, we calculate the error
\begin{align}
    \left(\int_0^t \vK(s)ds - \int_0^t \vK_{sec}(s)ds\right)_{ij} 
    &=K_{ij}\frac{\e^{-\ri( n_i-n_j)\bnu_0t}-1}{-\ri(n_i-n_j)\bnu_0} \indicator\big(|n_i - n_j| > m\big)\label{eq:m_sec_linear},
\end{align}
where we use the elementary integral
\begin{align}
    \int_0^t \e^{-\ri \omega t} dt= \frac{\e^{-\ri \omega t}-1}{-\ri \omega}. 
\end{align}

Now, consider the Fourier series for the function (with integer inputs)
\begin{align}
    f(n):=\frac{\e^{-\ri n\bnu_0 t}-1}{-\ri n\bnu_0} \indicator(\labs{n} > m) = \frac{1}{2\pi}\int_0^{2\pi} \tilde{f}(\theta) \e^{-\ri n\theta} d\theta.
\end{align}
Without explicitly evaluating $\tilde{f}(\theta)$, we can rewrite~\eqref{eq:m_sec_linear} as a linear combination of time-evolved operators
\begin{align}
    K_{ij}\frac{\e^{-\ri( n_i-n_j)\bnu_0t}-1}{-\ri(n_i-n_j)\bnu_0}  \indicator\big(|n_i - n_j| > m\big) &=  \frac{1}{2\pi} K_{ij} \int_0^{2\pi} \tilde{f}(\theta) \e^{-\ri ( n_i-n_j) \theta} d\theta \\
    &=  \frac{1}{2\pi} \int_0^{2\pi} \big( \vK(\theta/\bnu_0) \big)_{ij} \tilde{f}(\theta) d\theta,
\end{align}
where $\vK(\cdot)$ denotes the interaction picture.
Therefore, the $*$-norm can be bounded by 
\begin{align}
       \lnormp{ \frac{1}{2\pi} \int_0^{2\pi} \vK(\bnu_0/\theta) \tilde{f}(\theta) d\theta }{*} \le \normp{\vK}{*}\frac{1}{2\pi}\int_0^{2\pi} \labs{\tilde{f}(\theta)} d\theta &\le \normp{\vK}{*}\sqrt{\frac{1}{2\pi}\int_0^{2\pi} \labs{\tilde{f}(\theta)}^2 d\theta }\\
       &\le  \normp{\vK}{*} \sqrt{\sum_n \labs{f(n)}^2} \le \CO\L(   \frac{\normp{\vK}{*}}{\sqrt{m}\bnu_0}\R).
\end{align}
We use the triangle inequality and unitary invariance, Cauchy-Schwartz, that the Fourier series preserve inner-product, and the summation estimate $\sum_m^\infty 1/n^2 = \CO(1/m)$. 

\end{proof}
For our purposes, we use the $1-1$ super-operator norm (which is invariant under the interaction picture).
\begin{cor}[The secular approximation in $1-1$ norm]\label{cor:sec_1_1_norm}
\begin{align}
   \lnormp{ \int_0^t K^{\dagg}(t_1)dt_1 - \int_0^t K^{\dagg}_{sec}(t_1)dt_1 }{1-1} \le \CO\L(   \frac{\normp{K}{*}}{\sqrt{m}\bnu_0}\R).
\end{align}
\end{cor}
Since we will complete the interaction picture, we will also need to bound its norm. Thankfully, the secular-approximated superoperator is controlled by the original. 
\begin{prop}\label{prop:Ksec_K}
\begin{align}
    \normp{K^\dagg_{sec}}{1-1} = \CO\big(\normp{K^\dagg}{1-1} \log(m) \big).
\end{align}
\end{prop}
The original super-operator has a norm bounded by 
\begin{align}
    \lnormp{ K^{\dagg}}{1-1} \le c_K  =4\int_0^\infty\sum_{a,b}\norm{\vA^{a}}\norm{\vA^{b}}\labs{\braket{\vB'^{a}(t')\vB'^{b}}} dt' \le \CO(\frac{\labs{a}}{\Delta_{B}}).
\end{align}

\begin{proof}[Proof of Proposition~\ref{prop:Ksec_K}]
Again, consider the Fourier series for the function (with integer inputs)
\begin{align}
    f(n)= \indicator(\labs{n} \le m) := \frac{1}{2\pi}\int_0^{2\pi} \tilde{f}(\theta) \e^{-\ri n\theta} d\theta. 
\end{align}
We can evaluate explicitly
\begin{align}
    \labs{\tilde{f}(\theta)} &= \labs{\sum_{n=-m}^m \e^{\ri n \theta}} = \labs{\frac{\e^{\ri m \theta}-\e^{-\ri m \theta}}{1-\e^{\ri \theta}}} = \labs{ \frac{\sin(m\theta)}{\sin(\theta/2)} }.
\end{align}
Also,
\begin{align}
        \int_0^{\pi} \labs{\tilde{f}(\theta)} d \theta &= \int_0^{1/m} \labs{\tilde{f}(\theta)} d \theta + \int_{1/m}^{\pi} \labs{\tilde{f}(\theta)} d \theta \\
    &\le \int_0^{1/m} 2m d \theta+  \int_{1/m}^{\pi} \frac{2}{\theta} d \theta = \CO(1) + \CO(\log(m)),
\end{align}
where we use the bound $\labs{ \frac{\sin(m\theta)}{\sin(\theta/2)} } \le 2m$.
We can rewrite the secular approximation as a linear combination of time-evolved operators
\begin{align}
    K_{ij} \indicator(|n_i - n_j| \le m) &=  \frac{1}{2\pi} K_{ij} \int_0^{2\pi} \tilde{f}(\theta) \e^{-\ri ( n_i-n_j) \theta} d\theta \\
    &=  \frac{1}{2\pi} \int_0^{2\pi} \big( \vK(\theta/\bnu_0) \big)_{ij} \tilde{f}(\theta) d\theta,
\end{align}
where the notation $\vK(\cdot)$ denotes the interaction picture.
Therefore, the $*$-norm can be bounded by 
\begin{align}
       \lnormp{ \frac{1}{2\pi} \int_0^{2\pi} \vK(\bnu_0/\theta) \tilde{f}(\theta) d\theta }{*} \le \normp{\vK}{*}\frac{1}{2\pi}\int_0^{2\pi} \labs{\tilde{f}(\theta)} d\theta &= \CO\L( \normp{\vK}{*} \log(m) \R).
\end{align}

\end{proof}

Now, we can put the above estimates together. 
\begin{prop}[The secular approximation]\label{prop:secular_2nd_order} 
\begin{align*}
    \lnormp{ \e^{\CL^{\dagg}_St+\lambda^2Kt}-\e^{\CL^{\dagg}_St+\lambda^2 K_{sec} t} }{1-1} \le  \CO\L( \frac{\tau^2\labs{a}^{2}\log(m)^2}{\Delta^2_B\sqrt{m}\bnu_0t}\R)^{2/3} .
\end{align*}
\end{prop}
\begin{proof}
Let us first calculate the error for a short time $t_s$ with a second-order error in the  interacting picture.
\begin{align}
    \lnormp{ \e^{\CL^{\dagg}_St_s+\lambda^2Kt_s}-\e^{\CL^{\dagg}_St_s+\lambda^2 K^{\dagg}_{sec} t_s} }{1-1} &= \lambda^2\lnormp{ \e^{\CL^{\dagg}_St_s}\L( \int_0^t K^{\dagg}(s)ds - \int_0^t K^{\dagg}_{sec}(s)ds\R) }{1-1} + \CO\bigg( \lambda^4 t_s^2\big(\normp{K^\dagg}{1-1}^2+\normp{K^\dagg_{sec}}{1-1}^2\big) \bigg). \label{eq:2nd_order_int}
\end{align}
 Now, invoke the telescoping sum with $\ell_s$ segments and $t_s=t/\ell_s$
\begin{align}
    \lnormp{ \e^{\CL^{\dagg}_St+\lambda^2Kt}-\e^{\CL^{\dagg}_St+\lambda^2 K^{\dagg}_{sec} t} }{1-1} &\le (1+\CO(\epsilon) )\ell_s \lnormp{ \e^{\CL^{\dagg}_St_s+\lambda^2Kt_s}-\e^{\CL^{\dagg}_St_s+\lambda^2 K^{\dagg}_{sec} t_s} }{1-1}\\
    &\le \CO \L( \ell_s \frac{\lambda^2 \normp{K^\dagg}{1-1}}{\sqrt{m}\bnu_0}+ \frac{\lambda^4 t^2}{\ell_s^2}\L(\normp{K^\dagg}{1-1}^2+\normp{K^\dagg_{sec}}{1-1}^2\R) \R) \le \CO\L( \frac{\lambda^4 t\normp{K^\dagg}{1-1}^{2}\log(m)^2}{\sqrt{m}\bnu_0}\R)^{2/3}.
\end{align}
The $(1+\CO(\epsilon))$ factor in the first inequality is because the map $\e^{\CL^{\dagg}_St+\lambda^2Kt}$ is only approximately CPTP with an error bounded by the accumulated error in Theorem~\ref{thm:true_Davies}. We drop this subleading factor subsequently\footnote{There is no circular logic here, as it is at higher order.}. The second inequality is Corollary~\ref{cor:sec_1_1_norm} and~\eqref{eq:2nd_order_int}. The third inequality optimizes over $\ell_s$ and uses Proposition~\ref{prop:Ksec_K}.
\end{proof}

We evaluate the secular-approximated super-operator
\begin{align}
    K^{\dagg}_{sec}[\vrho] &= \sum_{ab} \sum_{\labs{\bomega - \bomega'}\le m \bnu_0 }\Gamma_{ab}(\bomega)  \vA^a(\bomega)  \vA^b(-\bomega') \vrho + \Gamma_{ab}(-\bomega)  \vrho \vA^b(-\bomega')\vA^a(\bomega) \notag\\
    &+ \Gamma_{ab}(-\bomega) \vA^a(\bomega)   \vrho \vA^b(-\bomega')+
    \Gamma_{ab}(\bomega) \vA^b(-\bomega') \vrho \vA^a(\bomega). 
\end{align}  
Relabel and regroup terms to obtain (in the Heisenberg picture) 
\begin{align}
    K_{sec}[\vX] &= \ri \sum_{ \labs{\bomega-\bomega'} \le m \bnu_0} \sum_{ab} S_{ab}(\bomega,\bomega')  [\vA^{a\dagg}(\bomega')\vA^b(\bomega),\vX]\notag\\
    & + \sum_{ \labs{\bomega-\bomega'} \le m \bnu_0 } \sum_{ab}\gamma_{ab}(\bomega,\bomega') \left( \vA^{a\dagg}(\bomega') \vX \vA^b(\bomega)-\frac{1}{2}\{\vA^{a\dagg}(\bomega')\vA^b(\bomega),\vX \} \right)
\end{align}
and we denote the second line by $\CD[\vX]$.
\subsubsection{Modifying for detailed balance and trace-preserving}

To simplify the detailed balance calculations, we further simplify the generator by introducing a nicer one
\begin{align}
    \CD^{'\dagg} &= \sum_{ \labs{\bomega-\bomega'} \le m \bnu_0 } \sum_{a}\gamma_{a}(\frac{\bomega+\bomega'}{2}) \left( \vA^a(\bomega) \vrho \vA^{a\dagg}(\bomega') -\frac{\e^{\beta \bomega_- }}{1+\e^{\beta \bomega_- }} \vA^{a\dagg}(\bomega')\vA^a(\bomega)\vrho - \frac{1}{1+\e^{\beta \bomega_- }} \vrho\vA^{a\dagg}(\bomega')\vA^a(\bomega) \right),
\end{align}
with $\bomega_- = \frac{\bomega -\bomega'}{2}$. This modification on the dissipative part makes it trace-preserving and detailed balanced (with respect to the rounded Gibbs state)\footnote{However, this does not generate a completely positive map. }. Recall Proposition~\ref{prop:D'_DB}, Corollary~\ref{cor:D'_fixed} we have
\begin{align}
    \CD^{'\dagg} = \sqrt{\bvsigma}\CD^{'} \L[\frac{1}{\sqrt{\bvsigma}}(\cdot)\frac{1}{\sqrt{\bvsigma}} \R] \sqrt{\bvsigma} \quad \text{and}\quad \CD[\bvsigma] = 0.
\end{align}

We control the 1-1 norm of the difference.
\begin{lem}[Modifying the dissipative part]\label{lem:modifying_D}
\begin{align}
    \lnormp{\CD^\dagg - \CD^{'\dagg}}{1-1} \le \CO\L(\sqrt{\frac{m^3 \bnu_0}{\Delta_{B}}(\beta^2+\frac{1}{\Delta_B^2})}\R).
\end{align}
\end{lem}
For calculations, recall the particular bath choice
\begin{align*}
    \braket{\vB^{a}(t')\vB^{a}}_{\vsigma_{B'}} &=\exp( \frac{\ri \beta \Delta^2_{Ba} t}{2}) \exp(-\frac{\Delta^2_{B}t^2}{2} ) +(h.c.),\\
    \gamma_{ab}(\omega) &= \int_{-\infty}^{\infty} \e^{\ri \omega s} \braket{\vB^{a}(s)\vB^{b}}_{\vsigma_{B'}} ds = \delta_{ab}\frac{1}{\sqrt{2\pi \Delta^2_{B}}}  \exp (\frac{-(\omega-\beta \Delta^2_{B}/2)^2}{2\Delta^2_{B}}).\notag
\end{align*}

\begin{proof}
We start with the subtraction
\begin{align*}
    \CD^{\dagg}-\CD^{'\dagg} &= 
    \sum_{ \labs{\bomega-\bomega'} \le m \bnu_0 } \sum_{a}\L( \gamma_{a}(\bomega,\bomega') - \gamma_{a}(\frac{\bomega+\bomega'}{2}) \R) \vA^{a\dagg}(\bomega') \vrho \vA^a(\bomega)\\
    &\hspace{2cm} - \L( \frac{1}{2}\gamma_{a}(\bomega,\bomega') - \frac{\e^{\beta \bomega_- }}{1+\e^{\beta \bomega_- }} \gamma_{a}(\frac{\bomega+\bomega'}{2}) \R) \vA^{a\dagg}(\bomega')\vA^a(\bomega)\vrho \\
    &\hspace{2cm} - \L( \frac{1}{2}\gamma_{a}(\bomega,\bomega') - \frac{1}{1+\e^{\beta \bomega_- }} \gamma_{a}(\frac{\bomega+\bomega'}{2}) \R) \vrho \vA^{a\dagg}(\bomega')\vA^a(\bomega).
\end{align*}
The proof idea is a 2-dimensional linear-combination-of-unitary argument. For any function on $\BZ^2$, the Fourier transform reads
\begin{align}
        f(n,n') =  \frac{1}{4\pi^2} \int_0^{2\pi}\int_0^{2\pi} \tilde{f}(\theta,\theta') \e^{-\ri n\theta} \e^{-\ri n'\theta'} d\theta d\theta'.
\end{align}
Then writing $\bomega = n \bnu_0, \bomega' = n' \bnu_0$, the term $\vA^{a\dagg}(\bomega') \vrho \vA^a(\bomega)$ can be expressed as
\begin{align}
    \sum_{n,n'} f(n,n')\vA^{a\dagg}(n'\bnu_0) \vrho \vA^a(n\bnu_0) &= \sum_{n,n'} \frac{1}{4\pi^2} \int_0^{2\pi}\int_0^{2\pi} \tilde{f}(\theta,\theta') \e^{-\ri n\theta} \e^{-\ri n'\theta'} \vA^{a\dagg}(n'\bnu_0) \vrho \vA^a(n\bnu_0) d\theta d\theta'\\
    & = \frac{1}{4\pi^2} \int_0^{2\pi}\int_0^{2\pi} \tilde{f}(\theta,\theta') \vA^{a}(\theta/\bnu_0) \vrho \vA^a(-\theta/\bnu_0) d\theta d\theta'.
\end{align}
In other words, to control the function $\tilde{f}(\theta,\theta')$, it suffices to study the function $f(n,n')$. 
Let us calculate 
\begin{align}
    \gamma_{a}(\bomega,\bomega') - \gamma_{a}(\frac{\bomega+\bomega'}{2}) &= \Gamma_a(\bomega)+ \Gamma^*_a(\bomega') - \Gamma_a(\frac{\bomega+\bomega'}{2})+ \Gamma^*_a(\frac{\bomega+\bomega'}{2})\\
    &=\int_{0}^{\infty} \e^{\ri \omega_+ s}(\e^{\ri \omega_- s}-1) \braket{\vB^{a}(s)\vB^{a}}_{\vsigma_{B'}} ds + \int_{0}^{\infty} \e^{-\ri \omega_+ s}(\e^{\ri \omega_- s}-1) \braket{\vB^{a}(s)\vB^{a}}_{\vsigma_{B'}} ds \label{eq:difference_gamma}
\end{align}
where we change variables
\begin{align}
\bomega_+:= \frac{\bomega+\bomega'}{2} \quad\text{and}\quad
\bomega_-:= \frac{\bomega-\bomega'}{2}.
\end{align}
We can now evaluate the sum
\begin{align}
    \sum_{n,n'} f(n,n')^2 = \sum_{n_-,n_+} f(n,n')^2 &= 
    \sum_{\bomega_-}\sum_{\bomega_+} \L(\gamma_{a}(\bomega,\bomega') - \gamma_{a}(\frac{\bomega+\bomega'}{2}) \R)^2 \\
    &= \sum_{\bomega_-} \CO\L(\frac{\bomega_-^2}{\bnu_0\Delta_{B}^3}\R) = \CO\L(\frac{m^3\bnu_0}{\Delta_{B}^3}\R).
\end{align}
We have used the Fourier series identity for each $\bomega_-$
\begin{align}
    \sum_{\bomega_+} \L(\int_0^{\infty} \e^{\ri \omega_+ s} f(s) ds \R)^2 = \frac{1}{2\pi \bnu_0}  \int_0^{2\pi/ \bnu_0} \big( f(s)+f(s+2\pi/\bnu_0) + \cdots ) \big)^2 ds
\end{align}
and estimate it for the Gaussian expression~\eqref{eq:difference_gamma}.\footnote{ Assuming $\Delta_B \gg \bnu_0$. } The last equality sums over frequencies $\labs{\bomega_-} \le m \bnu_0$. The above calculation feeds into the triangle inequality 
\begin{align}
    \lnormp{ \frac{1}{4\pi^2} \int_0^{2\pi}\int_0^{2\pi} \tilde{f}(\theta,\theta') \vA^{a}(\theta/\bnu_0) \vrho \vA^a(-\theta/\bnu_0) d\theta d\theta' }{1} &\le \labs{ \frac{1}{4\pi^2} \int_0^{2\pi}\int_0^{2\pi} \tilde{f}(\theta,\theta')}\\
    &\le \CO\L(\sqrt{\int_0^{2\pi}\int_0^{2\pi}  \tilde{f}(\theta,\theta')^2} \R) = \sqrt{ \sum_{n,n'} f(n,n')^2 }.
\end{align}
Similarly, for the other terms $\vrho \vA^{a\dagg}(\bomega')\vA^a(\bomega)$ and  $ \vA^{a\dagg}(\bomega')\vA^a(\bomega)\vrho$ we estimate
\begin{align}
        \sum_{\bomega_-}\sum_{\bomega_+} \L(\frac{1}{2}\gamma_{a}(\bomega,\bomega') - \frac{\e^{\beta \bomega_- }}{1+\e^{\beta \bomega_- }}\gamma_{a}(\frac{\bomega+\bomega'}{2}) \R)^2,\ \sum_{\bomega_-}\sum_{\bomega_+} \L(\frac{1}{2}\gamma_{a}(\bomega,\bomega') - \frac{1}{1+\e^{\beta \bomega_- }}\gamma_{a}(\frac{\bomega+\bomega'}{2}) \R)^2  
        = \CO\L(\frac{m^3 \bnu_0}{\Delta_{B}}(\beta^2+\frac{1}{\Delta_B^2})\R).
\end{align}
The extra factor $\beta^2$ comes from linearizing the Boltzmann factors $\e^{\beta \bomega_- }$. The above estimates combine to the advertised result.
\end{proof}

\section{Assumptions for the Eigenstate Thermalization Hypothesis and the Density of States}\label{sec:ETH}
Let us be frank that the calculation will be pretty messy and would be impossible without $\CO(\cdot),\Omega(\cdot)$ notations. 
Meanwhile, this reflects the flexibility of our arguments coming from the parameters in ETH and the density of states. This section collects the assumptions we make as an attempt to strike a balance between concreteness and flexibility. 

\subsection{ETH} 

ETH is a massive subject with various adaptations and loose ends (see, e.g.,~\cite{ETH_review_2016} for a review). Let us give a brief review of the original ETH prescription of Srednicki. 
\begin{hyp}[Srednicki's ETH~\cite{Srednicki_1999}]\label{hyp:old_ETH} 
Consider a Hamitonian $\vH$ and some ``suitable'' observable $\vA$. For energies $ \nu_i,\nu_j$ and their eigenstates $\ket{\nu_i},\ket{\nu_j}$ of $\vH$, 
\begin{align*}
A_{ij} = \bra{\nu_i}\vA\ket{\nu_j} &= O_{\vA}(\mu)\delta_{ij} +\frac{1}{\sqrt{dim(\vH) \cdot D(\mu)}} f_{\vA}(\mu,\omega)r_{ij}\label{eq:sredicki_ETH}
\end{align*}
where $\mu = (\nu_i + \nu_j)/2$, $\omega = \nu_i-\nu_j$, $D(\cdot)$ is the normalized density of states, $dim(\vH)=\tr[\vI]$ is the dimension of the Hilbert space, and $r_{ij}$ are random variables satisfying $\BE [r_{ij}]=0,\BE [r_{ij}^2]=1$. $O(\mu)$ and $f_{\vA}(\mu,\omega)$ are smooth functions of the energies.
\end{hyp}
The \textit{diagonal} prescriptions, as vague as it is, give a transparent explanation that time-averaging is equal to thermal averaging
\begin{align}
    \bar{\vA}:= \frac{1}{T}\int_0^{T} \bra{\phi} \vA(t) \ket{\phi} dt &= \frac{1}{T}\int_0^{T} \sum_{i,j} c^*_ic_j \e^{\ri (\nu_i-\nu_j) t} \vA_{ij} dt\\
    & \stackrel{T\rightarrow\infty}{=}  \sum_{i} O_{\vA}(\nu_i) \labs{c_{i}}^2 \approx \braket{\vA}_{\sigma_\beta}.
\end{align}
The last approximation holds if the energies of the test wave function $\ket{\phi}$ and the Gibbs state $\vsigma_{\beta}$ are both sufficiently narrow. This usually goes by the saying that ``even a (generic) eigenstate behaves like a Gibbs state'' and has been observed in numerics (see, e.g.,~\cite{ETH_review_2016}) and proven in suitable models~\cite{capel2021modified, KS_ETH_cluster_cor}. 

The \textit{off-diagonal} prescriptions explain that the typical fluctuations over time are suppressed by the dimension of the Hilbert space\footnote{This does not seem to be the best use of RMT.}
\begin{align}
    \frac{1}{T}\int_0^{T} (\bra{\phi}\vA(t)\ket{\phi}-\bar{\vA})^2 dt &= \frac{1}{T}\int_0^{T} \sum_{i,j} c^*_ic_j \e^{\ri (\nu_i-\nu_j) t} \sum_{k,\ell} c^*_kc_{\ell} \e^{\ri (\nu_k-\nu_{\ell}) t}\vA_{ij}\vA_{k\ell} dt\\
    & \stackrel{T\rightarrow\infty}{=}  \sum_{i,j} \labs{c_i}^2 \labs{c_j}^2\labs{\vA_{ij}}^2 = \CO(\e^{-\Omega(n)}).
\end{align}
The second equality holds under the non-degenerate condition that $\nu_i-\nu_j+\nu_k-\nu_{\ell}=0$ is only satisfied with $i=j, k=\ell$ or $i=\ell, k=j$. 
\subsection{The version we use}
However, there remain many unspecified parameters in the original ETH. Much of the subsequent works attempted to justify and specify. We incorporated some of them into our version of ETH.

\begin{hyp}[ETH for this work]\label{hyp:ETH_diagonal}
Consider a Hamiltonian $\vH$ and a set of $k$-local Hermitian operators $\vA^a$. We say they satisfy ETH if: 
\begin{itemize}
    \item It satisfies the original ETH for nearby energies $\labs{\nu_i-\nu_j} \le \Delta_{RMT}$.
    \item The random variables for entries $i,j$ and operators $\vA^a$ (\eqref{eq:sredicki_ETH}) are drawn from independent complex Gaussians $r_{ij}(a)\rightarrow g_{ij}(a)$ (up to the Hermitian constraint).
    \item The function depends only on the energy difference $\omega$
\begin{align*}
    f_{\vA}(\mu,\omega) = f_{\omega} \quad \text{where} \quad \min_{\omega \le \Delta_{RMT} } \labs{f_{\omega}} = \theta\L( \max_{\omega \le \Delta_{RMT} } \labs{f_{\omega}} \R). 
\end{align*}
\end{itemize}
\end{hyp}
Let us elaborate on the parameters we just filled in. We take the ``suitable'' observables to be $k$-local for constant $k$ comparing with the system size $k \ll n$ (i.e., few-body), but not necessarily spatially localized. This is observed numerically (see, e.g.,~\cite{ETH_review_2016}). 

As advertised, we explicitly introduce a random matrix theory prescription that will be a crucial proof ingredient. Of course, strictly speaking, there is no obvious source of randomness in any given Hamiltonian. Nevertheless, as commonly done in ETH literature~\cite{ETH_review_2016,dymarsky2018bound} (which goes back to Wigner~\cite{nuclear_random}), we \textit{model} the near diagonal band of different operators $\vA^a$ as if drawn from the complex Gaussian distribution and then \textit{fixed} for the rest of the calculation. 
Perhaps what is more unconventional is that we impose ETH for a \textit{set} of operators $\vA^a$ as \textit{i.i.d.} random, while traditionally ETH is discussed for individual operators. We crucially need the independence between random variables $\vA^a$ for all our concentration arguments. It would be an interesting follow-up to check whether independent Gaussians model the relationship between different operators $\vA^a$. 

The RMT prescription of ETH is very convenient but must be utilized with caution. Physically, the RMT time-scale $1/\Delta_{RMT}$ (Figure~\ref{fig:f(omega)}) is often thought of as the time when macroscopic phenomena ends~\cite{dymarsky2018bound,Dymarsky2018_isolated,2021_ETH_OTOC_Brenes}. For example, a self-consistency analytic argument~\cite{dymarsky2018bound} shows that the RMT time-scale $1/\Delta_{RMT}$ must be later than the time scale for transport (with supportive numerical evidences~\cite{2020_ETH_small_omega_Richter,wang2021eigenstate}). Quantitatively, the value of the RMT energy window $\Delta_{RMT}$ is believed to depend on the particular thermalization dynamics (e.g., transport) of the particular system~\cite{dymarsky2018bound}. Take 1d nearest-neighbor spin systems as an example, and plausible suggestions are such as~\cite{dymarsky2018bound, ETH_review_2016} 
\begin{align}
\Delta_{RMT} \stackrel{?}{=} \theta(\frac{1}{n^3}), \theta(\frac{1}{n^2}).
\end{align}
 Regardless, for the qualitative message of this work, we only need the $\poly(1/n)$ dependence. We will keep it as a tunable parameter $\Delta_{RMT}$ in our calculations.

Lastly, for simplicity, we assume the function depends only on the energy difference $\omega$, and the function $f_{\omega}$ is pretty flat for energies in the RMT window $\omega \le \Delta_{RMT}$.\footnote{ When $f$ depends on the energy $\mu$, the convergence rate may depend on temperature. Otherwise, we do not expect qualitative change to the convergence result.} While we will not need a particular form, let us mention that suggestions are such as~\cite[Sec. 4.3]{ETH_review_2016}
\begin{align}
    \labs{f_{\omega}}^2 =\begin{cases} 
    \theta(n^{1/d}) &\textrm{if  } \omega =\CO( \frac{1}{n^{1/d}} )\\
    (``\text{small}'') &\textrm{else },
    \end{cases}
\end{align}
and the polynomial dependence is all we need.

The above is all we will use to prove thermalization at finite times. For our calculations, we do not need the prescription for the diagonal entries (which has traditionally attracted more discussion); we do not care if the other entries outside the window $\Delta_{RMT}$ violate ETH or RMT.

\subsection{The density of states}\label{sec:DoS}
The density of states is a somewhat flexible parameter that goes into ETH and the random walk on the spectrum. 
\begin{hyp}[relative ratios of density at the scale $\Delta_{RMT}$]\label{assum:density_ratios}
(A)The spectrum can be treated as continuous that the density $D(\omega)$ is well-defined.
(B)For any $\labs{\nu-\nu'} \le \Delta_{RMT}$, 
\begin{align*}
\frac{D(\nu)}{D(\nu')} \le R(\Delta_{RMT},\bnu_0) =: R.    
\end{align*} 
\end{hyp}
The continuum assumptions (A) simplify the notation, and it is reasonable whenever the system size is large. The density ratio $R$ should be thought of as a global constant in the thermodynamics limit. These two assumptions are, in fact, implicit in ETH: when the spectrum exhibits finite-size effects or when the density ratio becomes large, then the ETH ansatz already needs corrections. In addition, the relative ratio of densities plays a role in concentration arguments for quantum expanders. There, if two frequencies $\nu_1, \nu_2$ have too disparate density of states, then will we need more interactions $\labs{a}$ to ensure a gap. What this practically means is that we would choose a small enough $\Delta_{RMT}$ so that $R$ remains $\CO(1)$.

However, the above assumption breaks down near the extreme eigenvalues of any finite-size system. To avoid mundane technical issues, we quickly patch another assumption to focus on the bulk of the spectrum by truncation (Figure~\ref{fig:DoS}). 
\begin{assum}[A truncated spectrum]\label{assum:truncated}
The spectrum has a sharp cut-off.
\end{assum}

To justify the sharp artificial cut-off, imagine starting with the original Hamiltonian $\vH'$ whose bulk of spectrum, containing all but exponentially rare states, satisfies ETH. Then, we truncate the Hamiltonian $\vH'\rightarrow \vH$ so that $\vH$ lies in the validity of ETH. So long as $dim(\vH)/dim(\vH')\rightarrow 1$ in the thermodynamic limit (saying a tiny portion $\sim \e^{-\Omega(n)}$ was truncated), the Gibb states (away from the ground state) should be very close. Further, during the run time of generator, a state initialized in the bulk should have very low chances to leave the bulk of the spectrum. In other words, $\vH$ can represent $\vH'$ for most purposes, but we do not further formalize this part.

Otherwise, if we do not truncate the Hamiltonian, there could be regimes that ETH does not hold, and/or the density of states becomes very low. A state initialized at such energies may be stuck forever. This is not an ETH phenomenon, nor the scope of this work.

We lastly assume a rather generic characterization of the density of \textit{the Gibbs state} (i.e. $D_{\vsigma}(\nu)\propto \e^{-\beta \nu} D(\nu)$) to give conductance estimates. This should not be confused with the density of states $D(\nu)$.
\begin{assum}[The Gibbs distribution with a characteristic scale $\Delta_{Gibbs}$]
We assume the Gibbs state $\vsigma$ satisfies the following:
\label{assum:delta_spec}
(A) There exists an interval $I_{bulk}=[E_{L}, E_R]$ that contains more than half of the weight such that for all $\nu\in I_{bulk}$,
\begin{align*}
     D_{\vsigma}(\nu) = \theta(\frac{1}{\Delta_{Gibbs}}).
\end{align*}
(B) The tail $I_R=[E_R, \infty]$ is decaying that for all $\nu\in I_{R}$,
\begin{align*}
    \int_{\nu}^\infty D_{\vsigma}(\nu') d\nu' =\CO( \frac{D_{\vsigma}(\nu)}{ \Delta_{Gibbs} } ),
\end{align*}
and similarly for $I_L$.
\end{assum}
These conditions guarantee a gap of $\poly(R,\e^{\beta \Delta_{RMT}}) \cdot \Omega(\Delta_{RMT}^2/\Delta_{Gibbs}^2)$ for the classical random walk on the spectrum with step size $\sim \Delta_{RMT}$ (Appendix~\ref{sec:conductance_1d}). Intuitively, the assumption makes sure the Gibb distribution has ``no bottlenecks''(Figure~\ref{fig:DoS}).

\begin{figure}[t]
    \centering
    \includegraphics[width=0.8\textwidth]{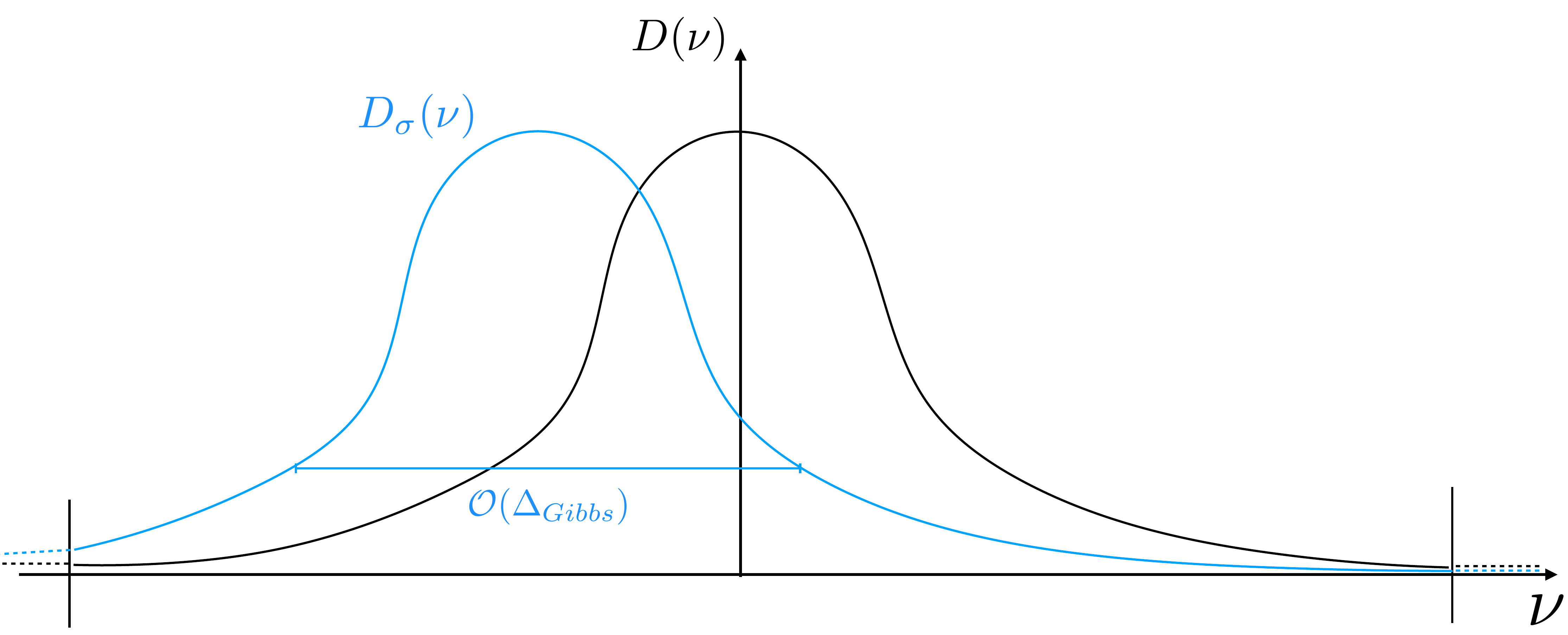}
    \caption{
    The density of states is truncated such that the continuum limit holds and the relative ratios of density are small, but at the same time not throwing away too many states. The Gibbs state, limited by the same truncation, may be centered at different energies. If we want the gap $\lambda_{RW}$ of random walk to be large enough, the Gibbs distribution better has no bottlenecks, such as a Gaussian. 
    }
    \label{fig:DoS}
\end{figure}
\subsection{Example: Gaussian density of state}
An example that satisfies the above assumptions is a Gaussian density of states, centered at zero, truncated at $\pm\frac{\norm{\vH}}{2}$. 
\begin{align}
    D(\nu) \propto \begin{cases}
    \displaystyle N_0\cdot  \exp (\frac{-\nu^2}{2\Delta^2_{spec}}) &\text{ if } \nu \in [- \frac{\norm{\vH}}{2}, \frac{\norm{\vH}}{2}]\\
    0 &\text{else}.
    \end{cases}
\end{align}
where $N_0$ is a normalization factor. 
Let us choose concretely that $\Delta_{Gibbs}=\sqrt{n}$, $\norm{\vH} = 2n$.
Then we can check the relative ratios (Assumption~\ref{assum:density_ratios}) 
\begin{align}
    \frac{D(\nu)}{D(\nu')}= \exp( \frac{\nu^2-\nu'^2}{2\Delta_{Gibbs}^2} ) \le  \exp( \frac{n\Delta_{RMT}+\Delta_{RMT}^2}{2n} ) \le \e^{\Delta_{RMT}/2}(1+o(1))=:R,
\end{align}
which is a constant for any reasonble values of $\Delta_{RMT}\le \CO(1)$.
Also, the spectrum remains continuous near the edge 
\begin{align}
    D(n) = \Omega( \frac{1}{\sqrt{n}} \e^{-n/2} ) \gg \frac{1}{2^n}. 
\end{align}

For Assumption~\ref{assum:delta_spec}, we can calculate its Gibbs state, which is a shifted Gaussian 
\begin{align}
    D_{\vsigma}(\nu) \propto \begin{cases}
    \displaystyle N_0'\cdot  \exp (\frac{-(\nu+\beta\Delta^2_{spec} )^2}{2\Delta^2_{spec}})\propto \exp (\frac{-(\nu+\beta n)^2}{2n}) &\text{ if } \nu \in [- n, n]\\
    0 &\text{else}.
    \end{cases}
\end{align}
where $N_0$ is a normalization factor. As long as $\beta < 1$, the tail rapidly decays as required.

The justification for presenting the Gaussian density example is due to central limit theorems for the spectrum of spatially local Hamiltonians~\cite[Lemma~8]{brandao2015equivalence}. It was shown that if the Gibbs state has decay of correlation, then the bulk of the Gibb state is Gaussian (similarly, set $\beta =0$ for the unbiased spectrum). At a finite system size $n$, however, the tail (outside of $1-\CO(\frac{1}{\sqrt{n}})$ portion of the weight) may deviate from a Gaussian density. This is partly why we make a somewhat general assumption on the Gibbs distribution (Assumption~\ref{assum:delta_spec}). 

\section{Main Result II: ETH Implies Convergence of the realistic generator}
\label{sec:true_Davies_conv}

\begin{thm}[Convergence of the realistic generator]\label{thm:ETH_true_Davies_convergence}
For a truncated Hamiltonian $\vH_S$ with a well-defined density of states (up to the truncation point), assume each $\vA^a$ for energy differences below $\Delta_{RMT}$ is i.i.d.sample from the ETH ansatz, and assume
\begin{align}
    R&=\CO(1) &\text{(small relative ratio of DoS)}\\
    \beta\Delta_{RMT} &= \CO(1) &\text{(small ETH window)}.
\end{align}
Then, with high probability (w.r.t to the randomness of ETH), running the realistic generator $\CL$ for effective time
\begin{align}
    \tau = \lambda^2 t =\tilde{\theta}\L( \frac{1}{\lambda_{RW}} \frac{\log(\epsilon)+ n+\beta \norm{\vH_S}}{ \labs{a}\int_{-\infty}^{ \infty} \gamma(\omega)\labs{ f_{\omega}}^2 d\omega}\R)
    \ \textrm{    ensures    } \lnormp{\e^{\CL^\dagg t}[\vrho]- \bvsigma}{1} \le \epsilon,
\end{align}
where
 \begin{align}
    \labs{a} &= \tilde{\theta}\L( \frac{1}{\lambda^2_{RW}}\R) &\text{(number of interactions)}\\
    \bmu_0 & = \tilde{\CO}\L( \frac{ \Delta^9_{RMT}\epsilon^{12}}{\beta^2 \labs{a}^8 \tau^{10} } \R)&\text{(coherence width)}\\
    \Delta_B & = \tilde{\theta}(\Delta_{RMT}) &\text{(bath width)}\\
    \bvsigma: &= \frac{ \e^{-\beta {\bvH}_S } }{\tr [ \e^{-\beta {\bvH}_S } ]} &\text{(rounded Gibbs)}.
 \end{align}
The number $\lambda_{RW}$ is the gap of a 1d classical random walk with the characteristic step size $\sim \Delta_{RMT}$ on the Gibbs distribution. The notation $\tilde{\theta},\tilde{\Omega}$ absorbs dependence on $\beta \Delta_{RMT}, R$ and poly-logarithmic dependences on any parameters.
\end{thm}
\begin{prop}
If the density of Gibbs state satisfies assumptions in Section~\ref{sec:DoS} (e.g., a Gaussian with variance $\Delta_{Gibbs}$), then $\lambda_{RW} :=\tilde{\Omega} \L(\frac{\Delta_{RMT}^2}{\Delta_{Gibbs}^2} \R)$.
\end{prop}
Intuitively, we need many interactions $\labs{a}$ to ensure concentration of the gap; we need the coherence width $\bmu_0$ to be small as it incurs error to the Gibbs state (through the Lamb-shift term); we need to choose the bath width $\Delta_{RMT}$ so that the bath function $\gamma(\omega)$ aligns with the ETH window $\Delta_{RMT}$.

The constraint $\beta \Delta_{RMT}$ might seem restrictive, but we can simply choose a small enough ETH window $\Delta'_{RMT} < \Delta_{RMT}$ to ensure $\beta \Delta'_{RMT}=\CO(1)$ (which only introduce polynomial dependence in temperature $1/\beta$ and not exponential!). In the detailed derivations, we have kept the sources of Boltzmann factors $\e^{\beta \Delta_{RMT}}$ for completeness.

\subsection{The gap of the dissipative part}
Our main estimate of convergence will be a gap calculation using the random matrix prescription of ETH. Since RMT may fail outside the window $\Delta_{RMT}$, we choose the bath width $\Delta_{B} $ such that very few transitions occur outside the ETH window
\begin{align}
 \labs{\gamma(\omega) - \gamma_{trun}(\omega)}  \ll 1 \quad\text{where}\quad \gamma_{trun}(\omega):=\indicator\big(\labs{\omega} < \Delta_{RMT}\big) \gamma( \omega ).
\end{align}

More precisely, we will calculate the gap of the truncated dissipator
\begin{align}
    \CD^{'}_{trun}[\vX] &= \sum_{ \labs{\bomega-\bomega'} \le \bmu_0 } \sum_{a}\gamma_{trun}(\frac{\bomega+\bomega'}{2}) \left(  \vA^{a\dagg}(\bomega')\vX \vA^a(\bomega) -\frac{\e^{\beta \bomega_- }}{1+\e^{\beta \bomega_- }} \vX\vA^{a\dagg}(\bomega')\vA^a(\bomega) - \frac{1}{1+\e^{\beta \bomega_- }} \vA^{a\dagg}(\bomega')\vA^a(\bomega)\vX \right)
\end{align}
with $\bomega_- = \frac{\bomega -\bomega'}{2}$. We need the truncation error $\lnormp{\CD'-\CD'_{trun}}{1-1}$ to be small. Using the Fourier transform argument, the truncation error will be some polynomial of other parameters multiplied by a Gaussian factor 
\begin{align}
    \lnormp{\CD'-\CD'_{trun}}{1-1} = \CO\L(\exp \L(\frac{-(\Delta_{RMT}-\beta \Delta^2_{B}/2)^2}{2\Delta^2_{B}}\R) \poly\L(n, \beta, \tau,\epsilon, \Delta_{RMT} \cdots\R)\R).\label{eq:trunc_D}
\end{align}
Therefore, choosing the bath width \footnote{ We have used that assumption $\beta \Delta_{RMT} = \CO(1)$ to simplify the expression.}
\begin{align}
\Delta_B = \tilde{\theta}( \Delta_{RMT} ) \quad\text{ensures}\quad \lnormp{\CD'-\CD'_{trun}}{1-1} =\CO(\epsilon)\label{eq:D_D_trun_epsilon}
\end{align}
where the notation $\tilde{\theta}$ suppresses logarithmic factors due to Gaussian decay in~\eqref{eq:trunc_D}.

\begin{lem}[Gap from concentration]\label{lem:gap_D}
Suppose the bath width $\Delta_{B}$ is such that function $\gamma(\omega)$ is negligible outside $\omega \in [ -\Delta_{RMT},\Delta_{RMT}]$. 
Using interactions 
\begin{align}
    \labs{a} = \tilde{\theta}\L( \frac{1}{\lambda^2_{RW}}\R)
\end{align}
ensures the second eigenvalue is at most 
\begin{align}
\lambda_{2}(\CD'_{trun}) = - \Omega\L( r \lambda_{RW}\R)
\end{align}
where 
\begin{align}
    r &= \tilde{\Omega}\L( \int_{-\infty}^{ \infty} \gamma(\omega)\labs{ f_{\omega}}^2 d\omega \R) \quad\text{and}\quad
    \lambda_{RW} :=\tilde{\Omega} \L(\frac{\Delta_{RMT}^2}{\Delta_{Gibbs}^2} \R) .
\end{align}
\end{lem}

We will estimate the gap by comparing it with the expected map
\begin{align}
    \CD^{'}_{trun} &= \BE \CD^{'}_{trun} + \L( \CD^{'}_{trun}-\BE \CD^{'}_{trun} \R).
\end{align}
\subsubsection{The expectation}
The expected map is importantly \textit{classical}, and we can control the gap via standard conductance estimates for Markov chains. We will work in the Schrodinger picture
\begin{align}
    \BE\CD^{'\dagg}[\vrho] &= \sum_{a,\bomega}\gamma_{a}(\bomega) \BE\left(   \vA^a(\bomega)\vrho \vA^{a\dagg}(\bomega)  -\frac{1}{2} \vA^{a\dagg}(\bomega)\vA^a(\bomega)\vrho - \frac{1}{2} \vrho \vA^{a\dagg}(\bomega)\vA^a(\bomega) \right).
\end{align}
We can decompose the inputs into the diagonal and the off-diagonal parts, and the expected map nicely preserves each.
For diagonal inputs, the expected map acts as a classical Markov chain generator with transition rates 
\begin{align}
    p(\nu_1\rightarrow \nu_2) &\equiv \bra{\nu_2}\BE\CD^{'\dagg}\big[\ket{\nu_1}\bra{\nu_1}\big] \ket{\nu_2} \\
    & = \gamma(\nu_1-\nu_2) \sum_a \BE[\labs{\vA^{a}_{\nu_2\nu_1}}^2] - \delta_{\nu_2\nu_1} \sum_{\nu_3} \gamma(\nu_1-\nu_3) \sum_a \BE[\labs{\vA^{a}_{\nu_3\nu_1}}^2].\label{eq:D'_expected_markov}
\end{align}
Indeed, this map is trace-preserving and satisfies detailed balance.
For off-diagonal inputs $\nu_1\ne \nu_2$, the map has negative eigenvalues 
\begin{align}
    \BE\CD^{'\dagg}\big[\ket{\nu_1}\bra{\nu_2}\big] = \frac{-1}{2} \ket{\nu_1}\bra{\nu_2}  \L( \sum_{\nu_3} \gamma(\nu_1-\nu_3) \sum_a \BE[\labs{\vA^{a}_{\nu_3\nu_1}}^2] + \sum_{\nu_3} \gamma(\nu_2-\nu_3) \sum_a \BE[\labs{\vA^{a}_{\nu_3\nu_2}}^2] \R).\label{eq:D'_off_diag}
\end{align}
It requires a conductance calculation (Appendix~\ref{sec:conductance_1d}) to estimate the gap of the Markov chain~\eqref{eq:D'_expected_markov}. Together with the eigenvalues for off-diagonal inputs~\eqref{eq:D'_off_diag}, we arrive at the following estimate on the second eigenvalue for the expected map.
\begin{prop}
The second eigenvalue of the expected map is at most
\begin{align}
\lambda_2(\BE\CD^{'\dagg}) \le - \Omega\L( r \lambda_{RW}\R).
\end{align}
\end{prop}

\subsubsection{Concentration around the expectation}\label{sec:true_Davies_concentration}
Finished with the expectation, we move on to obtain concentration around the expectation. We want to control fluctuations of the second eigenvalue $\lambda_2$ through the perturbation theory of eigenvalues. It will be crucial to work with the inner product
\begin{align}
\braket{\vO_1,\vO_2}_{\bvsigma}=\tr[\vO_1^\dagg \sqrt{ \bvsigma} \vO_2 \sqrt{ \bvsigma} ],    
\end{align}
under which $\BE[\CD^{'}_{trun}] ,\CD^{'}_{trun}$ are both self-adjoint.

Our goal is to bound
 \begin{align}
\lnormp{\delta\CD^{'}_{trun}}{\infty,\bvsigma}=  \lnormp{ \bvsigma^{-1/4} \delta\CD^{'}_{trun}\L[\bvsigma^{1/4} \cdot \bvsigma^{1/4}\R]\bvsigma^{-1/4}}{\infty} \quad\text{where}\quad  \delta\CD^{'}_{trun}:=\CD^{'}_{trun}-\BE \CD^{'}_{trun} .
\end{align}
Formally, to be more careful with norms, we should rewrite superoperator as a linear map on a doubled Hilbert space.
\begin{align}
    \CD^{'}_{trun} &\equiv \sum_{ \labs{\bomega-\bomega'} \le m \bnu_0 } \sum_{a}\gamma_{a}(\frac{\bomega+\bomega'}{2}) \left( \vA^{a}(-\bomega) \otimes \vA^{a*}(-\bomega') -\frac{1}{1+\e^{\beta \bomega_- }} \vA^{a}(-\bomega')\vA^a(\bomega)\otimes \vI - \frac{\e^{\beta \bomega_- }}{1+\e^{\beta \bomega_- }} \vI\otimes \vA^{a*}(-\bomega) \vA^{a*}(\bomega') \right)\\
    & =:  \CD^{'}_{\vA\otimes \vA} + \CD^{'}_{\vA\vA\otimes \vI}+ \CD^{'}_{\vI \otimes \vA\vA}.
\end{align}

After similarity transformation, we obtain
\begin{align}
    \bvsigma^{-1/4} \CD^{'}_{trun}[\bvsigma^{1/4}(\cdot)  \bvsigma^{1/4}]\bvsigma^{-1/4} & \equiv \sum_{ \labs{\bomega-\bomega'} \le m \bnu_0 } \gamma(\frac{\bomega+\bomega'}{2}) \sum_{a}\bigg( \frac{1}{\e^{\beta(\bomega+\bomega')/4 }}  \vA^a(-\bomega) \otimes \vA^{a*}(-\bomega')\notag\\ 
    &\hspace{4.5cm}-\frac{\e^{\beta \bomega_- /2}}{1+\e^{\beta \bomega_- }} \vA^{a}(-\bomega')\vA^a(\bomega)\otimes \vI\notag\\
    &\hspace{4.5cm}-\frac{\e^{\beta \bomega_- /2}}{1+\e^{\beta \bomega_- }} \vI\otimes \vA^{a*}(-\bomega) \vA^{a*}(\bomega') \bigg)\label{eq:similar_D}.
\end{align}
Partition the input density operator via projectors $\vP_{\bmu}$ as resolution of identity $\sum_{\bmu} \vP_{\bmu} = \vI$ where
\begin{align}
    \vP_{\bmu} &:= \sum_{\bnu = \bmu - m\bnu_0}^{\bmu+ (m-1)\bnu_0} \vP_{\bnu},\\
    \vrho &= \sum_{\bmu_1} \sum_{\bmu_2} \vP_{\bmu_1}\vrho\vP_{\bmu_2} =: \vrho_{\bmu_1\bmu_2}. 
\end{align}
Intuitively, for frequencies $\bmu$ being integer multiples of the coarse grained frequency (the coherence width) $\bmu_0=m\bnu_0$, each projector $\vP_{\bmu}$ is a bin that collects $m$ nearby projectors $\vP_{\bnu}$. 
Then we can bound $\delta\CD^{'}_{trun}$ at each sectors 
\begin{align}
     \vP_{\bmu_1'}\delta\CD^{'}_{trun}[\vP_{\bmu_1}(\cdot)\vP_{\bmu_2}  ] \vP_{\bmu_2'}
\end{align}
and then combine them to obtain a global bound on the spectral norm. Note that the projectors $\vP_{\bmu}$ commutes with the similarity transformation.
\begin{lem}\label{lem:true_davies_deltaD'}
The deviation, with high probability, is at most
\begin{align}
    \lnormp{    \delta\CD^{'}_{trun} }{\infty,\bvsigma} &= \tCO\L( \sqrt{\labs{a}}  \int_{-\infty}^{\infty} \labs{f_{\omega}}^2 \gamma(\omega) d \omega\R).
\end{align} 
\end{lem}

The concentration for Gaussian matrices will be useful and see Appendix~\ref{sec:Gaussian_calc} for their proofs.
\begin{fact}\label{fact:concentration_GoG} For rectangular matrices $\vG_i, \vG'_i$  with independent complex Gaussian entries, 
\begin{align*}
    \BE \lnormp{\sum_i a_i \vG_i\otimes  \vG'^*_i}{p}^p
    &\le (\BE \lnormp{\vG}{p}^p)^2\cdot (\sum_i a_i^2)^{p/2}.
\end{align*}
\end{fact}

\begin{fact}\label{fact:concentration_GG} For rectangular matrices $\vG_i, \vG'_i$  with independent complex Gaussian entries, 
\begin{align*}
    \BE \lnormp{\sum_i a_i \vG_i\vG'_i}{p}^p \le \BE \lnormp{\vG_i\vG'_i}{p}^p \cdot (\sum_i a^2_i)^{p/2} .
\end{align*}
\end{fact}

\begin{fact}\label{fact:rect_Gaussian_spectral_norm}   For rectangular matrices $\vG_{d_2d_1}$ with i.i.d. complex Gaussian entries, with variance $\BE[G_{ij}G_{ij}^{*}]=2$
\begin{align*}
\BE \lnormp{\vG}{p}^p \le   \min(d_1,d_2) \BE \norm{\vG}^p \le \min(d_1,d_2) \cdot \L( \sqrt{\max(d_1,d_2)}^p c_1^p+ (c_2\sqrt{p})^p \R),
\end{align*}
or
$
\Pr\L( \norm{\vG} \ge (\epsilon+c_3) \sqrt{\max(d_1,d_2)} \R) \le \exp\L(-\epsilon^2 c_4\min(d_1,d_2)^2\R)
$, for absolute constants $c_1, c_2,c_3, c_4$. 
\end{fact}

In other words, the power $p$ can be taken as large as the dimension, and they concentrate very sharply. In the end, we should have a union bound in mind over $\poly(n)$ choices of frequencies $\bnu_1,\bnu_2$, which is handled by the much stronger concentration. We will simplify the notation by dropping the tails and only saying ``with high probability''.

\begin{proof}
We will take a standard Gaussian decoupling approach~\cite{pisier2013rand_mat_operator}. For any super-operator $Q$ with zero mean $\BE[Q]=0$, we control the operator norm by the Schatten p-norm, 
 \begin{align}
      (\BE \lnormp{Q}{\infty,\bvsigma}^p )^{\frac{1}{p}}&\le  (\BE\lnormp{Q}{p,\bvsigma }^p)^{\frac{1}{p}}\\
      &\le (\BE \lnormp{Q - Q'}{p,\bvsigma}^p)^{\frac{1}{p}}.
 \end{align}
The second inequality uses convexity conditioned on super-operator $Q$ to introduce an identical copy $Q'$. 
Recall for any complex Gaussian matrix (entries with possibly different variances) can be decoupled with small constant overhead (from triangle inequality)
\begin{align}
    \vM \otimes\vM^* -\vM'\otimes\vM'^{*} &\equiv  \frac{(\vM+\vM')}{\sqrt{2}} \otimes \frac{(\vM+\vM')^*}{\sqrt{2}} - \frac{(\vM+\vM')}{\sqrt{2}} \otimes \frac{(\vM+\vM')^*}{\sqrt{2}}\\
    & = \vM\otimes \vM'^*+ \vM'\otimes \vM^*.
\end{align}

Let us evaluate the operator $\bvsigma^{-1/4} \vP_{\bmu_1'}\delta\CD^{'}_{trun}\L[\vP_{\bmu_1}\bvsigma^{1/4}(\cdot)\bvsigma^{1/4}\vP_{\bmu_2}  \R] \vP_{\bmu_2'}\bvsigma^{-1/4}$ term by term in~\eqref{eq:similar_D}. For the term $\vA^a(-\bomega) \otimes \vA^{a*}(-\bomega')$, 
\begin{align}
  &\lnormp{\vP_{\bmu'_1} \delta\CD'_{\vA\otimes \vA}[\vP_{\bmu_1}(\cdot)\vP_{\bmu_2}]\vP_{\bmu'_2} }{\infty,\bvsigma}\\ \le&\BE \lnormp{ \sum_{ \labs{\bomega-\bomega'} \le m \bnu_0 } \frac{\gamma(\frac{\bomega+\bomega'}{2})}{\e^{\beta(\bomega+\bomega')/4 }} \sum_{a}   \vP_{\bmu'_1}\vA^a(-\bomega)\vP_{\bmu_1} \otimes \vP^*_{\bmu'_2} \vA^{a*'}(-\bomega')\vP_{\bmu_2}^* }{p}^p \\
    &\le \indicator\L(\labs{\bmu'_1-\bmu_1 - \bmu'_2+\bmu_2}\le 2 \mu_0 \R)\BE \lnormp{ \sum_{ \bomega,\bomega'}  \frac{\gamma(\frac{\bomega+\bomega'}{2})}{\e^{\beta(\bomega+\bomega')/4 }} \sum_{a}   \vP_{\bmu'_1}\vA^a(-\bomega)\vP_{\bmu_1} \otimes \vP^*_{\bmu'_2} \vA^{a*'}(-\bomega')\vP_{\bmu_2}^*}{p}^p\\
     &\le \indicator\L(\labs{\bmu'_1-\bmu_1 - \bmu'_2 +\bmu_2}\le 2 \mu_0 \R)\BE \lnormp{ \sum_{a} \vG^a_{\bmu'_1\bmu_1} \otimes \vG^{a*'}_{\bmu'_2\bmu_2} }{p}^p,\\
     &\le \indicator\L(\labs{\bmu'_1-\bmu_1 - \bmu'_2 +\bmu_2}\le 2 \mu_0 \R) \sqrt{\labs{a}}\ \BE \lnorm{\vG_{\bmu'_1\bmu_1}}^p_p \BE \lnorm{\vG^{*'}_{\bmu'_2\bmu_2}}^p_p.
\end{align}
The first inequality uses that the expected p-norm is convex to throw in extra terms that are conditionally zero-mean. Second, we throw in extra Gaussians so that each entry has equal variance
\begin{align}
    \BE[\big( \vG_{\bmu'_1\bmu_1}\big)_{ij}^2] =  \max_{i\in \bmu'_1,j \in \bmu_1} ( \BE |A_{ij}|^2) \cdot \max_{\labs{\bmu'_1-\bmu_1-\omega } \le \CO(\bmu_0) }\frac{\gamma(\omega)}{\e^{\beta\omega/2 }},
\end{align}
and also sums over bohr frequencies $\bomega, \bomega'$. The last inequality uses Fact~\ref{fact:concentration_GoG} and 
\begin{align}
\BE \norm{\vG\otimes \vG'}^p_p &\le  \BE \norm{\vG}^p_p\cdot \BE \norm{\vG'}^p_p.
\end{align}
The p-norm estimates turn to concentration for power $p=\theta(\min(d_1,d_2))$. 
\begin{prop}With high probability,
\begin{align}
    \lnormp{\vP_{\bmu'_1} \delta\CD'_{\vA\otimes \vA}[\vP_{\bmu_1}(\cdot)\vP_{\bmu_2}]\vP_{\bmu'_2} }{\infty,\bvsigma}  &=\indicator\L(\labs{\bmu'_1-\bmu_1 - \bmu'_2 +\bmu_2}\le 2 \bmu_0 \R)\CO\L( \sqrt{c_{\bmu'_1,\bmu_1}c_{\bmu'_2,\bmu_2}}  \R),
\end{align}
where
\begin{align}
    c_{\bmu'_1,\bmu_1}:=\max\L(\tr[\vP_{\bmu'_1}],\tr[\vP_{\bmu_1}]\R) \sqrt{\labs{a}} \max_{i\in \bmu'_1,j \in \bmu_1} ( \BE |A_{ij}|^2) \cdot \max_{\labs{\bmu'_1-\bmu_1-\omega } \le \CO(\bmu_0) }\frac{\gamma(\omega)}{\e^{\beta\omega/2 }}.
\end{align}
\end{prop}

Similarly, we proceed with the other terms. Without changing the spectral norm, we drop the $(\cdot) \otimes \vP_{\bmu_2}\vP_{\bmu'_2}$ factor to simplify notations.
\begin{align}
    \lnormp{\vP_{\bmu'_1} \CD'_{\vA \vA}[\vP_{\bmu_1}(\cdot)\vP_{\bmu_2}]\vP_{\bmu'_2} }{\infty,\bvsigma} &\le\BE \lnormp{ \sum_{ \labs{\bomega-\bomega'} \le m \bnu_0 } \gamma(\frac{\bomega+\bomega'}{2})\frac{\e^{\beta \bomega_- /2}}{1+\e^{\beta \bomega_- }} \sum_{a} \vP_{\bmu'_1}\vA^{a}(-\bomega')\vA^a(\bomega)\vP_{\bmu_1} }{p}^p \\
    &\le \indicator\L(\labs{\bmu'_1-\bmu_1}\le 2 \mu_0 \R) \BE \lnormp{ \sum_{\bomega,\bomega'} \gamma(\frac{\bomega+\bomega'}{2})\frac{\e^{\beta \bomega_- /2}}{1+\e^{\beta \bomega_- }} \sum_{a} \vP_{\bmu'_1}\vA^{a}(-\bomega')\vA^a(\bomega)\vP_{\bmu_1} }{p}^p\\
    &= \indicator\L(\labs{\bmu'_1-\bmu_1}\le 2 \mu_0 \R) \BE \lnormp{ \sum_{\bmu_3} \sum_{\bomega,\bomega'} \gamma(\frac{\bomega+\bomega'}{2})\frac{\e^{\beta \bomega_- /2}}{1+\e^{\beta \bomega_- }} \sum_{a} \vP_{\bmu'_1}\vA^{a}(-\bomega')\vP_{\bmu_3}\vA^a(\bomega)\vP_{\bmu_1} }{p}^p\\
    & \le \indicator\L(\labs{\bmu'_1-\bmu_1}\le 2 \mu_0 \R)\ \BE \lnormp{ \sum_{\bmu_3} \sum_{a} \vG^a_{\bmu'_1\bmu_3} \vG^{a'}_{\bmu_3\bmu_1} }{p}^p\\
    &\le \indicator\L(\labs{\bmu'_1-\bmu_1}\le 2 \mu_0 \R) \L( \max(d_1,d_2)^{1/p} \cdot \sum_{\bmu_3} \sqrt{a} \BE\lnorm{\vG_{\bmu'_1\bmu_3}}\cdot \BE\lnorm{  \vG^{'}_{\bmu_3\bmu_1} } \R)^p.
\end{align}
The first inequality is again throwing in extra terms via convexity. The equality inserts projectors $\vP_{\bmu_3}$. The second inequality uses convexity to replace with Gaussians of equal variance
\begin{align}
    \BE[\big( \vG_{\bmu'_1\bmu_3}\big)_{ij}^2] =  \max_{i\in \bmu'_1,j \in \bmu_3} ( \BE |A_{ij}|^2) \cdot \frac{\e^{\theta(\beta\bmu_0) }}{1+\e^{\theta(\beta \bmu_0) }} \cdot \max_{\labs{\bmu'_1-\bmu_3-\omega } \le \CO(\bmu_0) } \gamma(\omega).
\end{align}
Lastly, we use Fact~\ref{fact:concentration_GG} and the following estimate to reduce to the spectral norm 
\begin{align}
\BE \norm{\vG\vG'}^p_p &\le  \max(d_1,d_2) \cdot \BE \norm{\vG}^p\cdot \BE \norm{\vG'}^p.
\end{align}
The p-norm estimates turn to concentration for $p=\theta(\min(d_1,d_2))$. Note that the dimensional factor diminishes
\begin{align}
    \max(d_1,d_2) ^{1/p} = \CO(1).
\end{align}
\begin{prop}With high probability,
\begin{align} 
   \lnormp{\vP_{\bmu'_1} \delta\CD'_{\vA \vA}[\vP_{\bmu_1}(\cdot)\vP_{\bmu_2}]\vP_{\bmu'_2} }{\infty,\bvsigma} &=\indicator\L(\labs{\bmu'_1-\bmu_1}\le 2 \mu_0 \R) \CO\L(\sum_{\bmu_3} \sqrt{c_{\bmu'_1,\bmu_3}c_{\bmu_1,\bmu_3}}  \R),
\end{align}
where
\begin{align}
    c_{\bmu'_1,\bmu_3}:=\max\L(\tr[\vP_{\bmu'_1}],\tr[\vP_{\bmu_3}]\R) \sqrt{\labs{a}} \max_{i\in \bmu'_1,j \in \bmu_3} ( \BE |A_{ij}|^2) \cdot \frac{\e^{\theta(\beta\bmu_0) }}{1+\e^{\theta(\beta \bmu_0) }} \cdot \max_{\labs{\bmu'_1-\bmu_3-\omega } \le \CO(\bmu_0) } \gamma(\omega).
\end{align}
\end{prop}

Now, we bound the global gap.
\begin{align}
    \lnormp{ \delta\CD^{'}_{\vA\otimes \vA}[\cdot] }{\infty,\bvsigma} &=\lnormp{    \sum_{\bmu'_1,\bmu_1, \bmu'_2 , \bmu_2} \vP_{\bmu_1'}\delta\CD^{'}_{\vA\otimes \vA}[\vP_{\bmu_1}(\cdot)\vP_{\bmu_2}  ] \vP_{\bmu_2'}}{\infty,\bvsigma} \\
    & \le \sum_{\bmu_3} \lnormp{ \sum_{\bmu'_1-\bmu_1 = \bmu_3, \bmu'_2 , \bmu_2} \vP_{\bmu_1'}\delta\CD^{'}_{\vA\otimes \vA}[\vP_{\bmu_1}(\cdot)\vP_{\bmu_2}  ] \vP_{\bmu_2'}}{\infty,\bvsigma} \\
    & \le \sum_{\bmu_3} \CO(1) \max_{\bmu'_1-\bmu_1 = \bmu_3, \bmu'_2 , \bmu_2}\L( \indicator\L(\labs{\bmu'_1-\bmu_1 - \bmu'_2 +\bmu_2}\le 2 \mu_0 \R) \sqrt{c_{\bmu'_1,\bmu_1}c_{\bmu'_2,\bmu_2}} \R)\\
    &= \CO\L( \sqrt{\labs{a}} R \int_{-\infty}^{\infty} \labs{f_{\omega}}^2 \frac{\gamma(\omega)}{\e^{\beta \omega/2}} d \omega\R). 
\end{align}
The first inequality pulls out the sum over differences $\bmu_3 =\bmu'_1-\bmu_1$ via the triangle inequality. The second inequality uses that each term in the summand collides with $\CO(1)$ others to the bound the spectral norm by the maximal\footnote{The reason we pull out the sum over $\bmu_3$ is to ensure each term collides with $\CO(1)$ others.}. The last estimate evaluates $c_{\bmu'_1,\bmu_1}$ and approximates the discrete sum by the integral. The factor of density ratio $R$ comes from $\max\L(\tr[\vP_{\bmu'_1}],\tr[\vP_{\bmu_1}]\R) \max_{i\in \bmu'_1,j \in \bmu_1} ( \BE |A_{ij}|^2)$. 
For the other terms,
\begin{align}
    \lnormp{ \delta\CD^{'}_{\vA \vA\otimes \vI}[\cdot] }{\infty,\bvsigma} &=\lnormp{    \sum_{\bmu'_1,\bmu_1} \vP_{\bmu_1'}\delta\CD^{'}_{\vA \vA\otimes \vI}[\vP_{\bmu_1}(\cdot) ] }{\infty,\bvsigma} \\
    & \le \max_{\bmu'_1,\bmu_1} \L( \indicator\L(\labs{\bmu'_1-\bmu_1}\le 2 \mu_0 \R) \CO(\sum_{\bmu_3} \sqrt{c_{\bmu'_1,\bmu_3}c_{\bmu_1,\bmu_3}}  ) \R)\\
    &= \CO\L( \e^{\CO(\beta \bmu_0)}\sqrt{\labs{a}} R \int_{-\infty}^{\infty} \labs{f_{\omega}}^2 \gamma(\omega) d \omega\R). 
\end{align}
And similarly for the term $\lnormp{ \delta\CD^{'}_{\vI \otimes \vA \vA}[\cdot] }{\infty,\bvsigma}$. Note that the factor $\e^{\CO(\beta \bmu_0)}\le \e^{\CO(\beta \Delta_{RMT})}$ is small. 

\end{proof}

\subsection{Effects of the Lamb-shift term}\label{sec:Lamb_error}
Now that we know the dissipative part $\CD'$ by itself converges to the Gibbs state, we move on to include the system Hamiltonian and the so-called Lamb-shift term $\CL_S+\lambda^2\CL_{LS}$. 
We want to show that they only slightly change the output state. 
Unfortunately, this operator has a large strength, so we cannot use triangle inequality in $1-1$ super-operator norm. More carefully, we need to utilize that the term $\CL_S+\lambda^2\CL_{LS}$ almost preserves the Gibbs state and is almost anti-Hermitian(in the $\bvsigma^{-1}$-weighted inner-product).

Denote the leading eigenspace projector of $\CD'$ by $P_0$ 
\begin{align}
P_0: =    \bvsigma \tr[ \cdot ] \quad \text{and}\quad P_1 := 1-P_0,
\end{align}
then the complement projector $P_1$ is exactly the eigenspace of the remaining eigenvectors \footnote{The eigenvectors are orthogonal because $\CD'$ is Hermitian w.r.t to $\bvsigma^{-1}$-weighted inner-product}. 
Accordingly, decompose the terms into blocks 
\begin{align}
    \CL_{S}^\dagg+ \lambda^2\CL_{LS}^\dagg &= P_1(\CL_{S}+ \CL_{LS})^\dagg P_1 +\lambda^2\big( P_0 \CL_{LS}^\dagg P_1+ P_1 \CL_{LS}^\dagg P_0 \big) \\
    &=:\CL_1^\dagg+ V^\dagg\\
    &=  \big(\CL_{1,A}^\dagg + \CL_{1,H}^\dagg \big)+ V^\dagg.
\end{align}
The first equality drops vanishing terms $P_0 \CL^\dagg_{S} = \CL_{S}^\dagg P_0=0$ and $P_0 \CL_{LS}^\dagg P_0=0$.\footnote{Since $\tr[ [\vH_{LS},\bvsigma] ]=0$.} In the third line, we further isolate the $\bvsigma^{-1}$-self-adjoint components. 
Recall, in the Schordinger picture, any super-operator $\CL$ decomposes into \footnote{To be more careful, the adjoint operators $(\cdot)^\dagg$ and $(\cdot)_{H}$ do not commute. We abuse notation to denote $\CL^\dagg_{1,H}=(\CL^\dagg_{1})_H$; we are working in the Schordinger picture so $\CL$ binds with $(\cdot)^\dagg$ first. }
\begin{align}
    \CL^\dagg &= \frac{1}{2}\left(\CL^\dagg+\sqrt{ \vsigma}\CL[\frac{1}{\sqrt{\vsigma}}\cdot \frac{1}{\sqrt{\vsigma}}]\sqrt{ \vsigma}\right) + \frac{1}{2}\left(\CL^\dagg-\sqrt{ \vsigma}\CL[\frac{1}{\sqrt{\vsigma}}\cdot \frac{1}{\sqrt{\vsigma}}]\sqrt{ \vsigma} \right)\\
    & =: (\CL^\dagg)_{H}+ (\CL^\dagg)_{A}.
\end{align}
We will see that the terms $\CL_S^\dagg +\CL_{LS}^\dagg$ is dominantly the anti-Hermitian component $\CL^\dagg_{1,A}$; the rest terms $\CL_{1,H}^\dagg$ and $V^\dagg$ are small. In other words, up to small error, the terms $\CL_S^\dagg +\CL_{LS}^\dagg$ preserves the inner-product as well as the subspace $P_1$. 

\begin{lem}[Self-adjoint components]\label{lem:LS_self-adjoint}
\begin{align}
    \lnormp{\L(P_1\CL_{LS}^\dagg P_1\R)_{H}}{1-1} &\le \CO\L(\sqrt{\frac{m^3\beta^2\bnu_0}{\Delta_{B}}}\R),\\
    \L(P_1\CL_{S}^\dagg P_1\R)_H &=0.
\end{align}
\end{lem}
\begin{proof} 
We first simplify
\begin{align}
    \L(P_1\CL_{S}^\dagg P_1\R)_{H} &= \big((1-P_0)\CL_{S}^\dagg (1-P_0)\big)_{H}\\
    & = (1-P_0)\big(\CL_{S}^\dagg \big)_{H}(1-P_0) = P_1\L(\CL_{S}^\dagg \R)_{H}P_1.
\end{align}
The second inequality uses that projector $P_0$ is self-adjoint 
\begin{align}
    (P_0)_H=P_0.
\end{align}
We then calculate
\begin{align}
    \big(\CL_{S}^\dagg \big)_{H}[\vrho] 
    & = \frac{1}{2}\L( - \ri [\bvH_{S}, \vrho] + \sqrt{ \bvsigma}\ri [\bvH_{S},\frac{1}{\sqrt{\bvsigma}}\vrho \frac{1}{\sqrt{\bvsigma}}]\sqrt{ \bvsigma} \R)\\
    & = -\ri \L( \bvH_{S} - \sqrt{\bvsigma}\bvH_{S}\frac{1}{\sqrt{\bvsigma}} \R) \vrho + \ri \vrho   \L( \bvH_{S} - \frac{1}{\sqrt{\bvsigma}} \bvH_{S} \sqrt{\bvsigma}\R)=0.
\end{align}
The last equality uses that the Hamiltonian $\bvH_S$ and the Gibbs state $\bvsigma$ can be simultaneously diagonalized.
Similarly, we proceed for the Lamb-shift term
\begin{align}
    \big(\CL_{LS}^\dagg \big)_{H}[\vrho] 
    & = -\ri \L( \vH_{LS} - \sqrt{\bvsigma}\vH_{LS}\frac{1}{\sqrt{\bvsigma}} \R) \vrho + \ri \vrho   \L( \vH_{LS} - \frac{1}{\sqrt{\bvsigma}} \vH_{LS} \sqrt{\bvsigma}\R).
\end{align}
We will have to get our hands dirty. Following the Fourier series argument in Lemma~\ref{lem:modifying_D}, we write down 
\begin{align}
    \vH_{LS} -\sqrt{\bvsigma} \vH_{LS} \frac{1}{\sqrt{\bvsigma}} = \sum_{ \labs{\bomega-\bomega'} \le m \bnu_0} \sum_{a} S_{a}(\bomega,\bomega')\L( 1-\e^{-\beta (\bomega'-\bomega)/2 }\R) \vA^{a\dagg}(\bomega')\vA^a(\bomega),
\end{align}
where 
\begin{align}
    S_{a}(\bomega,\bomega') &= \frac{1}{2\ri}\L( \Gamma_a(\bomega)-\Gamma^*_a(\bomega') \R)\\
    &=\int_{0}^{\infty} \e^{\ri \omega_+ s}\e^{\ri \omega_- s} \braket{\vB^{a}(s)\vB^{a}}_{\vsigma_{B'}} ds - \int_{0}^{\infty} \e^{-\ri \omega_+ s} \e^{\ri \omega_- s} \braket{\vB^{a}(s)\vB^{a}}_{\vsigma_{B'}} ds. \label{eq:S_omega_omega'}
\end{align}

We can now evaluate the sum
\begin{align}
    \sum_{n,n'} f(n,n')^2  &= 
    \sum_{\bomega_-}\sum_{\bomega_+} \L(S_{a}(\bomega,\bomega')(1 - \e^{\beta \bomega_-} )\R)^2 \\
    &= \sum_{\bomega_-} \CO(\frac{\beta^2\bomega_-^2}{\bnu_0\Delta_{B}}) = \CO(\frac{m^3\beta^2\bnu_0}{\Delta_{B}}).
\end{align}
We arrive at 
\begin{align}
    \lnormp{P_1\big(\CL_{LS}^\dagg \big)_{H}P_1[\vrho]}{1} &= \lnormp{\big(\CL_{LS}^\dagg \big)_{H}[\vrho-\bvsigma]}{1}\\
    &\le \L(\lnorm{\vH_{LS} - \sqrt{\bvsigma}\vH_{LS}\frac{1}{\sqrt{\bvsigma}}}+\lnorm{ \vH_{LS} - \frac{1}{\sqrt{\bvsigma}} \vH_{LS} \sqrt{\bvsigma}} \R) \normp{\vrho-\bvsigma}{1} = \CO\L(\sqrt{\frac{m^3\beta^2\bnu_0}{\Delta_{B}}}\R).
\end{align}
The last estimate uses that taking conjugate preserves the operator norm $\norm{\vO} = \norm{\vO^\dagg}$.
\end{proof}
\begin{prop}\label{prop:LS_V}
\begin{align}
    \lnormp{V^{\dagg}}{1-1} = \lnormp{P_0\CL_{LS}P_1 + P_1\CL_{LS}P_0}{1-1} \le \CO\L(\sqrt{\frac{m^3\beta^2\bnu_0}{\Delta_{B}}}\R).
\end{align}
\end{prop}
\begin{proof}
The proof is analogous to Lemma~\ref{lem:LS_self-adjoint}. 
We evaluate both terms using that the trace of commutator vanishes
\begin{align}
    P_0\CL^\dagg_{LS} P_1[\vrho] &= \ri \bvsigma \tr\L[ [\vH_{LS}, \vrho - \bvsigma]  \R] = 0 \quad\text{and}\quad \lnormp{P_1\CL^{\dagg}_{LS} P_0}{1-1} = \lnormp{\CL^\dagg_{LS} P_0}{1-1} = \lnormp{[\vH_{LS},\bvsigma]}{1}.
\end{align}
The second line turns projector $P_1$ into the identity $P_1+P_0=1$. Through the identical arguments as in Lemma~\ref{lem:LS_self-adjoint}, we express
\begin{align}
    [\vH_{LS}, \bvsigma ] &= \L(\vH_{LS} -\bvsigma \vH_{LS} \frac{1}{\bvsigma}\R) \bvsigma\\
    &= \sum_{ \labs{\bomega-\bomega'} \le m \bnu_0} \sum_{a} S_{a}(\bomega,\bomega')\L( 1-\e^{-\beta (\bomega'-\bomega) }\R) \vA^{a\dagg}(\bomega')\vA^a(\bomega) \bvsigma
\end{align}
and obtain
\begin{align}
    \lnormp{P_1\CL P_0}{1-1} \le \CO\L(\sqrt{\frac{m^3\beta^2\bnu_0}{\Delta_{B}}}\R).
\end{align}
This is the advertised result.

\end{proof}
We also need the fact that the generator only slightly grows the trace distance.
\begin{cor}[Trace contration and almost CPTP]\label{prop:LS_almost_CPTP}
\begin{align}
    \lnormp{ \e^{(\CL_S+ \CL_{LS} + \CD')^{\dagg}t }[\vX] }{1} \le \normp{\vX}{1}( 1 + \CO(\epsilon) ). 
\end{align}
\end{cor}
\begin{proof}
By Theorem~\ref{thm:true_Davies}, the superoperator is $\epsilon$-close to the CPTP map
\begin{align}
    \lnormp{\e^{(\CL_S+\lambda^2 \CL_{LS} + \lambda^2\CD')^{\dagg}t} - \CT_{t/\ell}^\ell}{1-1}\le \epsilon.
\end{align}
Any CPTP map contracts the trace distance 
\begin{align}
    \lnormp{ \CT_{t/\ell}^\ell[\vX] }{1} \le \lnormp{\vX}{1},
\end{align}
which concludes the proof.
\end{proof}

\subsection{Proof of Theorem~\ref{thm:ETH_true_Davies_convergence}}
We now show convergence of the map
\begin{align}
    \e^{\CL^\dagg t}[\vrho] \quad\text{where}\quad \CL =  \CL_S+\lambda^2( \CL_{LS} + \CD').
\end{align}
\begin{proof}
Expanding for small deviations $\CL_{1,H}$ and $V$,
\begin{align}
    \e^{\CL^\dagg t}[\vrho] &= \e^{ (\CL_1 + V + \lambda^2\CD^{'})^{\dagg}t } \\
    & = \e^{ ( \CL^\dagg_{1,A}+ \lambda^2\CD^{'\dagg}_{trun})t }+ \int_0^t \e^{ (\CL_1 + \lambda^2\CD^{'})^{\dagg} s } \ (\CD^{'}-\CD^{'}_{trun})^{\dagg}\ \e^{ (\CL_{1,A}^\dagg +\lambda^2\CD^{'\dagg}_{trun})s } ds\\
    &+ \int_0^t \e^{ (\CL_1 + \lambda^2\CD^{'})^{\dagg} s } \ \CL_{1,H}^{\dagg}\ \e^{ (\CL_{1,A}^\dagg +\lambda^2\CD^{'\dagg})s } ds + \int_0^t \e^{ (\CL_1 + \lambda^2\CD^{'}+V)^{\dagg} s }\ V^{\dagg}\ \e^{ (\CL_1 +\lambda^2\CD')^{\dagg}s } ds.
\end{align}
The second line expands the exponential twice via the following elementary identity.
\begin{fact}
\begin{align}
\e^{(\vA+\vB) t} = \int_0^t \e^{(\vA+\vB) t} \vB \e^{\vA t} ds.
\end{align}
\end{fact}
We begin by showing that the first term converges to the Gibbs state. Note that the truncated super-operator is detailed balanced and generates a trace-preserving map. Decompose the state into the leading eigenvector and the orthogonal
\begin{align}
    \vrho = P_0[\vrho]+ (1-P_0)[\vrho] = \bvsigma + (\vrho - \bvsigma). 
\end{align}
Then the gap ensures trace distance convergence: convert 1-norm to 2-norm by Cauchy Schwartz 
\begin{align}
\lnormp{\e^{ ( \CL_{1,A}^\dagg+ \lambda^2\CD^{'\dagg}_{trun})t }[\vrho -\bvsigma] }{1} \le \tr[ \bvsigma]\cdot 
    \lnormp{\e^{ ( \CL^\dagg_{1,A}+ \lambda^2\CD^{'\dagg}_{trun})t }[\vrho -\bvsigma]}{ \bvsigma^{-1},2} &\le \e^{\lambda_2\tau} \lnormp{ \vrho -\bvsigma }{ \bvsigma^{-1},2}\\ &\le \e^{\lambda_2\tau} \tr[( \vrho -\bvsigma)^2 \frac{1}{ \bvsigma}]   \le  2 \e^{\lambda_2\tau} \lnormp{\frac{1}{ \bvsigma}}{\infty}. 
\end{align}
The second inequality uses the gap. Importantly, the anti-self-adjoint super-operator $\CL^\dagg_{1,A}$ preserves the $\bvsigma^{-1},2$-norm and leave the subspace $P_1$ invariant. 
The third inequality is $\tr[\vA\vB\vA^\dagg \vB^\dagg] \le \tr[\vA\vA^\dagg \vB^\dagg \vB]$, and the fourth uses $\norm{( \vrho- \bvsigma)^2}_{1}\le 2$.

Next, we control the error incurred due to truncation $ (\CD^{'}-\CD^{'}_{trun})^{\dagg}$ and terms $V$ and $\CL_{1,H}$. By equation~\eqref{eq:D_D_trun_epsilon}, Lemma~\ref{lem:LS_self-adjoint}, Proposition~\ref{prop:LS_V}, and Proposition~\ref{prop:LS_almost_CPTP} 
\begin{align}
    &\lnormp{\int_0^t \e^{ (\CL_1 + \lambda^2\CD^{'})^{\dagg} s } \ (\CD^{'}-\CD^{'}_{trun})^{\dagg}\ \e^{ (\CL_{1,A}^\dagg +\lambda^2\CD^{'\dagg}_{trun})s } ds}{1-1} = \CO(\epsilon),\\
    &\lambda^2\lnormp{\int_0^t \e^{ (\CL_1 + \lambda^2\CD^{'})^{\dagg} s } \ \CL_{1,H}^{\dagg}\ \e^{ (\CL_{1,A}^\dagg +\lambda^2\CD^{'\dagg})s } ds}{1},\  \lambda^2\lnormp{\int_0^t \e^{ (\CL_1 + \lambda^2\CD^{'}+V)^{\dagg} s }\ V^{\dagg}\ \e^{ (\CL_1 +\lambda^2\CD')^{\dagg}s } ds}{1}  = \CO\L(\tau\sqrt{\frac{m^3\beta^2\bnu_0}{\Delta_{B}}}\R).
\end{align}
 We have dropped the multiplicative factor $1+\CO(\epsilon)$ (from the almost trace-preserving maps) as it is subleading. 
 Finally, recall the second eigenvalue (Lemma~\ref{lem:gap_D}) 
\begin{align}
\lambda_{2}(\CD') &= -\tilde{\Omega}\L( r\lambda_{RW} \R),\\
    \lambda_{RW} &:=\tilde{\Omega} \L(\frac{\Delta_{RMT}^2}{\Delta_{Gibbs}^2} \R),\\
    r &= \tOmega\L( \int_{-\infty}^{ \infty} \gamma(\omega)\labs{ f_{\omega}}^2 d\omega \R).
\end{align} 
It suffices to choose 
 \begin{align}
    \tau &= \theta\L( \frac{1}{\lambda_2} \L(\log(\epsilon)+ n+\beta \norm{\vH_S}\R) \R)\\
    \labs{a} &= \tilde{\theta}\L(  \frac{1}{\lambda^2_{RW}}\R)\\
     \bmu_0 & =\tilde{\CO}\L( \frac{\Delta_{RMT}\epsilon^2}{\tau^2 \beta^2 m^2}  \R) = \tilde{\CO}\L( \frac{ \Delta^9_{RMT}\epsilon^{12}}{\beta^2 \labs{a}^8 \tau^{10} } \R)
 \end{align}
 to ensure $\epsilon$-convergence in trace distance
 \begin{align}
      \lnormp{ \e^{\CL^\dagg t}[\vrho] - \bvsigma}{1} \le \epsilon.
 \end{align}
We use the estimate ${\bvsigma}^{-1} \le 2^n\cdot \e^{\beta \norm{\vH_S}}$ and the assumption for bath width $\Delta_{B} = \tilde{\theta}(\Delta_{RMT})$. This is the advertised result. 
\end{proof}

 

%

\appendix
\section{Davies' Generator of a rounded Hamiltonian} \label{sec:rounded_Davies}
In the previous sections, we derive and study the realistic generator $\CL$. It was quite complicated and not a Lindbladian as it did not generate a CP map. In this section, we present and justify a simpler generator $\bCL$ that captures much of the finite-time physics and has numerous nice properties. Define
 \begin{align}
      \bCL &:= \bCL_S+\lambda^2 (\bCL_{LS}+\bCD ),
 \end{align}
 where the Lamb-shift term and the dissipative term are 
 \begin{align}
    \bCL_{LS}&:=\ri[\vH_{LS}, \cdot ], \vH_{LS}:= \sum_{\bomega:\bnu_0|\bomega} \sum_{ab} S_{ab}(\bomega)  \vA^{a\dagg}(\bomega)\vA^b(\bomega)\\
    \bCD[\vX]&:=\sum_{\bomega:\bnu_0|\bomega }\CL_{\bomega} [\vX]= \sum_{\bomega} \sum_{ab}\gamma_{ab}(\bomega) \left( \vA^{a\dagg}(\bomega) \vX \vA^b(\bomega)-\frac{1}{2}\{\vA^{a\dagg}(\bomega)\vA^b(\bomega),\vX \} \right),\label{eq:rounded_D}
\end{align}
 and functions $S_{ab}(\bomega)$, $\gamma_{ab}(\bomega)$ are defined as in the original Davies' generator~\eqref{eq:S_ab},~\eqref{eq:gamma_ab}.
 Conceptually, it coincides with the \textit{original} Davies' generator for the (highly degenerated) Hamiltonian rounded at the resolution $\bnu_0$~\cite{wocjan2021szegedy}
\begin{align}
\bar{\vH}_S = \sum_{\bnu} \bnu \vP_{\bnu}.
\end{align}
The above rounded version $\bCL$ lies between the original Davies' generator $\CL_{WCL}$ (of the unrounded Hamiltonian $\vH_S$) and the true, complicated generator $\CL$~\eqref{eq:true_Davies}. Unlike the complicated generator $\CL$, the coherence width and rounded precision collapse into one energy scale $\bnu_0$. For any fixed rounding precision $\bnu_0$ that we assumed, decoherence occurs at long enough time $t$ 
\begin{align}
    t \gg \frac{1}{\bnu_0}.
\end{align}
Indeed, the blocks $\vA^{a\dagg}(\bomega)$ and $\vA^{a\dagg}(\bomega')$ for different energies $\labs{\bomega - \bomega'} \ge \bnu_0$  are ``incoherent'' since we do not have the cross term $\vA^{a\dagg}(\bomega) \vX \vA^b(\bomega')$. Still, unlike the original Davies generator $\CL_{WCL}$ (which distinguishes every eigenstate), the energies are collected at their rounded value $\bomega, \bnu$ and form a massive superposition within each $\bomega$.

Technically, it is nice being a generator of CPTP maps and satisfies exact detailed balance (see below); the rounding into integer multiples gives the convenient labeling of input by discrete energies $\vP_{\bnu'}\vX \vP_{\bnu}$, which simplifies the proof (postponed to Section~\ref{sec:simplify}). 
 The only caveat is that we can only implement this generator assuming we start with the rounded Hamiltonian $\vH_S$\footnote{ Alternatively, \cite{wocjan2021szegedy} uses quantum walk methods to implement the dissipative part $\bCD$~\eqref{eq:rounded_D}, assuming certain rounding guarantee for the Hamiltonian.}. This rounding assumption is not physical and should be thought of as a convenient toy model. 
\subsection{The fixed point}
This version of generator $\bCL$ is as nice as it can be: it is a Lindbladian and satisfies detailed balance for the rounded Gibbs state. This brings about many technical conveniences and sharper convergence guarantees.
\begin{fact}[Exact detailed balance for each $\bomega$]\label{fact:detailed_balance_bomega}
For the rounded Gibbs state $ \bvsigma\propto e^{-\beta \bar{\vH}}= \sum_{\bnu}\e^{-\beta \bnu}\vP_{\bnu}$,
\begin{align*}
    \sqrt{ \bvsigma}(\CL_{\bomega}+\CL_{-\bomega})[\vX]\sqrt{ \bvsigma}=(\CL_{\bomega}+\CL_{-\bomega})^\dagg[\sqrt{ \bvsigma}\vX\sqrt{ \bvsigma}] .
\end{align*}
\end{fact}
\begin{proof}
Observe 
\begin{align*}
    \sqrt{ \bvsigma}\vA^{a\dagg}(\bomega)=e^{-\beta \bar{\vH}/2}\vA^{a\dagg}(\bomega)  &= \e^{-\beta \bomega/2} \vA^{a\dagg}(\bomega)e^{-\beta \bar{\vH}/2}\\
    &=e^{-\beta \bomega/2} \vA^{a\dagg}(\bomega)\sqrt{ \bvsigma}.
\end{align*}
And similarly for conjugate $ \vA^{a}(\bomega)\sqrt{ \bvsigma}=e^{-\beta \bomega/2}\sqrt{ \bvsigma}\vA^{a}(\bomega) $.
Hence, for the first term $\CL_{\bomega,1}:= \sum_{ab}\gamma_{ab}(\bomega) \vA^{a\dagg}(\bomega) \vX \vA^b(\bomega)$,
\begin{align*}
    \sqrt{ \bvsigma}\CL_{\bomega,1}[\vX]\sqrt{ \bvsigma}& = \sum_{ab}\gamma_{ab}(\bomega) \sqrt{ \bvsigma}\vA^{a\dagg}(\bomega) \vX \vA^b(\bomega) \sqrt{ \bvsigma}\\
    & = \sum_{ab}\gamma_{ab}(\bomega)  \e^{-\beta \bomega} \vA^{a\dagg}(\bomega)\sqrt{ \bvsigma} \vX \sqrt{ \bvsigma} \vA^b(\bomega)\\
    & = \sum_{ab}\gamma_{ba}(-\bomega)   \vA^{a}(-\bomega)\sqrt{ \bvsigma} \vX \sqrt{ \bvsigma} \vA^{b\dagg}(-\bomega)\\    
    &=\CL_{-\bomega,1}^\dagg[\sqrt{ \bvsigma}\vX\sqrt{ \bvsigma}] .
\end{align*}
In the second equality we commute the Gibbs state $\sqrt{ \bvsigma}$ through $\vA$ and in the third we used the KMS condition~\eqref{eq:KMS} for function $\gamma_{ab}(\bomega)$.
For the second term $\CL_{\bomega,2}:= \sum_{ab}\gamma_{ab}(\bomega) \frac{-1}{2}\{\vA^{a\dagg}(\bomega)\vA^b(\bomega),\vX \}$
\begin{align*}
    \sqrt{ \bvsigma}\CL_{\bomega,2}[\vX]\sqrt{ \bvsigma}& = \sum_{ab}\gamma_{ab}(\bomega) \sqrt{ \bvsigma}\frac{-1}{2}\{\vA^{a\dagg}(\bomega)\vA^b(\bomega),\vX \} \sqrt{ \bvsigma}\\
    & = \sum_{ab}\gamma_{ab}(\bomega) \frac{-1}{2}\{\vA^{a\dagg}(\bomega)\vA^b(\bomega),\sqrt{ \bvsigma}\vX \sqrt{ \bvsigma}\}\\
    &=\CL_{\bomega,2}^\dagg[\sqrt{ \bvsigma}\vX\sqrt{ \bvsigma}] .
\end{align*}
In the last equality we used self-adjointness $\CL_{\bomega,2}=\CL_{\bomega,2}^\dagg$.
Finally, combine with the analogous calculation for $\CL_{-\bomega,1}$ and $\CL_{-\bomega,2}$ to  yield the advertised result.
\end{proof}

\subsection{Implementation assuming a rounded Hamiltonian }
\label{sec:implement_rounded_Davies}
 In this section, we show that the advertised generator $\bCL$ approximates iterations of the marginal joint-evolution 
 \begin{align}
     \CT(t)[\vrho]:=\tr_B\left[\e^{ (\bCL^\dagg_0+\lambda \CL^\dagg_I)t}[ \vrho\otimes  \vsigma_B]\right],
 \end{align}
where the Liouvilian $\bCL_0 = \ri[\bvH_S+\vH_{B},\cdot ]$ is rounded at a fixed precision $\bnu_0$.

\begin{thm}[Davies generator of a rounded Hamiltonian.]\label{thm:rounded_Davies}
 Assume there are $\labs{a}$ interaction terms in Lindbladian $\sum_a \lambda \vA^a\otimes \vB^a $, and the Hamiltonian $\bvH_S$ is rounded at resolution $\bnu_0$. With a quasi-free Fermionic bath, the rounded generator $\bCL$ for effective time $\tau=\lambda^2t$ can be implemented with accuracy $\epsilon$:
\begin{align}
    \forall \vrho, \ \lnormp{\CT(t/\ell)^\ell[\vrho] - \e^{ \bCL^\dagg t}[ \vrho]}{1} \le \epsilon, 
\end{align}
for bath function
\begin{align}
    \gamma_{ab}(\omega) = \delta_{ab}\frac{1}{\sqrt{2\pi \Delta^2_{B}}}  \exp \L(\frac{-(\omega-\beta \Delta^2_{B}/2)^2}{2\Delta^2_{B}}\R),
\end{align}
whenever 
\begin{align*}
    \ell&=\theta( \frac{\labs{a}^2\tau^2}{\epsilon \Delta_{B}^2})  &(\textrm{bath refreshes}) \\
    t&= \Omega\L(\frac{\labs{a}\tau\ell}{\epsilon}+ \frac{\labs{a}^2\tau^2}{\epsilon\Delta_{B}^2\bnu_0}\R) &(\textrm{total physical run-time})\\
    n_B&=  \tilde{\theta}\L( \frac{\labs{a}^2\tau t}{\epsilon\ell}(\frac{t}{\ell}+\beta)(\beta \Delta^2_{B}+\Delta_B).   \R) &(\textrm{size of bath})
\end{align*}
\end{thm}
Here we have a fixed rounding precision $\bnu_0$ while the time $t$ can be arbitrarily large. (In Theorem~\ref{thm:true_Davies}, we had two energy scales: the rounding precision $\bnu_0 \ll 1/t$ and the coherence width $\bmu_0 \gg 1/t$. )

\subsection{Proof for the Rounded Generator (Theorem~\ref{thm:rounded_Davies})}\label{sec:proof_rounded_Davies}

\begin{proof}[Proof of Theorem~\ref{thm:rounded_Davies}]
Most of the steps are identical to ~\ref{thm:true_Davies} with the only difference at the secular approximation. The assumption for a rounded Hamiltonian gives the much simpler form of the generator $\bCL$ with sharper estimates.\begin{align*}
    \CT_{t}:=\tr_B\left[\e^{ (\bCL^\dagg_0+\lambda \CL^\dagg_I)t}[ \vrho\otimes  \vsigma_B]\right]
    \approx \e^{ \bCL^\dagg_S t+\lambda^2 \bar{K}^\dagg t}[ \vrho] \approx  \e^{ \bCL^\dagg t}[ \vrho]. 
\end{align*}
The secular approximation depends on the resolution $\bnu_0$ of Hamiltonian $\bar{\vH_S}$ and the time $t$.  
The total error (recalling Section~\ref{sec:proof_true_Davies}) combines to
\begin{align}
    \lnormp{ (\CT_{t/\ell})^\ell - \e^{\bCL^\dagg t}}{1-1} &\le \lnormp{ (\CT_{t/\ell})^\ell - \e^{\ri\CL^\dagg_S+\lambda^2 \bar{K}^\dagg t}}{1-1}+ \lnormp{  \e^{\ri\CL^\dagg_S+\lambda^2 \bar{K}^{\dagg} t} - \e^{\bCL^\dagg t}  }{1-1}\\
    &\le \CO \bigg[   \frac{\labs{a}^2\tau^2}{\Delta^2_{B}\ell}+ \labs{a}\tau \frac{t}{\ell} \L( (\frac{t}{\ell}+\beta)\frac{ \bu_{max}\labs{a}}{n_B}+ \exp( -\frac{(\bu_{max}-\beta \Delta^2_{B}/2 )^2}{2\Delta^2_{B}}) \R)\notag\\
    &\hspace{3cm} + \frac{\labs{a}\tau  \ell }{\Delta_{B}t} +   \frac{\labs{a}^2}{\Delta_B^2}\frac{\tau^2 }{\bnu_0t}  \bigg]
\end{align}
 where $\ell$ is the number of refreshes, $\bnu_0$ is the energy resolution of the final Gibbs state, there are $\labs{a}$ interacting terms, and $\frac{1}{\Delta_{B}}$ is the characteristic time scale for bath.
 The last term is the simplified estimate for the secular approximation. (We do not need to modify the generator, as it is already a detailed-balanced Lindbladian.)
For any $\bnu_0$ and total accuracy $\epsilon$, it suffices to choose
\begin{align*}
    \ell&=\theta( \frac{\labs{a}^2\tau^2}{\epsilon \Delta_{B}^2})\\
    t&= \Omega\L(\frac{\labs{a}\tau\ell}{\epsilon}+ \frac{\labs{a}^2\tau^2}{\epsilon\Delta_{B}^2\bnu_0}\R)\\
    n_B&=  \tilde{\theta}\L( \frac{\labs{a}^2\tau t}{\epsilon\ell}(\frac{t}{\ell}+\beta)(\beta \Delta^2_{B}+\Delta_B)   \R).
\end{align*}
Note that in calculating $n_B$, we chose that
\begin{align}
    \bu_{max} = \tilde{\theta} ( \beta \Delta^2_{B}+\Delta_B) 
\end{align}
and $\tilde{\theta}$ supresses a root-logarithmic dependence(due to the Gaussian decay) on all other parameters $\sqrt{\log(\cdots)}$. 
\end{proof}

\subsubsection{The secular approximation}
Recall (in the Heisenberg picture) 
\begin{align}
    \bK[\vX] &= \sum_{ab} \sum_{\bomega, \bomega'} \Gamma_{ab}(\bomega)  \vX\vA^a(\bomega)  \vA^b(-\bomega')  + \Gamma_{ab}(-\bomega)  \vA^b(-\bomega')\vA^a(\bomega)\vX \notag\\
    &+\hspace{2cm} \Gamma_{ab}(-\bomega) \vA^b(-\bomega') \vX \vA^a(\bomega) +
    \Gamma_{ab}(\bomega)   \vA^a(\bomega)\vX \vA^b(-\bomega'),
\end{align} 
and we want to replace with its secular approximation 
\begin{align}
    \bK_{sec}[\vX] &:=  \lim_{T\rightarrow \infty} \frac{1}{2T}\int_{-T}^T \e^{\bCL_St'} \bar{K}\e^{-\bCL_St'} dt'\\
    &= \sum_{ab} \sum_{\bomega} \Gamma_{ab}(\bomega)  \vX\vA^a(\bomega)  \vA^b(-\bomega)  + \Gamma_{ab}(-\bomega)  \vA^b(-\bomega)\vA^a(\bomega)\vX \notag\\
    &\hspace{2cm}+ \Gamma_{ab}(-\bomega) \vA^b(-\bomega) \vX \vA^a(\bomega) +
    \Gamma_{ab}(\bomega)   \vA^a(\bomega)\vX \vA^b(-\bomega),\\
    &= \ri \sum_{ \bomega} \sum_{ab} S_{ab}(\bomega)  [\vA^{a\dagg}(\bomega)\vA^b(\bomega),\vX]\notag\\
    & \hspace{2cm}+ \sum_{ \bomega } \sum_{ab}\gamma_{ab}(\bomega) \left( \vA^{a\dagg}(\bomega) \vX \vA^b(\bomega)-\frac{1}{2}\{\vA^{a\dagg}(\bomega)\vA^b(\bomega),\vX \} \right) = \bCL_{LS}+\bCD.
\end{align}
This secular approximation "decoheres" the cross terms $\bomega \ne \bomega'$ (previously, we only remove terms that $\labs{\bomega -\bomega'} > \bmu_0 = m \bnu_0$. The rounded assumption allows us to obtain a simpler result with sharper estimates to all orders.) 

We first obtain a bound at a short time $t_s$.
\begin{lem}[The secular approximation]\label{lem:secular_rounded_H} 
For a rounded Hamiltonian $\bCL_S = \ri[ \bar{\vH}_S,\cdot ]$, let $\theta : = \tau_s \normp{\bar{K}^{'\dagg}}{1-1}$. Then for $ \theta/(\bnu_0t_s)\le \frac{1}{2}$, 
\begin{align*}
    \lnormp{ \e^{\bCL^{\dagg}_St_s+\lambda^2\bar{K}^{\dagg}t_s}-\e^{\bCL^{\dagg}_St_s+\lambda^2 \bar{K}^{\dagg}_{sec} t_s} }{1-1} \le  \CO(1)\cdot \frac{1}{\bnu_0t_s} \L( \theta+\theta^2 \e^{\theta} \R).
\end{align*}
\end{lem}
Note the exponential $\e^{\theta}$ essentially sets a time scale $t_s\sim(\lambda^2 \normp{\bar{K}^{'\dagg}}{1-1} )^{-1}$ after which the bound becomes vacuous. By telescoping, we obtain a bound for a longer time $t$.
\begin{cor}
\begin{align*}
    \lnormp{ \e^{\bCL^{\dagg}_St+\lambda^2\bar{K}^{\dagg}t}-\e^{\bCL^{\dagg}_St+\lambda^2 \bar{K}^{\dagg}_{sec} t} }{1-1} \le  \CO\L( \frac{\labs{a}^2}{\Delta_B^2}\frac{\tau^2 }{\bnu_0t} \R).
\end{align*}
\end{cor}
\begin{proof}
We use the short time bound for interval $t_s = t/\ell_s $, for an $\ell_s$ such that $ \theta =  \tau/\ell_s \cdot \normp{\bar{K}^{'\dagg}}{1-1} = \CO(1)$. 
Note that $\ell_s$ is different from the number of refreshes $\ell$. The total error accumulates to
\begin{align}
    \ell_s \cdot \frac{1}{\bnu_0t_s}  = \CO\L( \frac{ \normp{\bar{K}^{'\dagg}}{1-1}^2 \tau^2 }{\bnu_0t}\R) = \CO\L( \frac{\labs{a}^2}{\Delta_B^2}\frac{\tau^2 }{\bnu_0t} \R). 
\end{align}
\end{proof}
Now, we prove Lemma~\ref{lem:secular_rounded_H}.
\begin{proof}
The proof is again a Fourier series (or linear combination of unitaries) argument. 
We begin by rewriting in the eigenbasis of the super-operator $\CL_S$
\begin{align}
     (\vK_{sec})_{ij} := \sum_{ij} \vK_{ij} \cdot \indicator\L(n_i =n_j\R).
 \end{align}
(I) We have already worked out the linear term for general function $\indicator\L(\labs{n_i-n_j} > m\R)$ (Corollary~\ref{cor:sec_1_1_norm}). In this case, we obtain the analogous bound
\begin{align}
    \lnormp{\int_0^t \bar{K}^{'\dagg}(t_1)dt_1 - \int_0^t \bar{K}^{'\dagg}_{sec}(t_1)dt_1}{1-1} \le \normp{\bar{K}^{'\dagg}}{1-1} \sqrt{\sum_n \labs{f(n)}^2} \le \normp{\bar{K}^{'\dagg}}{1-1} \frac{4}{\bnu_0} \frac{\pi}{\sqrt{6}}.
\end{align}
The constant arises from the identity $\sum_1^\infty 1/n^2 = \pi^2/6$. 

(II) For the higher-order terms, we have to do more work. Consider a telescoping sum 
\begin{align}
    &\int_0^t\cdots\int_0^{t_{m-1}}\bar{K}^{'\dagg}(t_1)\cdots\bar{K}^{'\dagg}(t_m)dt_m\cdots dt_1 -  \int_0^t\cdots\int_0^{t_{m-1}}\bar{K}^{'\dagg}_{sec}\cdots\bar{K}^{'\dagg}_{sec} dt_m\cdots dt_1\notag\\
    &= \sum_{\ell=1}^{m} \int_0^t\cdots\int_0^{t_{m-1}}\bar{K}^{'\dagg}(t_1)\cdots ( \bar{K}^{'\dagg}(t_\ell)-\bar{K}^{'\dagg}_{sec} ) \bar{K}^{'\dagg}_{sec}\cdots dt_m\cdots dt_1.
\end{align}
Since $\bar{K}^{'\dagg}_{sec}$ are independent of time, we can evaluate the inner integrals

\begin{align}
     \left(\int_0^{t_{\ell-1}} \cdots \int_0^{t_{m-1}} (\bar{K}^{'\dagg}(t_\ell)-\bar{K}^{'\dagg}_{sec} ) dt_m\cdots dt_\ell \right)_{ij}&= \left(\int_0^{t_{\ell-1}} (\bar{K}^{'\dagg}(s)-\bar{K}^{'\dagg}_{sec} ) \frac{s^{m-\ell}}{(m-\ell)!} ds\right)_{ij}\\
     &=\delta_{n_i\ne n_j} \int_0^{t_{\ell-1}} \bar{K}^{'\dagg}_{ij} \e^{-\ri(n_i-n_j)\bnu_0 t} \frac{t^{m-\ell}}{(m-\ell)!} dt\\
     &= \frac{-\e^{-\ri(n_i-n_j)\bnu_0 t_{\ell-1}}}{\bnu_0^{m-\ell+1}} \sum_{q=0}^{m-\ell}  \frac{1}{(\ri(n_i-n_j))^{m-\ell-q+1}}\frac{(\bnu_0t_{\ell-1})^q  }{q!}  \\
     &=:\frac{1}{2\pi} \bar{K}^{'\dagg}_{ij} \int_0^{2\pi} \tilde{f}(\theta) \e^{\ri ( n_i-n_j) \theta} d\theta
\end{align}
where the third equality is the elementary integration by parts formula.
\begin{fact}
\begin{align}
\int_0^a \e^x \frac{x^k}{k!} dx = \e^{a }(-1)^k \sum_{q=0}^{k} \frac{(-a)^q  }{q!} .       
\end{align}
\end{fact}

Again, by a Fourier series argument,
\begin{align}
    \lnormp{\int_0^{t_{\ell-1}} \cdots \int_0^{t_{m-1}} (\bar{K}^{'\dagg}(t_j)-\bar{K}^{'\dagg}_{sec} ) dt_j\cdots dt_1}{1-1} &\le \normp{\bar{K}^{'\dagg}}{1-1}\L( \frac{1}{2\pi}\int_0^{2\pi} \labs{\tilde{f}(\theta)}^2 d\theta \R)^{1/2}\\
    &\le \normp{\bar{K}^{'\dagg}}{1-1} \frac{1}{\bnu_0^{m-\ell+1}} \cdot \sum_{q=0}^{\ell-m}  \sqrt{ \sum_n^\infty\frac{1}{n^{2(m-\ell-q+1)}}} \frac{(\bnu_0t_{\ell-1})^q  }{q!} \\
    &\le \normp{\bar{K}^{'\dagg}}{1-1} \frac{1}{\bnu_0^{m-\ell+1}} \frac{\pi}{\sqrt{6}} \cdot \sum_{q=0}^{\ell-m}  \frac{(\bnu_0t_{\ell-1})^q  }{q!}.
\end{align}
Completing the remaining integrals, we now have the m-th order bound
\begin{align}
    &\lnormp{\int_0^t\cdots\int_0^{t_{m-1}}\bar{K}^{'\dagg}(t_1)\cdots\bar{K}^{'\dagg}(t_m)dt_m\cdots dt_1 -  \int_0^t\cdots\int_0^{t_{m-1}}\bar{K}^{'\dagg}_{sec}\cdots\bar{K}^{'\dagg}_{sec} dt_m\cdots dt_1}{1-1}\\
    &\le \frac{\normp{\bar{K}^{'\dagg}}{1-1}^{m}}{\bnu_0^{m}} \frac{\pi}{\sqrt{6}} \sum_{\ell=1}^{m} \sum_{q=\ell-1}^{m-1}  \frac{(\bnu_0t)^q  }{q!} \notag\\
    &= \frac{\normp{\bar{K}^{'\dagg}}{1-1}^{m}}{\bnu_0^{m}} \frac{\pi}{\sqrt{6}} \sum_{q=0}^{m-1}  \frac{(q+1)(\bnu_0t)^q  }{q!}.\notag
\end{align}
Finally, rearrange
\begin{align}
    \lnormp{ \e^{\CL^{\dagg}_St+\lambda^2\bar{K}^{'\dagg}t}-\e^{\CL^{\dagg}_St+\lambda^2 \CD t} }{1-1} 
    &\le \frac{2\pi}{\sqrt{6}} \sum_{m=1}^\infty \frac{(\lambda^{2}\norm{\bar{K}^{'\dagg}}t)^m }{(\bnu_0t)^{m}}  \sum_{q=0}^{m-1}  \frac{(\bnu_0t)^q  }{(q-1)!}\\
    &= \frac{2\pi}{\sqrt{6}} \sum_{q=0}^{\infty}  \frac{(\bnu_0t)^q  }{(q-1)!} \sum_{m=q+1}^\infty \frac{(\lambda^{2}\normp{\bar{K}^{'\dagg}}{1-1}t)^m }{(\bnu_0t)^{m}}  \\ 
    &\le \frac{4\pi}{\bnu_0t\sqrt{6}} \sum_{q=0}^{\infty} \frac{(\lambda^{2}\normp{\bar{K}^{'\dagg}}{1-1}t)^{q+1} }{(q-1)!}  \\ 
    &\le \frac{4\pi}{\bnu_0t\sqrt{6}} \L( \lambda^{2}\normp{\bar{K}^{'\dagg}}{1-1}t+(\lambda^{2}\normp{\bar{K}^{'\dagg}}{1-1}t )^2\e^{\lambda^{2}\normp{\bar{K}^{'\dagg}}{1-1}t } \R).
\end{align}
The second inequality sums the geometric series using the assumption
\begin{align}
    \frac{\lambda^{2}\normp{\bar{K}^{'\dagg}}{1-1}t}{\bnu_0t}\le \frac{1}{2}, 
\end{align}
which concludes the proof.
\end{proof}

\section{Preliminaries for Proofs of Fast Convergence}
 To show fast convergence for the Davies' generator of a rounded Hamiltonian $\bCL$, we collect some preliminary tools and properties. These handy facts provide concrete grounds for better convergence guarantees and understanding for Lindbladian $\bCL$ (which the realistic generator $\CL$ does not enjoy) 
\subsection{Modified Log Sobolev Inequality}
To keep our discussion self-contained, we instantiate the minimal facts as well as references; this is not intended to be a complete introduction. Recall a (finite-dimensional, unital) von Neumann algebra $\CM$ is an algebra with involution and identity
\begin{align*}
    &\forall \vX,\vY  \in \CM,\ \vX+\vY  \in \CM\\
    &\forall \vX,\vY  \in \CM,\ \vX\vY  \in \CM\\
    &\forall \vX \in \CM,\  \vX^* \in \CM\\
    &\forall \lambda \in \BC,\ \lambda \vI \in \CM.
\end{align*}
For an subalgebra $\CN \subset \CM$, a \textit{conditional expectation} $E_\CN:B(\CH)\rightarrow \CN$ is a completely positive unital map that
\begin{align*}
    &\forall \va \in \CN, E_\CN[\va] =\va\\
    &\forall \va,\vb \in \CN, E_\CN[\va\vX \vb] = \va E[\vX]\vb,
\end{align*}
and $E^{\dagg}_{\CN}$ as the adjoint (w.r.t $\tr[\cdot]$) is a CPTP (channel), i.e.
\begin{align}
    \forall \vX,\vY \in \CB(\CH), \tr[\vY E^{\dagg}_{\CN}\vX]= \tr[E[\vY ]\vX].
\end{align}
Consider relative entropy for pair of states $\Supp( \vrho) \subset \Supp( \vsigma)$
\begin{align*}
    D( \vrho|| \vsigma) = \tr[ \vrho( \ln( \vrho)-\ln( \vsigma))].
\end{align*}
For Lindbladian $\CL: \CM \rightarrow \CM$, it is said to satisfy \textit{Modified Log Sobolev inequality with constant $\alpha$} if \begin{align}
    \alpha D( \vrho || E^{\dagg}[ \vrho]) \le -\frac{d}{dt} D(e^{\CL^{\dagg}t}[ \vrho] || E^{\dagg}[ \vrho]) = \tr[\CL^{\dagg}[ \vrho](\ln( \vrho)-\ln(E^{\dagg}[ \vrho]))],
\end{align}
where $E$ is the projection onto the fixed point algebra of $\CL$
\begin{align}
    \lim_{t\rightarrow \infty} \e^{\CL t} =:E.
\end{align}
And note that one can choose an arbitrary stationary (or invariant) state $\vec{\omega}, E^\dagger[\vec{\omega}]=\vec{\omega}$ so that \cite[Lemma~2]{capel2021modified}
\begin{align}
\tr[\CL^{\dagg}[ \vrho](\ln( \vrho)-\ln(E^{\dagg}[ \vrho]))]  =\tr[\CL^{\dagg}[ \vrho](\ln( \vrho)-\ln(\vec{\omega}))].  \label{eq:free_choice_invariant}
\end{align}
In other words, an MLSI gives handy bounds on convergence to the stationary state $ \vsigma=E^{\dagg}[ \vsigma]$, measured by relative entropy. It converts to trace distance 
\begin{align}
    \lnormp{\e^{\CL^{\dagg}t}[ \vrho]- \vsigma}{1} \le \e^{-\alpha(\CL)t}\cdot \sqrt{2\ln(\norm{ \vsigma^{-1}})}\label{eq:ln_to_trace}
\end{align}
through Pinsker's inequality $\lnormp{ \vrho- \vsigma}{1}^2\le 2D( \vrho|| \vsigma)$ and initial relative entropy estimate $D( \vrho|| \vsigma)\le \ln(\norm{ \vsigma^{-1}})$.

Often the total Lindbladian has constituents which we understand individually. A recent development called \textit{approximation tensorization} of relative entropy allows us to estimate the \textit{global} MLSI constant from the \textit{local} ones\footnote{We thank Cambyse Rouz\'e for related discussions about~\cite{gao2021complete}.}.

\begin{fact}[Approximate Tensorization {\cite{gao2021complete,laracuente2021quasifactorization}}]\label{fact:approx_tensor}
Consider subalgebras $\CN_1, \cdots, \CN_\ell \subset \CB(\CH)$ and the intersection $\CN:=\cap_i \CN_i$. For some conditional expectations $E_{\CN} :\CB(\CH)\rightarrow\CN$, $E_{i} :\CB(\CH)\rightarrow\CN_i$, consider some common invariant state $ \vsigma$ that $E_{\CN}^\dagger[ \vsigma]= \vsigma$, $E_{i}^\dagger[ \vsigma]= \vsigma$. 
Then 
\begin{align}
    D( \vrho || E^{\dagg}[ \vrho]) \le \frac{2k}{1-\epsilon^2(2\ln(2)-1)^{-1}} \frac{1}{m} \sum_{i=1}^{m} \sum_{s\in S_i} D( \vrho || E^\dagg_{s}[ \vrho])
\end{align}

whenever $k$ satisfies
\begin{align}
    (1-\epsilon) E_\CN \le_{cp}(\Phi^{\dagger})^k \le_{cp} (1+\epsilon)E_\CN,
\end{align}
where the inequality denotes completely positive order. $\Phi^{\dagg}$ can be an average of products of conditional expectations corresponding to any regrouping into disjoint subsets $\{1,\cdots, \ell\} = S_1 \perp\cdots \perp S_m$
\begin{align}
\Phi^{\dagg}:=\sum_{i=1}^m \frac{1}{2m}\L( \prod^{\rightarrow}_{s \in S_i } E_{s}+\prod^{\leftarrow}_{s \in S_i } E_{s} \R).
\end{align} 
The product $\prod^{\rightarrow}_{s \in S_i }$ can take arbitrary order as long as $\prod^{\leftarrow}_{s \in S_i }$ is reversed.

\end{fact}
The above is well known to imply global MLSI constants. Note that while we require the completely positive order for $\Phi^\dagg$, we do not need a \textit{complete} MLSI constant~\cite{gao2021complete}. All the MLSI constants discussed in this paper do not consider ancillae. 

\begin{cor}[Global MLSI from local] In the setting of Fact~\ref{fact:approx_tensor}, 
\begin{align*}
    \alpha(\CL) &\ge \frac{1-\epsilon^2(2\ln(2)-1)^{-1}}{2k} m \min_i \alpha_i(\CL_i).  
\end{align*}
\end{cor}
\begin{proof} Following the arguments in~\cite{gao2021complete},
\begin{align*}
   \sum_{i=1}^\ell D( \vrho || E^{\dagger}_\ri[ \vrho]) &\le \sum_{i=1}^{\ell} \frac{1}{\alpha_i} \tr[\CL^{\dagg}_\ri[ \vrho](\ln( \vrho)-\ln( \vsigma))]\\
    &\le  \sum_{i=1}^{\ell} \frac{1}{\min_i \alpha_i} \tr[\CL^{\dagg}_\ri[ \vrho](\ln( \vrho)-\ln( \vsigma))]\\
    &\le  \frac{1}{\min_i \alpha_i} \tr[\CL^{\dagg}[ \vrho](\ln( \vrho)-\ln( \vsigma))].
\end{align*}
First, we use the freedom to choose an arbitrary invariant state \eqref{eq:free_choice_invariant}. Lastly, we use that all conditional expectations share some common invariant state. Rearrange to obtain the advertised result.
\end{proof}

\subsection{Properties of finite-resolved Davies' generator}
\label{sec:simplify}
This preliminary section simplifies the Lindbladian at hand and prepares ourselves for the proof. 

\textbf{Observation I: the transition between far-apart energies cannot decrease the gap.}

ETH only tells us about the matrix element between nearby energies. First, we emphasize that the transitions between far apart energies cannot close the gap as long as each of them satisfies detailed balance.

\begin{fact}[The gap of sum of Lindbladian~\cite{Onorati_2017}]\label{fact:gap_of_sum}
For a set of $\vsigma$-detailed balance Lindbladians $\CL_s, s \in S$, the gap is monotone
\begin{align*}
    \textrm{Gap}(\sum_{s\in S} \CL_s) \ge \textrm{Gap}(\sum_{s\in S'\subset S} \CL_s). 
\end{align*}
\end{fact}
\begin{proof}
Since eigenvalues are independent of the basis, do a similarity transformation w.r.t. $\vsigma$-weighted norm to make it Hermitian. We reduce to the case where each (super-)operator is non-negative and Hermitian, which was proven in~\cite{Onorati_2017}.  
\end{proof} 

In our settings, the Gibbs state is in the common kernel of each $\CL_{\bomega}$ due to being detailed balance w.r.t. the rounded Gibbs state (Fact~\ref{fact:detailed_balance_bomega}). 
Therefore, we have reduced the problem to showing a gap for the subset of Lindbladians for which ETH holds.

\begin{align}
    \CD_{\le \Delta_{RMT}} &=\sum_{|\bomega|\le \Delta_{RMT}} \sum_{a} \gamma_{a}(\bomega)\left( \vA^{a\dagg}(\bomega) \vX \vA^a(\bomega)-\frac{1}{2}\{\vA^{a\dagg}(\bomega)\vA^a(\bomega),\vX \} \right) \\
    &=\sum_{0 \le \bomega\le \Delta_{RMT}} \sum_{a} \L[\gamma_{a}(\bomega)\left( \vA^{a\dagg}(\bomega) \vX \vA^a(\bomega)-\frac{1}{2}\{\vA^{a\dagg}(\bomega)\vA^a(\bomega),\vX \} \right)+ (\bomega \rightarrow -\bomega)\R]
     &:= \sum_{0\le\bomega\le \Delta_{RMT}} \CL_{\bomega}.
\end{align}
\textbf{Observation II: the Lindbladian splits into the block diagonal and the off-block-diagonal Sectors.}

\begin{figure}[t]
    \centering
    \includegraphics[width=0.5\textwidth]{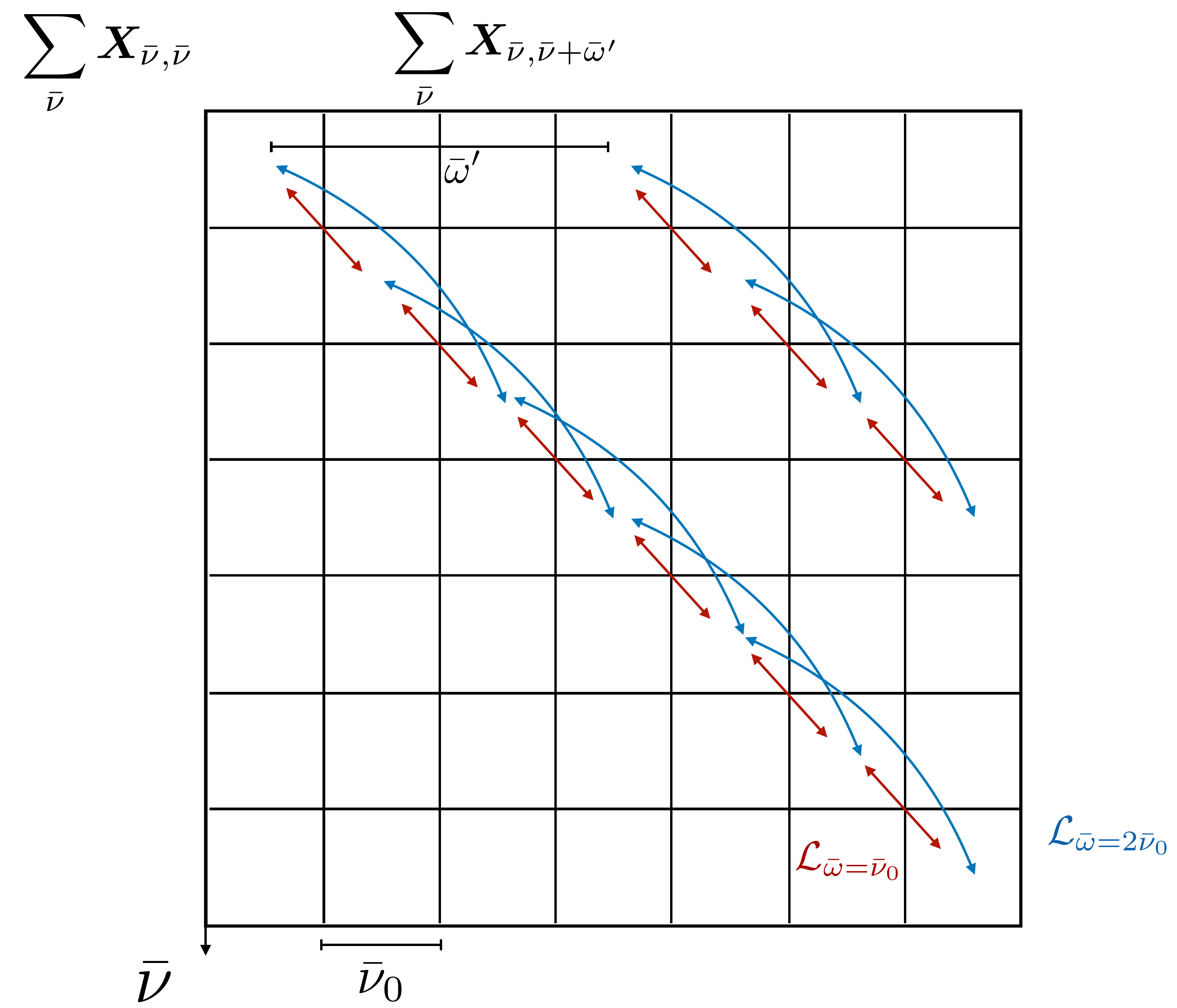}
    \caption{ The inputs $\vX $ (or $\vrho$) can be dissected into blocks per resolution $\bnu_0$. The Lindbladians term $\CL_{\bomega}$ hops within sectors (i.e., along the 45-degree slope.) The Gibbs state lies in the diagonal sector where most arguments are devoted, including the classical random walk. The off-diagonal block is simpler to analyze. 
    }
    \label{fig:L_sectors}
\end{figure}

Let us open up the operators 
\begin{align}
         \vA^{a}(\bomega) = \sum_{\bnu_1-\bnu_2=\bomega} \vP_{\bnu_2} \vA^a\vP_{\bnu_1}.
\end{align} 
The Lindbladian splits into sectors that can be discussed individually~\cite{Temme_2013_gap} (Figure~\ref{fig:L_sectors}). Labeled by $\bomega'$, any input can be decomposed into
\begin{align}
    \vX = \sum_{\bnu} \sum_{{\bnu'}} \vP_{\bnu}\vX \vP_{\bnu'}  &=\sum_{\bomega'}\sum_{\bnu} \vP_{\bnu}\vX \vP_{\bnu+\bomega'}=:\sum_{\bomega'}Q_{\bomega'}[\vX],
\end{align}
where the Gibbs state lies in the block diagonal sector $Q_{\bomega'=0}$. Observe that the Lindbladian nicely preserves each sector that $\CD Q_{\bomega'}= Q_{\bomega'}\CD$. 
The action on the \textit{diagonal blocks} ($\bomega'=0, \sum_{\bnu} \vX_{\bnu\bnu} :=\sum_{\bnu}\vP_{\bnu} \vX \vP_{\bnu} ),$ takes the form
\begin{align}
   \CL_{\bomega}[\sum_{\bnu} \vX_{\bnu\bnu}] 
    &=  \sum_{\bnu_1-\bnu_2=\bomega}\sum_{a}\bigg[\frac{\gamma_{a}(\bomega)}{2} \left( \vA^{a}_{\bnu_1\bnu_2}  \vX  \vA^{a}_{\bnu_2\bnu_1}-\frac{1}{2}\{\vA^{a}_{\bnu_1\bnu_2} \vA^a_{\bnu_2\bnu_1},\vX_{\bnu_1\bnu_1} \} \right)
     + \frac{\gamma_{a}(-\bomega)}{2} (\bnu_1\leftrightarrow\bnu_2 )\bigg]\\
    &:= \sum_{\bnu_1-\bnu_2=\bomega } \CL_{\bnu_1,\bnu_2}[Q_{\bomega'=0}[\vX]].\label{eq:diag_block}
\end{align}
Note that we fix $\bomega \ge 0$, i.e. $\bnu_1 \ge \bnu_2$. Also, the action on the \textit{off-block-diagonal} inputs, for any $\bomega'$, is   
\begin{align}
   &\CL_{\bomega}[\sum_{\bnu} \vX_{\bnu,\bnu+\bomega'}] \notag\\
    &=  \sum_{\substack{\bnu_1-\bnu_2=\bomega}}\sum_{a}\bigg[\frac{\gamma_{a}(\bomega)}{2} \L( \vA^{a}_{\bnu_1\bnu_2}  \vX   \vA^{a}_{\bnu_2+\bomega', \bnu_1+\bomega'}
    - \frac{1}{2} \vA^{a}_{\bnu_1\bnu_2} \vA^a_{\bnu_2\bnu_1}\vX_{\bnu_1,\bnu_1+\bomega'} -\frac{1}{2}\vX_{\bnu_1,\bnu_1+\bomega'} \vA^{a}_{\bnu_1+\bomega', \bnu_2+\bomega'} \vA^a_{\bnu_2+\bomega', \bnu_1+\bomega'} \R)\notag\\
    & \hspace{4cm}+(\bomega \rightarrow -\bomega, \bnu_1\leftrightarrow\bnu_2)\bigg]\\
    &:=\sum_{\substack{\bnu_1-\bnu_2=\bomega}} \CL_{\bnu_1\bnu_2,\bomega'}.\label{eq:off_block}
\end{align}

\subsection{Quantum expanders}\label{sec:expanders}

Let us briefly review the idea of quantum expanders. Consider a (self-adjoint) channel as an average over a set of unitaries
\begin{align}
    \CN: = \frac{1}{2\labs{a}} \sum_a  \vU_a [ \cdot ] \vU_a^\dagger +  \vU_a^\dagger [ \cdot ] \vU_a,
\end{align}
 then we say it gives \textit{Quantum expander}~\cite{Hastings_2007,pisier2013rand_mat_operator}, if\footnote{For the purposes of this work, we focus on the $1/\sqrt{\labs{a}}$ behavior.}
\begin{align}
    \lambda_2\L( \CN \R) \le \CO(\frac{1}{\sqrt{\labs{a}}}). 
\end{align}
The spectral gap then implies rapid convergence to the maximally mixed state
\begin{align}
\normp{\CN^\ell[\vrho] - \frac{\vI}{\tr[\vI]}}{1} \le \epsilon,\ \  \ell = \frac{\log(1/\epsilon)+2\log(dim)}{\log(\lambda_2)}.
\end{align}
For example, i.i.d. unitaries drawn from the Haar measure~\cite{Hastings_2007} give quantum expanders, where the RHS is $\frac{\sqrt{2\labs{a}-1}}{\labs{a}}$, with high probability. In general, unitaries giving expanders need not be randomized constructions. 

Back to our problem, our convergence results rely on identifying and proving (Section~\ref{sec:local_gap}) a notion of a quantum expander for Davies' generator, for both the diagonal block~\eqref{eq:diag_block} and the off-diagonal blocks~\eqref{eq:off_block}. In the case of unitary quantum expanders, the constants in $\CO(\cdot)$ are absolute. Here, there are extra parameters like the Boltzmann factors (once the common factor $\gamma(-\bomega)$ is divided) as well as the dimensions of subspaces $\vP_{\bnu_1}, \vP_{\nu_2}$. 
\begin{defn}[Quantum expanders for Lindbladian, with randomness] \label{defn:QExpander_random}
We say a set of interactions $\vA^a$ gives quantum expanders at certain energy scale $\Delta_{RMT}$ if for energies $ \Delta_{RMT}+\bnu_2 \ge \bnu_1\ge \bnu_2 , \bomega'$ and with high probability(w.r.t to randomness of $\vA^a$) 
\begin{align}
    \lambda_2\L( \CL_{\bnu_1,\bnu_2} \R ) &= \L(1 - \CO(\frac{\poly(R,\e^{\beta \Delta_{RMT}})}{\sqrt{\labs{a}}})\R)   \lambda_2\L( \BE\CL_{\bnu_1,\bnu_2} \R )\\ 
    \lambda_1\L( \CL_{\bnu_1,\bnu_2,\bomega'} \R ) &= \L(1 - \CO(\frac{\poly(R,\e^{\beta \Delta_{RMT}}) }{\sqrt{\labs{a}}})\R)   \lambda_1\L( \BE\CL_{\bnu_1,\bnu_2,\bomega'} \R ). 
\end{align}

\end{defn}
Note that the leading eigenvalue on the diagonal block Lindbladian is zero $\lambda_1=0$. Since our problem has randomness, the natural object on the RHS is the expectation. If the interactions $\vA^a$ are not random, as in any given Hamiltonian, one may substitute for an appropriate analog of the RHS. 
\section{Main Result II-1: Convergence of Davies of a rounded Hamiltonian}\label{sec:ETH_rounded_Davies_conv}

In this section, we show convergence for Davies's generator of the rounded Hamiltonian, whose dissipative part takes the form
\begin{align}
    \bCD = \sum_{\bomega}\sum_{a}\gamma(\bomega) \left( \vA^{a\dagg}(\bomega) \vX \vA^a(\bomega)-\frac{1}{2}\{\vA^{a\dagg}(\bomega)\vA^a(\bomega),\vX \} \right).
\end{align}
The bath function $\gamma(\bomega)$ only needs to satisfy detailed balance,
and a favorable choice would be the Gaussian (Theorem~\ref{thm:true_Davies}) 
\begin{align}
\gamma_{ab}(\omega) = \delta_{ab}\frac{1}{\sqrt{2\pi \Delta^2_{RMT}}}  \exp (\frac{-(\omega-\beta \Delta^2_{RMT}/2)^2}{2\Delta^2_{RMT}})\label{eq:gamma_ab_ETH}.
\end{align}
We have set $\Delta_B=\Delta_{RMT}$ so that $\gamma(\bomega)$ roughly aligns with the transitions of $ETH$.

\begin{thm}[Convergence of the rounded generator]\label{thm:ETH_rounded_Davies_convergence}
Consider a Hamiltonian $\bvH_S$ rounded to precision $\bnu_0$ and truncated with a well-defined density of states (up to the truncation point). Assume each $\vA^a$ for energy differences below the ETH window $\Delta_{RMT}$ is an i.i.d.sample from the ETH ansatz, and assume
\begin{align}
    \bnu_0 &\le \CO(\Delta_{RMT}) &\text{( high rounding precision)}\\
    R&=\CO(1) &\text{(small relative ratio of DoS)}\\
    \beta\Delta_{RMT} &= \CO(1) &\text{(small ETH window)}.
\end{align}
Then, with high probability (w.r.t to the randomness of ETH), using $\labs{a} =\tOmega(1)$ interactions and  
running the rounded Davies' generator $\bCL$ for effective time
\begin{align}
    \tau = \lambda^2 t = \tilde{\Omega}\L(  \ln(1/\epsilon) \L( \frac{1}{\alpha_{diag}} + \frac{1}{\lambda_{off}} (n+\beta \norm{\vH_S}) \R)\R)\ \textrm{    ensures    } \lnormp{\e^{\bCL^\dagg t}[\vrho]- \bvsigma}{1} \le \epsilon,
\end{align}
where
\begin{align}
\alpha_{diag}&=\tOmega\L(  \frac{  \lambda_{RW} \labs{a} \Delta_{RMT} }{n(n+\beta \norm{\vH_S}) } \min_{\bomega \le \Delta_{RMT} }\gamma(-\bomega)\labs{f_{\bomega}}^2 \R)  \\
 \lambda_{off}&= \tOmega \L( \labs{a}\Delta_{RMT}\min_{\bomega \le \Delta_{RMT} } \gamma(-\bomega) \labs{f_{\bomega}}^2 \R)\\
 \bvsigma: &= \frac{ \e^{-\beta {\bvH}_S } }{\tr [ \e^{-\beta {\bvH}_S } ]}.
\end{align}
The number $\lambda_{RW}$ is the gap of a 1d classical random walk with the characteristic step size $\sim \Delta_{RMT}$ on the Gibbs distribution. The notation $\tilde{\Omega}$ absorbs dependence on $R, \beta \Delta_{RMT}$ and logarithmic dependence on any parameters.
\end{thm}
\begin{prop}
If the density of Gibbs state satisfies assumptions in Section~\ref{sec:DoS} (e.g., a Gaussian with variance $\Delta_{Gibbs}$), then $\lambda_{RW} = \tOmega( \frac{\Delta_{RMT}^2}{\Delta_{Gibbs}^2})$.
\end{prop}
Comparing with the convergence with the realistic generator $\CL$ (Theorem~\ref{thm:ETH_true_Davies_convergence}), the generator $\bCL$ for the rounded Hamiltonian only requires $\labs{a} = \CO(1)$ interactions. (For $\CL$, the interactions depends on the gap $\labs{a} = \CO(\frac{1}{\lambda^2_{RW}})$). 
Further, due to the nice properties of generator $\bCL$, we can even isolate a randomness-free condition for ETH that, as a black box, ensures convergence.
\begin{lem}
The RMT prescription of ETH can be replaced by that the interactions $\vA^a$ give Lindbladian Quantum expanders (Definition~\ref{defn:QExpander_random}).
\end{lem}


See Section~\ref{sec:ETH} for ETH, the assumptions on density of states, and the Gibbs state density that defines the 1d classical random walk. For the $\gamma(\omega)$ given by~\eqref{eq:gamma_ab_ETH},  $\min_{\bomega} \gamma(-\bomega) =\gamma(-\Delta_{RMT}) =\theta(\frac{\exp(\beta^2 \Delta_{RMT}^2/4) }{\Delta_{RMT}})$.  

The proof is sketched as follows. In the preliminary section (Section~\ref{sec:simplify}), we have observed that (I) transition between far-apart energies can be removed, and (II) the Lindbladian splits into the block diagonal and the off-block-diagonal sectors. Most of the proof is devoted to analyzing the \textit{diagonal blocks} ($\bomega'=0, \sum_{\bnu} \vX_{\bnu\bnu} :=\sum_{\bnu}\vP_{\bnu} \vX \vP_{\bnu} ).$ There, we first show that the interactions $\vA^a$ satisfying ETH gives quantum expanders, which convert to local MLSI estimates (Section~\ref{sec:local_gap}). Second, we use approximate tensorization (Fact~\ref{fact:approx_tensor}) to lift the local MLSI constant to the global MLSI constant. This reduces the problem to a 1d \textit{classical} random walk (Section~\ref{sec:Davies_local_to_global}), whose gap estimate from conductance calculations is left in Appendix~\ref{sec:conductance_1d}.
Lastly, we show the quantum expander condition holds for the \textit{off-block-diagonal} inputs (Section~\ref{sec:off_block}). All the above estimates combine in Section~\ref{sec:proof_rounded_Davies_conv}.

\subsection{Local Gap and MLSI Estimates: ETH Gives Quantum Expanders (block-diagonal Inputs)}\label{sec:local_gap}

Let us estimate the  spectral gap for the ``local'' Lindbladian at nearby energies $\bnu_1,\bnu_2$.

\begin{lem}\label{lem:L_localgap} 
In the block diagonal sector $Q_{\bomega'=0}$ for $\bnu_1\ne \bnu_2$, each term $\CL^*_{\bnu_1,\bnu_2}$ has stationary states of form
\begin{align*}
    \lim_{t\rightarrow \infty}\e^{\CL^*_{\bnu_1,\bnu_2}t}[Q_{\bomega'=0}[ \vrho]] =  \vsigma_{\bnu_1,\bnu_2}\tr[\vP_{\bnu_2} \vrho \vP_{\bnu_2}+\vP_{\bnu_1} \vrho \vP_{\bnu_1}]+ \sum_{\bnu\ne \bnu_1,\bnu_2} \vP_{\bnu}  \vrho \vP_{\bnu} , \ \ 
\end{align*}
where $ \vsigma_{\bnu_1,\bnu_2}$ denotes the local Gibbs state 
\begin{align*}
 \vsigma_{\bnu_1,\bnu_2}:=    \frac{\vP_{\bnu_1}\e^{-\beta \bnu_1 }+\vP_{\bnu_2}\e^{-\beta \bnu_2 }}{\tr[\vP_{\bnu_1}]e^{-\beta \bnu_1 }+\tr[\vP_{\bnu_2}]e^{-\beta \bnu_2 }}.
\end{align*}
Equivalently in the Heisenberg picture, the conditional expectation gives
 \begin{align}
 \lim_{t\rightarrow \infty}\e^{\CL_{\bnu_1,\bnu_2}t}[Q_{\bomega'=0}\vX]:=E_{\bnu_1,\bnu_2}[Q_{\bomega'=0}\vX] &= \tr[ \vsigma_{\bnu_1,\bnu_2}\vX](\vP_{\bnu_2}+\vP_{\bnu_1})+ \sum_{\bnu\ne \bnu_1,\bnu_2} \vP_{\bnu} \vX \vP_{\bnu}.
 \end{align}
Furthermore, with high probability, the gap is at least 
  \begin{align}
     \lambda_{gap}(\CL_{\bnu_1,\bnu_2}Q_{\bomega'=0}) \ge \Omega \L( \frac{1}{R}\labs{f_{\bomega} }^2 \bnu_0
\labs{a}\gamma_a(-\bomega)  \R).
 \end{align}
which converts to the MLSI constant
 \begin{align}
     \alpha_{MLSI}(\CL_{\bnu_1,\bnu_2}Q_{\bomega'=0}) \ge \Omega \L( \frac{1}{R}\frac{\labs{f_{\bomega} }^2 \bnu_0 \labs{a}\gamma_a(-\bomega) }{n+\beta \bomega} \R).
 \end{align}
\end{lem}
The proof relies on analyzing the expectation and the deviation 
\begin{align}
     \hat{\CL}_{\bnu_1,\bnu_2} = (\hat{\CL}_{\bnu_1,\bnu_2} - \BE[\hat{\CL}_{\bnu_1,\bnu_2}]) +\BE[\hat{\CL}_{\bnu_1,\bnu_2}]
\end{align}
where the expectation is evaluated over the ETH ansatz (Hypothesis~\ref{hyp:ETH_diagonal}). Remarkably, this decomposition fits neatly with the RMT prescription: the expectation gives an generator of a \textit{classical} Markov chain in the energy basis, and at large $\labs{a}$, the deviation concentrates with a relative size $\CO(\frac{1}{\sqrt{\labs{a}}})$.

\subsubsection{The Expected Lindbladian}
It is more intuitive to further decompose into two Lindbladians
\begin{align}
    \BE[\hat{\CL}_{\bnu_1,\bnu_2}] = \BE[\hat{\CL}_{\vG, \bnu_1,\bnu_2}]+ \BE[\hat{\CL}_{\vD, \bnu_1,\bnu_2}].
\end{align}
For the pair of frequency $\bnu_1,\bnu_2 $, we decompose $\vA$ into independent Gaussian matrices $\vA_{\bnu_2\bnu_1}=\vG_{\bnu_2\bnu_1}+\vD_{\bnu_2\bnu_1}$: the coarse-grained part $\vG_{\bnu_2\bnu_1}$ has the same variance $\min_{\bnu_1,\bnu_2}(\BE |A_{ij}|^2)$ for all Gaussian entries $\vG_{ij}$ and the deviation part $\vD_{\bnu_2\bnu_1}$ takes care of the dependence of the refined scales $i\in \bnu$. In the following analysis we only need to consider $\vG_{\bnu_2\bnu_1}$ since $\vD_{\bnu_2\bnu_1}$ does not shrink the gap.
\begin{prop}\label{prop:expected_L}
The eigenvectors and eigenvalues of $\BE[\hat{\CL}_{\vG, \bnu_1,\bnu_2}]Q_{\bomega'=0}$ are
\begin{align*}
    \vP_{\bnu_1}+\vP_{\bnu_2}&,\ \  \lambda_1 = 0\\
    \frac{e^{\beta \bomega}}{\tr[\vP_{\bnu_1}]}\vP_{\bnu_1}-\frac{1}{\tr[\vP_{\bnu_2}]}\vP_{\bnu_2}&, \ \  \lambda_4 = -V\tr[\vP_{\bnu_2}]\sum_a\gamma_{a}(\bomega)-V\tr[\vP_{\bnu_1}]\sum_a\gamma_{a}(-\bomega)\\
    \vO_{\bnu_2}, \tr[\vO_{\bnu_2}]=0&,\ \  \lambda_2 = -V\tr[\vP_{\bnu_1}]\sum_a\gamma_a(-\bomega)\\
    \vO_{\bnu_1}, \tr[\vO_{\bnu_1}]=0&,\ \  \lambda_3 = -V\tr[\vP_{\bnu_2}]\sum_a\gamma_a(\bomega)
\end{align*}
where $V:=\min_{\bnu_1,\bnu_2}(\BE |A_{ij}|^2)$.
\end{prop}
We see that the second largest eigenvalue is $\lambda_2=-V\tr[\vP_{\bnu_1}]\sum_a\gamma_a(-\bomega)$.
The proof is based on a simple Gaussian calculation.
\begin{fact}\label{fact:expect_Gaussian}
For complex Gaussian rectangular matrix $d_2\times d_1$ 
\begin{align*}
    \BE_{G}[ \vA_{d_2d_1}(\vX_{d_1})\vA_{d_1d_2}^\dagg ]= \sum_{i_1,i_2} \ket{i_2}\bra{i_2}\cdot \vX_{i_1i_1} \cdot \BE[\vA_{i_1i_2}\vA^*_{i_2i_1}].
\end{align*}
In particular, if all entries share the same variance $V$, then
$
     \BE_{G}[ \vG_{d_2d_1}(\vX_{d_1})\vG_{d_1d_2}^\dagg ]= V \cdot \tr_{d_1}[\vX_{d_1}]\cdot \vI_{d_2}. 
$
\end{fact}
\begin{proof}[Proof of Proposition~\ref{prop:expected_L}]
Consider the coarse-grained $\vG$ component
\begin{align}
  &\sum_{a}\frac{1}{2}\BE\bigg[\gamma(\bomega) \left( \vG^{a}_{\bnu_1\bnu_2} (\vX) \vG^{a}_{\bnu_2\bnu_1}-\frac{1}{2}\{\vG^{a}_{\bnu_1\bnu_2} \vG^a_{\bnu_2\bnu_1},\vX_{\bnu_1\bnu_1} \} \right)+ \gamma(-\bomega)\left(\vG^{a}_{\bnu_2\bnu_1}(\vX)\vG^a_{\bnu_1\bnu_2} -\frac{1}{2}\{\vG^{a}_{\bnu_2\bnu_1} \vG^a_{\bnu_1\bnu_2},\vX_{\bnu_2\bnu_2} \}\right)\bigg]\\
    &=\sum_{a}\frac{V}{2}\bigg[\gamma(\bomega) \bigg( \vP_{\bnu_1}\tr[\vP_{\bnu_2}\vX \vP_{\bnu_2}] - \tr[\vP_{\bnu_2}]\vP_{\bnu_1}\vX \vP_{\bnu_1}\bigg)+\gamma(-\bomega) \bigg(\vP_{\bnu_2}\tr[\vP_{\bnu_1}\vX \vP_{\bnu_1}] - \tr[\vP_{\bnu_1}]\vP_{\bnu_2}\vX \vP_{\bnu_2} \bigg)\bigg]
\end{align}
where $V=\BE[G_{ij}G^*_{ij}] = \min_{\bnu_1,\bnu_2}(\BE |A_{ij}|^2) $ is the variance of a Gaussian entry. 
\end{proof}
Adding back the deviation $\vD$ does not shrink the gap (Fact~\ref{fact:gap_of_sum}). Writing out $\vD$, the bilinear term $A\vX D$ vanishes under the expectation
\begin{align}
     &\sum_{a}\frac{1}{2}\BE\bigg[\gamma(\bomega) \left( \vD^{a}_{\bnu_1\bnu_2} (\vX ) \vD^{a}_{\bnu_2\bnu_1}-\frac{1}{2}\{\vD^{a}_{\bnu_1\bnu_2} \vD^{a}_{\bnu_2\bnu_1},\vX_{\bnu_1\bnu_1} \} \right)+ \gamma(-\bomega)\left(\vD^{a}_{\bnu_2\bnu_1}(\vX )\vD^{a}_{\bnu_1\bnu_2} -\frac{1}{2}\{\vD^{a}_{\bnu_2\bnu_1} \vD^{a}_{\bnu_1\bnu_2},\vX_{\bnu_2\bnu_2} \}\right)\bigg]
\end{align}
It is indeed another Lindbladian that generates CPTP maps. Formally, we also see it does have the same stationary state due to the coefficient $\gamma(\bomega)$, i.e., it satisfies detailed balance.

\subsubsection{Concentration around the expectation}
Finished with the expectation, we move on for concentration. We want to control fluctuations of the second eigenvalue $\lambda_2$ through the perturbation theory of eigenvalues. The arguments here are similar to Section~\ref{sec:true_Davies_concentration}. Again, we work with the inner product 
\begin{align}
\braket{\vO_1,\vO_2}_{\bvsigma}=\tr[\vO_1^\dagg \sqrt{ \bvsigma} \vO_2 \sqrt{ \bvsigma} ],    
\end{align}
 under which $\BE[\CL_{\bnu_1,\bnu_2}] ,\CL_{\bnu_1,\bnu_2}$ are all self-adjoint.

We want to obtain concentration for each $\CL_{\bnu_1,\bnu_2}$
\begin{align}
    \CL_{\bnu_1,\bnu_2} &\equiv \sum_{a}\bigg[\frac{\gamma(\bomega)}{2} \left( \vA^{a}_{\bnu_1\bnu_2}\otimes \vA^{*a}_{\bnu_1\bnu_2} -\frac{1}{2} \vA^{a}_{\bnu_1\bnu_2} \vA^a_{\bnu_2\bnu_1}\otimes \vP_{\bnu_1} -\frac{1}{2}\vP_{\bnu_1}\otimes \vA^{*a}_{\bnu_1\bnu_2} \vA^{*a}_{\bnu_2\bnu_1}\right) +\frac{\gamma(-\bomega)}{2}(\bnu_1\leftrightarrow\bnu_2)\bigg].
\end{align}
Note that we use the entry-wise conjugate (not to confuse with transpose $\vA^*\ne \vA^\dagg$). 
\begin{prop}\label{prop:concentration_L12} The deviation $\delta \CL_{\bnu_1,\bnu_2}:=\CL_{\bnu_1,\bnu_2} - \BE[\CL_{\bnu_1,\bnu_2}]$, with high probability, is at most
\begin{align*}
      \lnormp{\delta \CL_{\bnu_1,\bnu_2}}{\infty,\bvsigma} &= \CO\L( c_{\bnu_1,\bnu_2} \cdot \e^{\beta (\bnu_1-\bnu_2)} \R),
\end{align*}
where 
\begin{align*}
     c_{\bnu_1,\bnu_2} :=  \max_{\bnu_1,\bnu_2}(\BE |A_{ij}|^2)\cdot \max\L(\tr[\vP_{\nu_1}],\tr[\vP_{\nu_2}]\R) \sqrt{\labs{a} }\gamma(-\bomega)\label{eq:c_bnu_bnu}
\end{align*}
 and $\lnormp{\cdot}{\infty,\bvsigma}:=\lnormp{\cdot}{(2,\bvsigma)-(2,\bvsigma)}$ is the operator norm w.r.t. the inner product $\braket{\cdot ,\cdot }_{\bvsigma}$.
\end{prop}

\begin{proof}
The proof is analogous to ~\ref{lem:true_davies_deltaD'}, and we briefly review the steps. We control the operator norm by the Schatten p-norm and decouple them
 \begin{align}
      (\BE \lnormp{\delta \CL_{\bnu_1,\bnu_2}}{\infty,\bvsigma}^p )^{\frac{1}{p}}&\le  (\BE\lnormp{\delta \CL_{\bnu_1,\bnu_2}}{p,\bvsigma }^p)^{\frac{1}{p}}\\
      &\le (\BE \lnormp{\delta \CL_{\bnu_1,\bnu_2} - \delta \CL'_{\bnu_1,\bnu_2}}{p,\bvsigma}^p)^{\frac{1}{p}}.
 \end{align}
We then show that $\CL_{\bnu_1,\bnu_2}$ gives quantum expanders 
\begin{align}
    \L(\BE \lnormp{\delta \CL_{\bnu_1,\bnu_2} - \delta \CL'_{\bnu_1,\bnu_2}}{p, \bvsigma}^p\R)^{\frac{1}{p}} &\le (\BE\bigg\lVert \sum_{a}\bigg[\frac{\gamma(\bomega)}{\e^{\beta \bomega/2}} \left( \vA^{a}_{\bnu_1\bnu_2}\otimes \vA'^{*a}_{\bnu_1\bnu_2}+\vA^{a}_{\bnu_2\bnu_1}\otimes \vA'^{*a}_{\bnu_2\bnu_1} \right)\\
    &\ \ \ \ -\frac{\gamma(\bomega)}{2}\left( \vA^{a}_{\bnu_1\bnu_2} \vA'^a_{\bnu_2\bnu_1}\otimes \vP_{\bnu_1} +\vP_{\bnu_1}\otimes \vA^{*a}_{\bnu_1\bnu_2} \vA'^{*a}_{\bnu_2\bnu_1}\right)\\
    & \ \ \ \ +\frac{\gamma(-\bomega)}{2} \left( \vA^{a}_{\bnu_2\bnu_1} \vA'^a_{\bnu_1\bnu_2}\otimes \vP_{\bnu_2}+ \vP_{\bnu_2}\otimes \vA^{*a}_{\bnu_2\bnu_1}\vA'^{*a}_{\bnu_1\bnu_2}\right)\bigg] \bigg\rVert_{p}^p)^{\frac{1}{p}}\\
    &=\CO\L(\max_{\bnu_1,\bnu_2}(\BE |A_{ij}|^2) \max(\tr[\vP_{\bnu_1}],\tr[\vP_{\bnu_2}])\sqrt{\labs{a} \gamma(\bomega)^2}\R) \\
    &= \CO\L( \frac{R^2\e^{2\beta \Delta_{RMT}}}{\sqrt{\labs{a}}} \lambda_2( \BE[\CL_{\bnu_1,\bnu_2}]) \R)\notag\label{eq:L12_expander}
\end{align}
where a factor 2 cancels out. The first equality throws in extra Gaussian to uniformize the variance of entries to use Gaussian concentration inequalities ( Fact~\ref{fact:concentration_GoG}, Fact~\ref{fact:concentration_GoG}, Fact~\ref{fact:rect_Gaussian_spectral_norm}). 

In the last equality, we compare with the second eigenvalue with the expectation to manifest the form of quantum expander, up to a polynomial of Boltzmann factors $\e^{\beta \Delta_{RMT}}$ (due to $\gamma(\bomega)/\gamma(-\bomega)$) and density ratios $R$ (due to $V_{max}/V_{min}$ and $\tr[\vP_{\bnu_1}]/ \tr[\vP_{\bnu_2}]$ ).
In other words, the ETH assumption can be replaced by the interactions $\vA^a$ giving quantum expanders (with a suitable choice of expectation $\BE[\CL_{\bnu_1,\bnu_2}]$. ) 
\end{proof}

\subsubsection{Proof of Lemma~\ref{lem:L_localgap}}
We will need simple facts.
\begin{fact}[gap to MLSI (see, e.g., {~\cite[Remark~3.5]{gao2021complete}})] \label{fact:gap_to_MLSI}
For primitive Lindbladian $\CL$,
\begin{align*}
    \frac{\lambda_{gap}(\CL)}{\ln(\norm{ \vsigma^{-1}})+2} \le \frac{\alpha_{MLSI}(\CL)}{2}
\end{align*}
where $\vsigma$ is the unique kernel of $\CL$.
\end{fact}
This conversion costs a factor of $n$. 

\begin{fact}\label{fact:MLSI_append_other_freq}
For Lindbladian $\CL$, define 
\begin{align*}
    ( \CL\oplus 0 )^\dagg  [ \vrho_1\oplus \vrho_2] := \CL^\dagg[\vrho_1]\oplus 0.
\end{align*}
Then 
$
    \alpha_{MLSI}(\CL\oplus 0) = \alpha_{MLSI}(\CL).
$
\end{fact}
\begin{proof}
\begin{align}
D(\vrho_1\oplus\vrho_2||E^\dagg[\vrho_1\oplus\vrho_2] )&=    \tr\L[ \vrho_1\oplus\vrho_2 \L( \ln (\vrho_1\oplus\vrho_2 )- \ln(E^\dagg[\vrho_1\oplus\vrho_2] )\R) \R]\\ 
&=\tr[ \vrho_1 \L(\ln (\vrho_1)- \ln(E^\dagg[\vrho_1])\R) ].
\end{align}
The first equality simplifies the logarithm by $\ln(\vrho_1\oplus \vrho_2)=\ln(\vrho_1)\oplus\ln(\vrho_2)$. 
Similarly,
\begin{align}
    \tr[(\CL\oplus 0)^* [\vrho_1\oplus\vrho_2] \left( \ln(\vrho_1\oplus\vrho_2) - E^\dagg[\vrho_1\oplus\vrho_2]\right)  ]&= \tr[\CL^*[ \vrho_1] \left( \ln(\vrho_1) - E^\dagg[\vrho_1]\right)  ]
\end{align}
which concludes the proof.
\end{proof}

\begin{proof}[Proof of Lemma~\ref{lem:L_localgap}]
Recall we decomposed into expectation and deviation
\begin{align}
\hat{\CL}_{\bnu_1,\bnu_2} &=\BE[\hat{\CL}_{\vD, \bnu_1,\bnu_2}] + (\hat{\CL}_{\bnu_1,\bnu_2} - \BE[\hat{\CL}_{\bnu_1,\bnu_2}])\\
&= \BE[\hat{\CL}_{\vG, \bnu_1,\bnu_2}]+(\hat{\CL}_{\bnu_1,\bnu_2} - \BE[\hat{\CL}_{\bnu_1,\bnu_2}]) + \BE[\hat{\CL}_{\vD, \bnu_1,\bnu_2}] 
\end{align}

We may drop $\BE[\CL_{\vD, \bnu_1,\bnu_2}]$ as it only increase the gap. 
By Weyl's inequality for eigenvalues~\cite{efficent_gibbs}, 
\begin{align}
    \lambda_{gap}(\CL_{\bnu_1,\bnu_2}Q_{\bomega'=0}) &\ge   \lambda_{gap} (\BE[\CL_{\vG, \bnu_1,\bnu_2}Q_{\bomega'=0}]) -    \lnormp{\CL_{\bnu_1,\bnu_2} - \BE[\CL_{\bnu_1,\bnu_2}Q_{\bomega'=0}] }{\infty,\bvsigma} \\
    &=\Omega\L( \labs{a} \min_{\bnu_1,\bnu_2}(\BE |A_{ij}|^2)\min\L( \tr[\vP_{\bnu_1}], \tr[\vP_{\bnu_2}]\R) \gamma(-\bomega )(1-\CO(\frac{R^2\e^{\beta \Delta_{RMT}}} {\sqrt{\labs{a}}}) ) \R)\\
    &=  \Omega\L( \frac{\labs{a}\bnu_0}{R} \labs{f_{\bomega} }^2 \gamma(-\bomega) \cdot (1-\CO(\frac{R^2\e^{\beta \Delta_{RMT}}} {\sqrt{\labs{a}}}) ) \R).
\end{align}
In the last line we get a factor of density ratio $R$ due to $ \min_{\bnu_1,\bnu_2}(\BE |A_{ij}|^2)\min\L( \tr[\vP_{\bnu_1}], \tr[\vP_{\bnu_2}]\R)$. This means that it suffices to choose the number of interaction terms
\begin{align}
    \labs{a} = \Omega(R^{4}\e^{2\beta \Delta_{RMT}}) 
\end{align}
to guarantee the deviation (Proposition~\ref{prop:concentration_L12}) does not close the gap from the expectation $\BE[\CL_{\vG, \bnu_1,\bnu_2}]$. We convert to MLSI constant (Fact~\ref{fact:gap_to_MLSI}) by a factor of $1/\ln(\norm{\vsigma^{-1}})=\CO(n+\beta \Delta_{RMT})$, 
which concludes the proof by Fact~\ref{fact:MLSI_append_other_freq}.
\end{proof}

\subsection{Local to Global}\label{sec:Davies_local_to_global}
Now, we have just obtained convergence for $\CL_{\bnu_1,\bnu_2}$ to a ``local'' Gibbs state
 \begin{align}
 \lim_{t\rightarrow \infty}\e^{\CL_{\bnu_1,\bnu_2}t}[Q_{\bomega'=0}\vX ]&:=E_{\bnu_1,\bnu_2}[Q_{\bomega'=0}\vX ] = \tr[ \vsigma_{\bnu_1,\bnu_2}\vX ](\vP_{\bnu_2}+\vP_{\bnu_1})+ \sum_{\bnu\ne \bnu_1,\bnu_2} \vP_{\bnu} \vX  \vP_{\bnu},\\ 
  \vsigma_{\bnu_1,\bnu_2}&:=    \frac{\vP_{\bnu_1}\e^{-\beta \bnu_1 }+\vP_{\bnu_2}\e^{-\beta \bnu_2 }}{\tr[\vP_{\bnu_1}]e^{-\beta \bnu_1 }+\tr[\vP_{\bnu_2}]e^{-\beta \bnu_2 }}.
 \end{align}
To show global convergence, recall approximate tensorization (Fact~\ref{fact:approx_tensor}) and let 
\begin{align}
    \Phi^*:=\frac{1}{m} \sum_{0< \bomega \le \Delta_{RMT}}  \L( \prod_{(\bnu_1,\bnu_2 )\in S(\bomega) }  E_{\bnu_1,\bnu_2}+\prod_{(\bnu_1,\bnu_2 )\in S'(\bomega) }  E_{\bnu_1,\bnu_2} \R) Q_{\bomega'=0},\label{eq:Phi}
\end{align}
where $m:= 2\lfloor \frac{\Delta_{RMT}}{\bnu_0} \rfloor$ is the appropriate normalization, and for each $\bomega$, we regroup the terms $E_{\bnu_1,\bnu_2}$ (with the same difference $\bomega$) into two sets $S(\bomega)$ and $S'(\bomega)$. Within $S(\bomega)$, we demand terms $E_{\bnu_1,\bnu_2}$ to act on disjoint sets
\begin{align}
    \forall (\bnu_1,\bnu_2 )\ne (\bnu'_1,\bnu'_2 )\in S(\bomega),  \{\bnu_1,\bnu_2\}\cap \{\bnu'_1,\bnu'_2\}=\emptyset,
\end{align}
and similarly for $S'(\bomega)$. This can be assigned greedily (Figure~\ref{fig:double_partition}); the specific choice does not change the subsequent proofs (up to absolute constants).

\begin{figure}[t]
    \centering
    \includegraphics[width=0.5\textwidth]{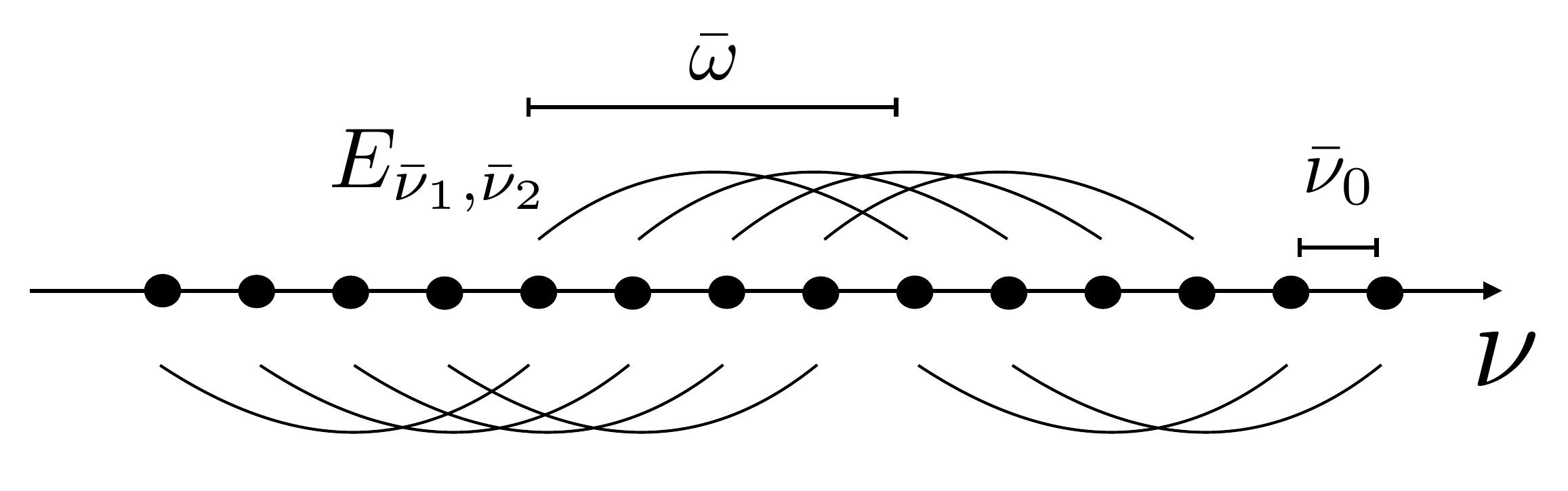}
    \caption{ Partition (For each $\bomega$) of local conditional expectations $E_{\bnu_1,\bnu_2}$ into two groups $S(\bomega), S'(\bomega)$ such that within each group, none of them overlaps. 
    }
    \label{fig:double_partition}
\end{figure}
\begin{lem} \label{lem:k_from_lambda_2}
In the block-diagonal sector $\bomega'=0$, 
\begin{align*}
    k = \Omega\L( \max_{\bnu}\labs{ \ln(\tr[\vsigma_{\bnu}]) } \cdot\big(1+\frac{1}{\lambda_{RW}} \big) \R)
\end{align*}
guarantees completely positive order
$
    (1-\epsilon)E_{global} \le (\Phi^*)^k \le (1+\epsilon)E_{global} ,
$
where
\begin{align*}
    E_{global}[\vX]:= \tr[\bvsigma\vX] \vI
\end{align*}
is the conditional expectation w.r.t the rounded global Gibbs state, and $\lambda_{RW}$ is the gap of the classical Markov chain associated with $\Phi^*$. 
This implies
\begin{align*}
    \alpha_{MLSI}\L(\sum_{0< \bomega \le \Delta_{RMT}}\sum_{\bnu_1-\bnu_2=\bomega} \CL_{\bnu_1, \bnu_2 } Q_{\bomega'=0}\R) \ge \Omega\L( \frac{m}{k} \min_{0<\bnu_1-\bnu_2 \le \bomega} \alpha_{MLSI}\L(\CL_{\bnu_1, \bnu_2 }Q_{\bomega'=0}\R)\R),
\end{align*}
with $m:= 2\lfloor \frac{\Delta_{RMT}}{\bnu_0} \rfloor$.
\end{lem}

We will see the map $\Phi^*$ is very much a \textit{classical} random walk (with a minor technical difference), whose mixing rate (in terms of the second eigenvalue) determines $k$. 

\begin{proof}
 Decompose the input operator into $Q_{\bomega'=0}[\vX]  = \sum_{\bnu} C_{\bnu}\vP_{\bnu} +\vO_{\bnu}$. The action on $\sum_{\bnu} C_{\bnu}\vP_{\bnu}$ is equivalent to a classical Markov chain.  
\begin{align}
    \Phi^* [\vP_{\bnu}] = \frac{1}{m} \L( \sum_{0< \bomega \le \Delta_{RMT}} E_{\bnu+\bomega,\bnu}[ \vP_{\bnu}]+\sum_{0< \bomega \le \Delta_{RMT}} E_{\bnu,\bnu-\bomega}[ \vP_{\bnu}] \R),
\end{align}
where when $\vP_{\bnu+\bomega}$ exceeds the boundary of the spectrum, we conveniently use the same notation $E_{\bnu+\bomega,\bnu}$ for the identity map.  Taking iterations, the input contracts towards the leading eigenvector
\begin{align}
    \Phi^{*k} [ \vP_{\bnu} ] & = \frac{\braket{\vP_{\bnu}, \vI}_{\vsigma} }{\braket{\vI, \vI}_{\vsigma}} \cdot \vI +  \Phi^{*k}[\vP_{\bnu}^{\perp}] \\
    & = E_{global}[\vP_{\bnu}]+\sum_{\bnu' } c_{\bnu'} \vP_{\bnu'},
\end{align}
with a guaranteed rate due to gap (w.r.t to $\sqrt{ \braket{\cdot,\cdot}_{\vsigma}} $)
\begin{align}
    c_{\bnu'}\sqrt{ \tr[\vsigma_{\bnu'}] }&\le \sqrt{ \sum_{\bnu' }\tr[\vsigma_{\bnu'}]c_{\bnu'}^2 }\\
      &\le (1-\lambda_{RW})^k\sqrt{ \tr[\vsigma_{\bnu}] }.
\end{align}
On the other hand, the action on the traceless parts $\sum_{\bnu} c_{\bnu}\vO_{\bnu}$ is simply 
\begin{align}
\Phi^* [\vO_{\bnu}] =\begin{cases}
0 &\textrm{if both $\bnu+\bomega$, $\bnu-\bomega$ are valid}\\
\frac{1}{2}\vO_{\bnu} &\textrm{else}.
\end{cases}        
\end{align}

To establish completely positive order, it suffices to show for inputs being a positive operator $\vK_{B,\bnu}$ acting on the subspace $\vP_{\bnu}$ tensored with arbitrary ancilla $B$. 
\begin{align}
    \Phi^{*k} [ \vK_{ B ,\bnu} ] &= \Phi^{*k}[ \vK_B \otimes \vP_{\bnu} + \vO_{B,\bnu}]\\
    & =\vK_B \otimes  \sum_{\bnu' } (\tr[\vsigma_{\bnu'}] + c_{\bnu'} )\vP_{\bnu'} + \Phi^{*k}[ \vO_{B,\bnu}],
\end{align}
where we evaluate $E_\CN[ \vP_{\bnu}] =  \sum_{\bnu' } \tr[\vsigma_{\bnu'}] \vP_{\bnu'} $. 
To establish completely positive order\footnote{For the second equation, we use that $\Phi^{*k}[ \vO_{B,\bnu}] \propto \vO_{B,\bnu}$ and the generalized depolarizing channel is completely positive.}, it suffices to show for every $\bnu' \ne \bnu$
\begin{align}
       \labs{ c_{\bnu'} } \le \epsilon \tr[\vsigma_{\bnu}],\ \        \frac{1}{2^k}+\labs{ c_{\bnu} } \le \epsilon \tr[\vsigma_{\bnu}],
\end{align}
 which is possible by demanding 
\begin{align}
    k = \Omega\L( \max_{\bnu}\labs{ \ln(\tr[\vsigma_{\bnu}])} \cdot\big(1+\frac{1}{\lambda_{RW}} \big) \R).
\end{align}
We have dropped the $\log(\epsilon)$ dependence since $\epsilon$ only needs to be a constant $(\epsilon^2\le 2\ln(2) -1 )$; the decay due to $\frac{1}{2^k}$ is absorbed into $\Omega(\cdot)$.
\end{proof}

\subsubsection{Calculations for the random walk gap $\lambda_{RW}$}

Given any Gibbs distribution, the gap $\lambda_{RW}$ of the classical random walk can be obtained via classical conductance estimates. We will calculate for Gibbs distribution satisfying Assumption~\ref{assum:delta_spec}, which includes Gaussians with variance $\Delta_{Gibbs}$. 
\begin{prop}\label{prop:Gaussian_density_gap_dumbell} 
For Gibbs distribution satisfying Assumption~\ref{assum:delta_spec} with the characteristic scale $\Delta^2_{spec}$
\begin{align*}
    \lambda_{RW} = \Omega\L( \frac{ \Delta^2_{RMT}}{ \Delta^2_{spec}  } \frac{\e^{-2\beta \Delta_{RMT}}}{R^2} \R).
\end{align*}
\end{prop}
Intuitively, the Markov chain is a 1d random walk with a characteristic scale being the ratio between the width of distribution $\theta(\Delta_{Gibbs})$ and the hop $ \theta(\Delta_{RMT}) $. 
See Appendix~\ref{sec:conductance_1d} for the conductance calculation.

\subsection{Off-block-diagonal Inputs}\label{sec:off_block}
We now show the off-block-diagonal block gives quantum expanders. 
The off-block-diagonal inputs  $\sum_{\bnu}\vP_{\bnu+\bomega'}\vX  \vP_{\bnu}=:Q_{\bomega'}[\vX ]$ has a very much similar story, with the distinction that its local leading eigenvalues are \textit{negative}. We do not need a sophisticated local-to-global estimate. We will see that the off-block-diagonal sectors contract (decohere) faster than the random walk in the diagonal block. 
\begin{align}
   &\CL_{\bomega}[\sum_{\bnu} \vX_{\bnu,\bnu+\bomega'}]=\sum_{\substack{\bnu_1-\bnu_2=\bomega}} \CL_{\bnu_1\bnu_2,\bomega'} \notag\\
    &=  \sum_{\substack{\bnu_1-\bnu_2=\bomega}}\sum_{a}\bigg[\frac{\gamma(\bomega)}{2} \L( \vA^{a}_{\bnu_1\bnu_2}  \vX   \vA^{a}_{\bnu_2+\bomega', \bnu_1+\bomega'}
    - \frac{1}{2} \vA^{a}_{\bnu_1\bnu_2} \vA^a_{\bnu_2\bnu_1}\vX_{\bnu_1,\bnu_1+\bomega'} -\frac{1}{2}\vX_{\bnu_1,\bnu_1+\bomega'} \vA^{a}_{\bnu_1+\bomega', \bnu_2+\bomega'} \vA^a_{\bnu_2+\bomega', \bnu_1+\bomega'} \R)\notag\\
    & \hspace{4cm}+(\bomega \rightarrow -\bomega, \bnu_1\leftrightarrow\bnu_2)\bigg]
\end{align}

\begin{lem}\label{lem:off_diag_contraction} For each $\bomega'$, $0 \le \bomega \le \Delta_{RMT}$, with high probability,
\begin{align*}
    \lambda_{max}\L(\sum_{0\le \bomega \le \Delta_{RMT}} \sum_{\substack{\bnu_1-\bnu_2=\bomega}} \CL_{\bnu_1\bnu_2,\bomega'}\R)  &\le -\Omega\L (\frac{1}{R}\int_{0}^{\Delta_{RMT}} \labs{a}\gamma(-\omega) \labs{f_{\omega}}^2 d\omega \R)
\end{align*}
\end{lem}
We obtain bounds for the expectation and deviation, with analogous proofs as in Proposition~\ref{prop:expected_L}, Proposition~\ref{prop:concentration_L12}.
\begin{prop}
The expected Lindbladian $\BE [\CL_{\bnu_1\bnu_2,\bomega'}]$, acting on $\vX_{\bnu_2,\bnu_2+\bomega'}+\vX_{\bnu_1,\bnu_1+\bomega'}$ has two clusters of eigenvalue 

\begin{align*}
   -\frac{\labs{a} \gamma(-\bomega)}{4} \L(\max_{\bnu_1,\bnu_2}(\tr[\vP_{\bnu_1}] \BE |A_{ij}|^2) + \tr[\vP_{\bnu_1+\bomega'}]\R) \max_{\bnu_1+\bomega',\bnu_2+\bomega'}(\BE |A_{ij}|^2)  &\le \lambda\L( P_1\BE[\CL_{\bnu_1\bnu_2}] P_1 \R)\le  (min\leftrightarrow max)\label{eq:off_block_each_freq}\\
    -\frac{\labs{a} \gamma(\bomega)}{4} \L(\tr[\vP_{\bnu_2}]\max_{\bnu_1,\bnu_2}(\BE |A_{ij}|^2)  + \tr[\vP_{\bnu_2}] \max_{\bnu_1+\bomega',\bnu_2+\bomega'}(\BE |A_{ij}|^2) \tr[\vP_{\bnu_2+\bomega'}]\R) &\le \lambda(P_2\BE[\CL_{\bnu_1\bnu_2}] P_2) \le  (min\leftrightarrow max),
\end{align*}
where superoperator $P_i[\cdot]=\vP_{\bnu_i}[\cdot]\vP_{\bnu_i+\bomega'}$ are the corresponding projectors.
\end{prop}
Note that the cross term $\vA^{a}_{\bnu_1\bnu_2}  \vX   \vA^{a}_{\bnu_2+\bomega', \bnu_1+\bomega'}$ vanishes in expectation, which is why the leading eigenvalue is negative.
\begin{prop}\label{prop:concentration_L12_bomega} The deviation $\delta \CL:=(\CL_{\bnu_1\bnu_2,\bomega'} - \BE[\CL_{\bnu_1\bnu_2,\bomega'}])$ is, with high probability, at most
\begin{align*}
 \lnormp{\delta \CL }{\infty,\bvsigma} &= \CO\L( \sqrt{c_{\bnu_1,\bnu_2}c_{\bnu_1+\bomega',\bnu_2+\bomega'} } \cdot \e^{\beta \bomega} \R) \le \CO\L(  \frac{R^2\e^{\beta\Delta_{RMT}}}{\sqrt{\labs{a}}}\lambda_{max} ( \BE \CL ) \R)
\end{align*}
\end{prop}
Where $c_{\bnu_1,\bnu_2}$ is defined in Proposition~\ref{prop:concentration_L12}. The inequality is arithmetic-geometric as well as estimates due to density variations and Boltzmann factors. 
We now prove Lemma~\ref{lem:off_diag_contraction}.
\begin{proof}
For each $ \CL_{\bnu_1\bnu_2,\bomega'}$, repeat the proof of Lemma~\ref{lem:L_localgap}. For each $\bomega$, summing over $\bnu_1$ yields an operator with global maximal eigenvalue at most as large as~\eqref{eq:off_block_each_freq}. Note that $V_{\bnu_1,\bnu_2}^{max} \tr[\vP_{\bnu_1}]\lesssim R\labs{f_{\bomega}}^2 \bnu_0 $ is only dependent on the Bohr frequency $\bnu_1-\bnu_2=\bomega$. Sum over $\bomega$ and pass to an integral to obtain the advertised result.
\end{proof}

\subsection{Proof of Theorem~\ref{thm:ETH_rounded_Davies_convergence}}\label{sec:proof_rounded_Davies_conv}
\begin{proof}
We obtain the global MLSI constant via local gap estimates (Lemma~\ref{lem:L_localgap}) and approximate tensorization ( Lemma~\ref{lem:k_from_lambda_2} )
\begin{align}
\alpha_{diag}:=\alpha( \CL\circ Q_{\bomega'=0}) = \alpha( \CD \circ Q_{\bomega'=0} )\ge \alpha( \CD_{\le \Delta_{RMT}} \circ Q_{\bomega'=0}) &=\alpha\L( \sum_{0\le \bomega \le \Delta_{RMT}}\sum_{\bnu_1-\bnu_2=\bomega } \CL_{\bnu_1,\bnu_2}\circ Q_{\bomega'=0}[\cdot] \R) \\
&\ge \alpha\L( \sum_{S_1,S_2} \CL_{S_1, S_2 } \circ Q_{\bomega'=0}\R)\\
&\ge  \Omega\L( \frac{m}{k} \min_{0<\bnu_1-\bnu_2 \le \Delta_{RMT}}\alpha\L(\CL_{\bnu_1, \bnu_2 }\R)\R)\\
&\ge \Omega\L( \lambda_{RW} \cdot\frac{ \labs{a} \Delta_{RMT}\min_{0<\bomega \le \Delta_{RMT} } \gamma(-\bomega)\labs{f_{\bomega}}^2 }{(n+\beta \norm{\vH_S}) (n+\beta \Delta_{RMT})} \R),\notag
\end{align}
where we plug in the parameters
\begin{align}
    m &= 2\lfloor \frac{\Delta_{RMT}}{\bnu_0} \rfloor,\\
    k &= \Omega\L( \max_{\bnu}\labs{ \ln(\tr[\vsigma_{\bnu}]) } \cdot\big(1+\frac{1}{\lambda_{RW}} \big) \R),\\
     \alpha(\CL_{\bnu_1,\bnu_2}\circ Q_{\bomega'=0}) &\ge  \Omega\L( \frac{\labs{a}\bnu_0 \labs{f_{\bomega} }^2 \gamma(-\bomega) }{R(n+\beta \Delta_{RMT})}\R).
\end{align}
We also need enough interactions to ensure the deviation is smaller than the expectation
\begin{align}
    \labs{a} = \Omega(R^{4}\e^{2\beta \Delta_{RMT}}) .
\end{align}

For the off-block-diagonal, the rate is much faster (Lemma~\ref{lem:off_diag_contraction})
\begin{align}
    \lambda_{off}:= - \lambda_{max}\L(\sum_{0\le \bomega \le \Delta_{RMT}} \sum_{\substack{\bnu_1-\bnu_2=\bomega}} \CL_{\bnu_1\bnu_2,\bomega'}\R)  \ge \Omega\L( \frac{\labs{a}\Delta_{RMT}}{R} \min_{\bomega \le \Delta_{RMT} } \gamma(-\bomega)  \labs{f_{\bomega}}^2\R).
\end{align}

Putting everything together,
\begin{align}
    \lnormp{\e^{\bCL^\dagg t}[\vrho]- \bvsigma}{1} &\le \lnormp{\e^{\bCL^\dagg t}Q_{\bomega'=0}[\vrho]- \bvsigma}{1}+ \lnormp{\e^{\bCL^\dagg t}Q_{\bomega'\ne 0}[\vrho]}{1}\\
    &\le  \exp(-\alpha_{diag}\tau )\cdot \sqrt{2\ln(\norm{ \bvsigma^{-1}})} + \exp(-\lambda_{off}\tau)\cdot \lnormp{\frac{1}{ \bvsigma}}{\infty}.
\end{align}
The second inequality uses standard conversions between norms \eqref{eq:ln_to_trace} for the first term and $\norm{\vrho}_1 \le \norm{\vrho}_{\vsigma^{-1},2} \le \norm{\vsigma^{-1}}$ for the second term. This means for $\epsilon$ precision, it suffices to evolve for time $\tau$.
\begin{align}
    \tau_{\epsilon}= \tilde{\Omega}\L(  \ln(1/\epsilon) \L( \frac{1}{\alpha_{diag}} + \frac{1}{\lambda_{off}} (n+\beta \norm{\vH_S}) \R)\R)\ \textrm{ensures  } \lnormp{\e^{\bCL^\dagg t}[\vrho]- \bvsigma}{1} \le \epsilon,
\end{align}
where $\tilde{\Omega}$ absorbs logarithmic dependence on $n, \beta, \norm{\vH_S}$. This is the advertised result.
\end{proof}

\subsection{Comments on Non-diagonal {$\gamma_{ab}(\bomega)$}.}
We have focused on Lindbladian with diagonal $\gamma_{ab}(\bomega) = \delta_{ab}\gamma_{aa}(\bomega)$. For other Lindbladians where $\gamma_{ab}\ne \delta_{ab} \gamma_{ab}$, our proof strategy can be adapted with a quick transformation.
For each $\bomega$, perform a basis transformation to diagonalize $\gamma_{ab}=U^\dagg_{aa'}D_{a'b'}\delta_{a'b'}U_{b'b}$. (Note $\gamma_{ab}$ is positive-semi-definite.) Then
\begin{align}
    \sum_{ab} \gamma_{ab} \vG^\dagg_a \otimes \vG_b &= \sum_{a'b'ab} U^\dagg_{aa'}D_{a'b'}\delta_{a'b'}U_{bb'} \vG^\dagg_a \otimes \vG_b\notag\\
    &=\sum_{a'b'} D_{a'b'}\delta_{a'b'}(\sum_{a}U^\dagg_{aa'}\vG^\dagg_a)\otimes(\sum_bU_{bb'}\vG_b) \stackrel{dist}{\sim} \sum_{a'} D_{a'a'} \vG^\dagg_{a'}\otimes \vG_{a'}
\end{align}
in the last line recall that unitary preserves i.i.d.Gaussian vectors $U_{b'b}g_b\sim g_{b}$. 
This means we only need to replace 
\begin{align}
    \sum_{a}\gamma_{a}(\bomega)&\rightarrow \lnormp{\vec{\gamma}(\bomega)}{1}\\
    \sqrt{\sum_{a}\gamma_{a}(\bomega)^2}&\rightarrow \lnormp{\vec{\gamma}(\bomega)}{2},
\end{align}
since 2-norm and 1-norm are preserved by the transformation. For the concentration argument (Section~\ref{sec:local_gap}), all we require for the bath is that for any constant (more precisely any fixed but potential large constant)
\begin{align}
\frac{\lnormp{\gamma(\bomega)}{1}}{\lnormp{\gamma(\bomega)}{2}} \ge Const
\end{align}
is possible for a large enough number of interaction terms $\labs{a}$.

\subsection{Comments on Optimality}\label{sec:comment_optimal}
Our convergence results, although depend polynomially on all relevant parameters, honestly speaking, seem awkwardly \textit{slow}. To find the causes, it is instructive to compare our results for the rounded generator (Theorem~\ref{thm:ETH_rounded_Davies_convergence}) with the 1d commuting cases~\cite{capel2021modified}. For a fair comparision, we fix the environment so that $\gamma_{ab} (\bomega) $ is independent of the subsystem size, $\gamma_{ab} (\bomega) = \CO(1)$ for $\bomega \le \CO(1)$. For our Lindbladians at finite resolution, set the same $\gamma_{ab}(\bomega)$, the characteristic scale of Gibbs state to be $\Delta^2_{spec} = \theta(n)$, and the interactions to be all single site operators $\labs{a}=\theta(n)$, 
\begin{align}
    \tau_{comm}=\CO(\log(n)) \text{  v.s. } \tau_{RMT}\approx \log(1/\epsilon) \frac{\Delta^2_{spec}}{ \Delta^2_{RMT}} \cdot \frac{(n+\beta \norm{\vH_S})n}{\labs{a} \int_0^{\Delta_{RMT}} \gamma(-\omega)\labs{f_{\omega}}^2 d\omega }.  
\end{align}
 The classical random walk gap and transition rates are arguably the best we can hope for, and it is the presumed parameters of ETH ($\Delta_{RMT}$,  $f_{\bomega}$) that slow us down; the part of the proof that may be improved is the conversion between norms $(n+\beta \norm{\vH_S})n$. 

In other words, if we hope our Lindbladian to be as nearly good as the commuting cases, we need (I) a version of ETH that the random matrix window is large $\Delta'_{RMT} = \CO(1)$ and the transition rate $f_{\omega}$ has most weight supported in $\omega = \pm\CO(1)$ and (II) better conversions between norms and notions of mixing. 

For the realistic generator $\CL$ (Theorem~\ref{thm:ETH_true_Davies_convergence}), we think the requirement for many interactions $\labs{a} = \Omega(\frac{1}{\lambda^2_{RW}})$ may be an artifact of the proof strategy; to obtain a fewer number of terms $\labs{a}$, we will need some version of local-to-global-lift (an analog of approximate tensorization) for approximately CPTP maps.

\subsubsection{Resource costs}
In addition to thermalization at the effective time $\tau$, the implementation of Davies's generator (Theorem~\ref{thm:true_Davies}) uses physical resources
that scales quadratically with effective time $\tau^2$. This is rooted in the weak-coupling approximation (Lemma~\ref{lem:half-integral}) and the secular approximation (Lemma~\ref{lem:true_davies_secular}) that are only provably accurate to the leading order. It is possible that the rounded generator $\bCL$ or the realistic generator $\CL$ can be implemented at a lower resource cost via weak-coupling or other methods(e.g.,~\cite{wocjan2021szegedy}), but it is not apparent from the current proof strategy. 

The estimates for the coherent width $\bmu_0$ (and the integer $m$) for the realistic generator $\CL$ and the bath size $n_B \sim\tau^3$ are potentially loose bounds. 

\section{Gaussian Calculations}\label{sec:Gaussian_calc}
\begin{fact}[recap.] For rectangular matrices $\vG_{d_2d_1}$ with i.i.d.complex Gaussian entries, with variance $\BE[G_{ij}G_{ij}^{*}]=2$
\begin{align*}
\BE \lnormp{\vG}{p}^p \le   \min(d_1,d_2) \BE \norm{\vG}^p \le \min(d_1,d_2) \cdot \L( \sqrt{\max(d_1,d_2)}^p c_1^p+ (c_2\sqrt{p})^p \R)
\end{align*}
for absolute constants $c_1, c_2$.
\end{fact}
\begin{proof}
Let us begin with the case $N=d_1=d_2$ for demonstration. We reproduce a proof in~\cite[Section~2.3.1]{Tao2012TopicsIR} via a simple epsilon net and union-bound argument. This strategy worked well here because the concentration is exponential in dimension $\e^{-\Omega(N)}$ that compensates the cardinality of the epsilon net. 
Consider a 1/2 maximal net $\Sigma_{1/2}$ on unit sphere $\lnormp{x}{\ell_2}=1$, i.e. points in $\Sigma_{1/2}$ are at least $1/2$ apart, but adding any point $x$ must be $1/2$ close to some element in $\Sigma_{1/2}$.
First, take the union bound
\begin{align}
    \Pr(\norm{\vG} \ge \lambda) \le \Pr(\exists x\in S_{1/2}, \lnormp{Gx}{\ell_2} \ge \lambda/2).
\end{align}
In other words, there exists an optimizing $y$ that $\lnormp{\vG y}{\ell_2}=\norm{\vG}$. We can find a nearby point in the net $x\in S_{1/2}$ such that $\lnormp{x-y}{\ell_2}\le 1/2$, then by the triangle inequality
\begin{align}
     \lnormp{\vG y}{\ell_2}\ge \lnormp{\vG x}{\ell_2}-\lnormp{\vG (x-y)}{\ell_2} \ge \norm{\vG}/2.
\end{align}
Namely, some point in the epsilon net must be at least half as large as the optimum. Second, we bound the probability of each event by Bernstein's inequality
\begin{fact}[Bernstein's inequality] For a sum of centered, zero-mean random variables,
\begin{align*}
\Pr(\sum x_i \ge \epsilon ) \le \exp(\frac{-\epsilon^2/2}{v+L\epsilon/3})    
\end{align*}
for $v = \sum \BE x^2_i$ and any $L$ such that for all $k>2$,  
$
    \BE[x_i^k] \le \frac{\BE[x_i^2]}{2} k! L^{k-2}.
$
\end{fact}
Hence,
\begin{align}
   \Pr\L(\lnormp{\vG x}{\ell_2} \ge (t+1)\sqrt{2N}\R) = \Pr\L( \sum^{2N} g^2_i -2N \ge (t+1)^2-1)2N \R)&\le \exp\L(-\frac{4N^2t'^4/2}{2N\BE[(g^2-1)^2]+2Nt'^2L/3}\R)\\
   &\le \exp\L(-\frac{Nt^4}{2+64t^2/3}\R)\\
   &\le \begin{cases}
   \exp(-\frac{Nt^4}{4}) &\text{if}\ t^2<3/32 \\
   \exp(-\frac{3}{128}Nt^2) &\text{if}\ t^2\ge 3/32
   \end{cases}
\end{align}
We used elementary estimates $t'^2:=(t+1)^2-1 \ge t^2$, $\BE[(g^2-1)^2]=2$, and $L=64$ for Bernsteins' inequality
\begin{align}
    \BE (g^2-1)^k \le \BE (g^2-g'^2)^k \le 2^k \cdot (2k-1)!! \le \frac{1}{2}\BE (g^2-1)^2 k! 4^k \le \frac{ k!}{2}\BE (g^2-1)^2 64^{k-2}.
\end{align}
Lastly, plugging in the union bound, 
\begin{align}
    \Pr\L(\norm{\vG} \ge (t+1)2\sqrt{2N}\R) \le 9^N \exp(-\frac{3}{128}Nt^2) \le \exp(-N(\frac{3}{128}t^2-\ln 36))
\end{align}
where the cardinality of $\Sigma_{1/2}$ is at most $(\frac{3/2}{1/4})^{2N}$ by volumetric argument. Integrating the tail,
\begin{align}
\BE \norm{\vG}^p \le \sqrt{N}^p( c_1^p+ (c_2\sqrt{\frac{p}{N}})^p ).
\end{align}
At small values of $t$, the integral gives $c_1$. At large enough values of t, the tail is exponential $\sim \e^{-Nt^2}$, which gives dependence on $p$ that is suppressed by $1/\sqrt{N}$.
When the matrix is rectangular $\vG_{d_2d_1}$ (WLG let $d_1\ge d_2$ so $t^2\ge 3/32$) we get 
\begin{align}
    \Pr\L(\norm{\vG} \ge (t+1)2\sqrt{2d_2}\R) \le 9^{d_1} \exp(-\frac{3}{128}d_2t^2), 
\implies \BE \norm{\vG}^p \le \sqrt{d_2}^p\L[ \L(1+\sqrt{\frac{d_1}{d_2}}\R)^pc_1^p + (c_2\sqrt{\frac{p}{d_2}})^p \R].
\end{align}
This is the advertised result.
\end{proof}

\begin{fact}[Recap.] For independent Gaussian matrices $\vG_i,\vG '_i$ with i.i.d.complex Gaussian entries, 
\begin{align*}
    \BE \lnormp{\sum_i a_i \vG_i\vG'_i}{p}^p \le \BE \lnormp{\vG_i\vG'_i}{p}^p (\sum_i a^2_i)^{p/2} 
\end{align*}
\end{fact}
\begin{proof}
Following \cite[Theorem~4.4]{pisier2013rand_mat_operator}, expand the expected trace
\begin{align}
    \BE \lnormp{\sum_i a_i \vG_i\vG'_i}{p}^p &= \sum_{i_p,\cdots, i_1} a_{i_p}\cdots a^*_{i_2} a_{i_1} \BE\tr[ \vG^{'\dagg}_{i_p}\vG ^\dagg_{i_p} \cdots \vG ^{'\dagg}_{i_2}\vG^\dagg_{i_2} \vG_{i_1}\vG'_{i_1}]\\
    &=\sum_{i_p,\cdots, i_1} a_{i_p}\cdots a^*_{i_2} a_{i_1} \sum_{w} \phi(w)\indicator(w\sim i_p,\cdots,i_1)\\
    &\le \sum_{w} \phi(w) \sum_{i_p,\cdots, i_1} a_{i_p}\cdots a^*_{i_2} a_{i_1} \indicator(w\sim i_p,\cdots,i_1)\\
    &\le \sum_{w} \phi(w) (\sum_i a_ia_i^*)^{p/2} \\
    &\le \BE \lnormp{\vG _1\vG '_1}{p}^p (\sum_i a^2_i)^{p/2}
\end{align}
In the second inequality we sum over all Wick contractions, each contributing with some positive function $\phi(w)$. For each sequence $i_p,\cdots, i_1$, we use indicator $\indicator(w\sim i_p,\cdots,i_1)$ to enforce $\vG_i$($\vG '_i$) only contract with its adjoint $\vG ^\dagg_i$($ \vG'^\dagg_i$). Switching the summation in the third, the observation in the fourth is that each pairing $w$ can at most come from $(\sum_i a^2_i)^{p/2}$ different combinations. Note that $a_{i_p}\cdots a^*_{i_2} a_{i_1} \indicator(w\sim i_p,\cdots,i_1)$ are non-negative as enforced by pairing. Lastly, we recombine the contractions back as the moment of one copy $\vG_1\vG '_1$. This is the advertised result.
\end{proof}
\begin{fact}\label{fact:concentration_GoG_recap}For $\vG_i, \vG'_i$ i.i.d.rectangular matrices with i.i.d.complex Gaussian entries, 
\begin{align*}
    \BE \lnormp{\sum_i a_i \vG_i\otimes  \vG'^*_i}{p}^p
    &\le (\BE \lnormp{\vG}{p}^p)^2\cdot (\sum_i a_i^2)^{p/2}.
\end{align*}
\end{fact}
\begin{proof}
Again, 
\begin{align}
     \BE \lnormp{\sum_i a_i \vG_i\otimes \vG'_i}{p}^p &= \sum_{i_p,\cdots, i_1} a_{i_p}\cdots a^*_{i_2} a_{i_1} \BE\tr[\vG^\dagg_{i_p} \cdots \vG^\dagg_{i_2} \vG_{i_1}]\BE\tr[ \vG^{'\dagg}_{i_p} \cdots \vG^{'\dagg}_{i_2}\vG '_{i_1}]^*,
\end{align}
and the rest arguments is analogous, with a different function $\phi(w)$.
\end{proof}
\section{Conductance Calculation}
Our quantum problems reduce to classical Markov chains. Here, we quickly review how to estimate the gap of Markov chains via conductance estimates and later present the calculation needed for the main text. Consider the function of two sets
\begin{align}
    Q(A\rightarrow B):= \sum_{x\in A, y\in B}\pi(x)\Pr(x\rightarrow y),
\end{align}
then the bottleneck ratio is defined by 
\begin{align}
    \phi := \min_{\pi(A) \le 1/2} \frac{Q(A\rightarrow A^c)}{\pi(A)},
\end{align}
where the RHS can be understood as the chance of leaving set $S$ after a step, divided by its stationary weight.
This gives a two-sided estimate of the spectral gap (we will only need the RHS). 
\begin{fact}[Cheeger's inequality]\label{fact:bottleneck_to_gap}
\begin{align}
1-2\phi \le \lambda_2 \le 1-\frac{\phi^2}{2}.
\end{align}
\end{fact}
See, e.g.,~\cite{Markovchain_mixing} for a textbook introduction.

\subsection{The 1d Random Walk with Gaussian Density of state}
\label{sec:conductance_1d}
\subsubsection{The Markov chain from the Lindbladian}
\begin{prop}\label{prop:Gaussian_density_lindbladian_recap}
Suppose the Gibbs state satisfies Assumption~\ref{assum:delta_spec} (e.g., Gaussian distribution with variance $\Delta^2_{spec}$).
Consider the Markov chain associated with 
\begin{align*}
    \Phi^* [\vP_{\bnu}] = \frac{1}{m} \L( \sum_{0< \bomega \le \Delta_{RMT}} E_{\bnu+\bomega,\bnu}[ \vP_{\bnu}]+\sum_{0< \bomega \le \Delta_{RMT}} E_{\bnu,\bnu-\bomega}[ \vP_{\bnu}] \R).
\end{align*}
where $E_{\bnu+\bomega,\bnu}$ (or $E_{\bnu,\bnu-\bomega}$) is the identity map if $\bnu+\bomega$ (or $\bnu-\bomega$) is outside of the spectrum.

Then, the second eigenvalue is at most 
\begin{align}
    \labs{\lambda_2} = 1- \Omega\L(\frac{\Delta^2_{RMT}}{\Delta^2_{spec}} \frac{\e^{-2 \beta \Delta_{RMT}}}{R^2}  \R).
\end{align}
\end{prop}
Note that the scale $\Delta_{RMT}$ is the length scale of hops, and this markov chain is characterized entirely by the Gibbs distribution at inverse temperature $\beta$ and the local density of state (which boils down to the ratio $R$).

\begin{figure}[t]
    \centering
    \includegraphics[width=1.0\textwidth]{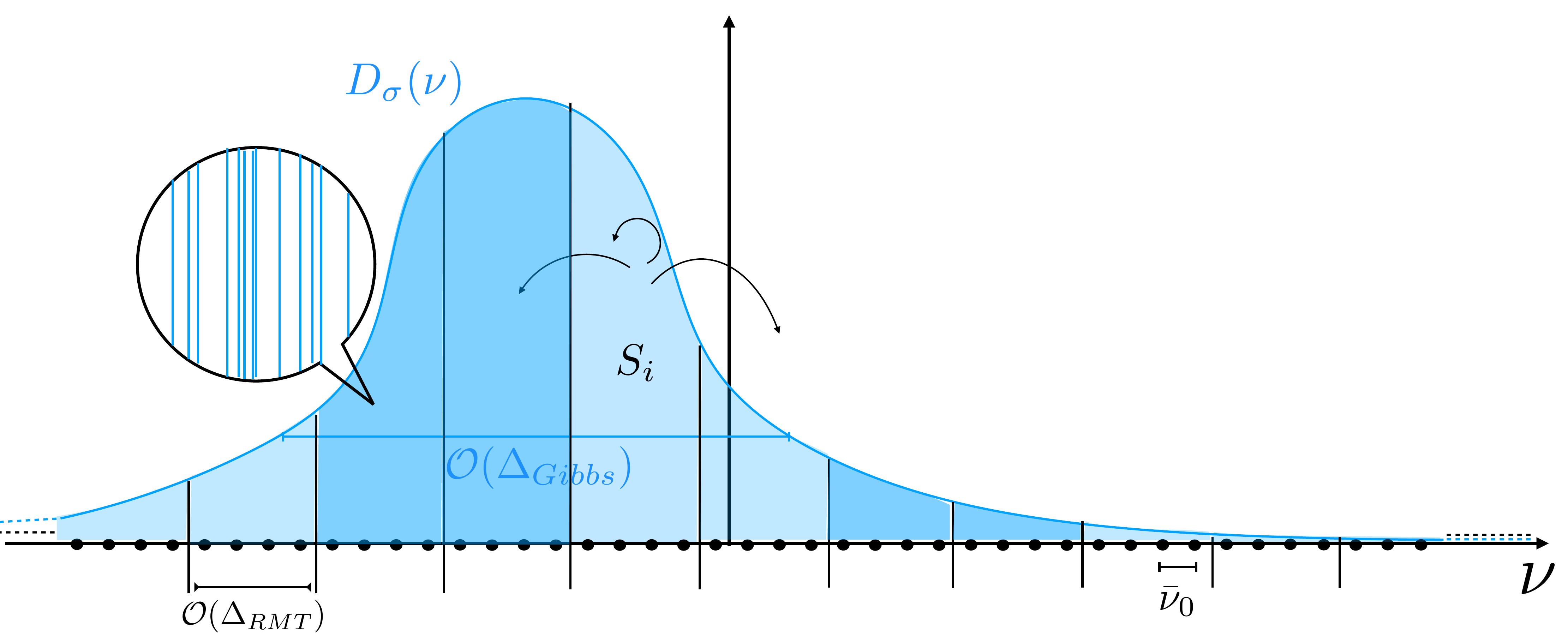}
    \caption{ Partition of the spectrum into disjoint sets $S_i$, each with size $\Delta_{RMT}$. To estimate the bottleneck ratio, for any possible set $A$, we assign for each $S_i$ a boolean varible $b_i$ whether $S_i$ has more than $3/4$ energies occupied (darker), or not (brighter).  
    }
    \label{fig:RW_on_spectrum}
\end{figure}

\begin{proof}
Let us present the proof\footnote{Some of the estimates may be adapted from an unpublished note by Charles Xu.} by having a Gaussian distribution in mind, but it directly works for distribution satisfying Assumption~\ref{assum:delta_spec}. First, let us partition the spectrum by disjoint adjacent sets $S_i$, with sizes (Figure~\ref{fig:RW_on_spectrum})
\begin{align}
\labs{S_i} \bnu_0 &= \theta(\Delta_{RMT} ),\\
\max_{\bnu_1\in S_i, \bnu_2 \in S_{i+1}}( \bnu_1 -\bnu_2 ) &\le \frac{\Delta_{RMT}}{2}.
\end{align}
To contain the `` remainder'' sites near the edge, enlarge the last two sets $S$ to trice the size.

We will have to go through the bottleneck ratio estimate. For arbitrary set $A$, assign a variable $b_i$ for each site $S_i$ (Figure~\ref{fig:RW_on_spectrum}) such that 
\begin{align}
    \begin{cases}
    b_i = 1 &\text{if $A\cap S_i$ occupies more than $3/4$ sites of $S_i$}\\
    b_i = 0 &\text{else}.
    \end{cases}
\end{align}

Let us consider a $S_i$ adjacent to $S_{i+1}$. Recall that the conditional expectation $E_{\bnu_1 , \bnu_2}$ sends any local state to the local Gibbs state 
\begin{align}
 \vsigma_{\bnu_1,\bnu_2} :=  \frac{\vP_{\bnu_1} \e^{-\beta \bnu_1} + \vP_{\bnu_2} \e^{-\beta \bnu_2} }{\tr[\vP_{\bnu_1}] \e^{-\beta \bnu_1} + \tr[\vP_{\bnu_2}] \e^{-\beta \bnu_2}}.
\end{align}
For example, suppose $b_i=1, b_{i+1}=0$, then due to the conditional expectations $E_{\bnu_1,\bnu_2}$, 
\begin{align}
    \frac{Q(A\cap S_{i} \rightarrow S_{i+1})}{\pi(A\cap S_{i})} &=\Omega\L( \frac{\e^{-\beta \Delta_{RMT}}}{R}\R), \label{eq:S_i-S_i+1}
\end{align}
where ${\e^{-\beta \Delta_{RMT}}}/{R}$ comes from a crude estimate on the ratio $\tr[\vP_{S_{i+1}}\vsigma_{\bnu_1,\bnu_2}]/\tr[\vP_{S_{i}}\vsigma_{\bnu_1,\bnu_2}]$. Similarly, suppose $b_i=1, b_{i-1}=0$, then
\begin{align}
    \frac{Q(A\cap S_{i} \rightarrow S_{i-1})}{\pi(A\cap S_{i})} &=\Omega\L( \frac{1}{R}\R), \label{eq:S_i-S_i-1}
\end{align}
and note the difference from~\eqref{eq:S_i-S_i+1} is whether the energy is increasing ($S_i\rightarrow S_{i+1}$) or decreasing ($S_i\rightarrow S_{i-1}$).
Now, decompose $A$ into disjoint sets $A=A_1+A_2$
\begin{align}
A_1&:= A \cap \{ S_i| b_i=1, \text{ or } b_{i\pm 1}=1,\},\\
A_2&:= A -A_1.
\end{align}
In other words, $A_1$ collects those sites $S_i$ with $b_i =1$, and also the nearest neighbors This choice ensures that complement $A_2$ is very ``conducting'' since its neighbors and itself has $1/4$ empty spots (i.e., $b_i=0$)
\begin{align}
    \frac{Q(A_2 \rightarrow A^c)}{\pi(A_2)} &=\Omega\L( 1 \R).
\end{align}
The bottleneck comes from $A_1$, and we have to get our hands dirty. 
We will need the density of states (weighted by Boltzmann factor) 
\begin{align}
D(\omega) = \frac{1}{\sqrt{2\pi \Delta^2_{spec}}}  \exp (\frac{-(\omega-\beta \Delta^2_{spec}/2)^2}{2\Delta^2_{spec}}),
\end{align}
which is a shifted Gaussian $\omega' = \omega-\beta \Delta^2_{spec}/2$. It suffices to consider a contiguous chain $b_L=0, b_{i+1}=1,\cdots, b_R =0$ on this density of states, with end points $\omega'_L<\omega^{'}_R$.\\
\textbf{Case 1}: If 
\begin{align}
    \omega^{'}_R \in [ -\frac{\Delta_{Gibbs}'}{4},\frac{\Delta_{Gibbs}}{4} ],
\end{align}
then 
\begin{align}
    Q(A_1 \rightarrow A^c) &=\Omega\L( \frac{\e^{-\beta \Delta_{RMT}}}{R}  \tr[ \vP_{S_{R}}\vsigma] + \tr[ \vP_{S_{L}}\vsigma ] \R)\\
    & =\Omega\L( \frac{\e^{-\beta \Delta_{RMT}}}{R} \frac{\Delta_{RMT}}{\Delta_{Gibbs}} \R),
\end{align}
where we used that $\tr[ \vP_{S_{R}}\vsigma]= \theta( \Delta_{RMT} \cdot \frac{1}{\Delta_{Gibbs}})$ and ~\eqref{eq:S_i-S_i+1}. The computation is identical for $\omega'_L \in [ -\frac{\Delta_{Gibbs}'}{4},\frac{\Delta_{Gibbs}}{4} ]$, but without the Boltzmann factor. 

\textbf{Case 2:} If 
\begin{align}
    \omega'_L , \omega^{'}_R \in [ -\frac{\Delta_{Gibbs}'}{4},\frac{\Delta_{Gibbs}}{4} ]^c ,
\end{align}
WLG let us consider $\omega^{'}_R \le -\frac{\Delta_{Gibbs}'}{4}$, then
\begin{align}
    Q(A_1 \rightarrow A^c) &=\Omega\L( \tr[ \vP_{S_{R}}\vsigma]\cdot \frac{\e^{-\beta \Delta_{RMT}}}{R} \R)\\
    & =\Omega\L(  \frac{\Delta_{RMT}}{\Delta_{Gibbs}} \exp (\frac{-\omega_R^{'2}}{2\Delta^2_{spec}} )\cdot \frac{\e^{-\beta \Delta_{RMT}}}{R}  \R),
\end{align}
but $A_1$ is weighted on a Gaussian tail
\begin{align}
    \pi(A_1)&= \CO(\int_{-\infty}^{\omega^{'}_R}  \frac{1}{\Delta_{Gibbs}} \exp (\frac{-\omega^{'2}}{2\Delta^2_{spec}} ) d\omega')\\
    &= \CO\L ( \frac{\Delta_{Gibbs}}{\omega^{'}_R} \exp(\frac{-\omega^{'2}_R}{2\Delta^2_{spec}}   \R).  
\end{align}
which means
\begin{align}
    \frac{ Q(A_1 \rightarrow A^c)}{\pi(A_1)} &= \Omega( \frac{\Delta_{RMT}\omega^{'}_R}{\Delta^2_{spec}} \frac{\e^{-\beta \Delta_{RMT}}}{R} ) \\
    &= \Omega( \frac{\Delta_{RMT}}{\Delta_{Gibbs}} \frac{\e^{-\beta \Delta_{RMT}}}{R} ).
\end{align}

Combining the estimates for $A_2,$ and disjoint subsets $A^{(j)}_1$ of $A_1$ (from case 1 or 2),
\begin{align}
    \frac{ Q( (A_2+\sum_j A^{(j)}_1 )\rightarrow A^c)}{\pi(A_2+\sum_j A^{(j)}_1)} = \frac{ Q( A_2\rightarrow A^c)+\sum_j Q( A^{(j)}_1\rightarrow A^c)}{\pi(A_2)+\sum_j\pi(A^{(j)}_1)} =  \Omega( \frac{\Delta_{RMT}}{\Delta_{Gibbs}} \frac{\e^{-\beta \Delta_{RMT}}}{R} ),
\end{align}
which converts to eigenvalue estimate by Fact~\ref{fact:bottleneck_to_gap}. More generally, the above derivation is analogous for the ``Gaussian-like'' Gibbs distribution satisfying Assumption~\ref{assum:delta_spec}, which feeds into the arguments in Case 1 and Case 2.  
\end{proof}

\subsubsection{The Markov chain generator from the realistic generator}

Consider the Markov chain generator $\vL$
\begin{align}
    L_{\nu_2\nu_1} &:= \gamma(\nu_1-\nu_2) \sum_a \BE[\labs{\vA^{a}_{\nu_2\nu_1}}^2] - \delta_{\nu_2\nu_1} \sum_{\nu_3} \gamma(\nu_1-\nu_3) \sum_a \BE[\labs{\vA^{a}_{\nu_3\nu_1}}^2].
\end{align}
Add and subtract so that 
\begin{align}
\vL &= \vL+r\vI - r\vI,    
\end{align}
where $r$ is the maximum of the second term
\begin{align}
    r: &= \max_{\bnu_1} P(\bnu_1\rightarrow \bnu_1)\\
    P(\bnu_1\rightarrow \bnu_1) &:= \sum_{\nu_3} \gamma(\nu_1-\nu_3) \sum_a \BE[\labs{\vA^{a}_{\nu_3\nu_1}}^2].
\end{align}
The calculations are very much the same as above, with the only technical tweak that we replace the size of bins $\Delta_{RMT}$ with the bath energy width $\Delta_B = \tilde{\theta}(\Delta_{RMT}) < \Delta_{RMT}$ . 
\begin{prop}
For the Markov chain generator $\vL$ and function $\gamma(\omega)$ characterized by bath energy width $\Delta_B$, the second eigenvalue is at most
\begin{align}
    \lambda_2(\vL) \le - \Omega\L( r \lambda_{RW}\R),
\end{align}
where 
\begin{align}
    r &= \Omega\L(\frac{\e^{-2\beta \Delta_{RMT}}\labs{a}}{R} \int_{-\infty}^{ \infty} \gamma(\omega)\labs{ f_{\omega}}^2 d\omega \R),\\
    \lambda_{RW} &:=\tilde{\Omega} \L(\frac{\e^{-\beta^2\Delta_{RMT}^2/4-4\beta \Delta_{RMT}}}{R^4}\frac{\Delta_{RMT}^2}{\Delta_{Gibbs}^2} \R) .
\end{align}

\end{prop}

\begin{proof}
The proof is analogous to the existing calculation (Proposition~\ref{prop:Gaussian_density_lindbladian_recap}).
The ETH ansatz ensures that 
\begin{align}
    P(\bnu_1\rightarrow \bnu_1) &= \labs{a}\sum_{\nu_2 } \frac{1}{dim(\vH) \cdot D(\frac{\nu_1+\nu_2}{2})} \labs{ f(\nu_2-\nu_1)}^2 \gamma(\nu_1-\nu_2)\\
    &\le \labs{a} R\int_{-\infty}^{ \infty} \gamma(-\omega)\labs{ f_{\omega}}^2 d\omega
\end{align}
The last inequality uses the symmetry $f_{\omega} = f_{-\omega}$, that $\gamma(-\omega)$ is largely supported within $\Delta_{RMT}$, and the KMS condition $\gamma(\omega) = \gamma(-\omega)\e^{\beta \omega}$. Similarly, we can bound the other side 
\begin{align}
    P(\bnu_1\rightarrow \bnu_1) \ge \Omega\L(\frac{\e^{-2\beta \Delta_{RMT}}\labs{a}}{R} \int_{-\infty}^{ \infty} \gamma(\omega)\labs{ f_{\omega}}^2 d\omega \R).
\end{align}
We now shift and rescale the Markov chain generator to obtain a trace-preserving Markov chain $\vM'$
\begin{align}
    \vM':= \frac{\vL +r\vI }{r}.
\end{align}
The transition rates from $S_i$ to $S_{i+1}$ reads
\begin{align} 
    \frac{Q(S_{i} \rightarrow S_{i+1})}{\pi(S_i)} &= \frac{1}{r}\theta \L(\labs{a}\min_{\nu_1\in S_i} \sum_{\nu_2\in S_{i+1}}  \frac{1}{dim(\vH) \cdot D(\frac{\nu_1+\nu_2}{2})} \labs{ f(\nu_2-\nu_1)}^2 \gamma(\bnu_1-\bnu_2)\R) \\
    &=\frac{1}{r} \Omega \L(\frac{\labs{a}}{R} \int_0^{\Delta_{B} } \gamma(-\omega)\labs{ f_{\omega}}^2 d\omega \R) =\Omega\L(\frac{\e^{-\beta^2\Delta_{RMT}^2/8-2\beta \Delta_{RMT}}}{R^2} \R).
\end{align}
In the last inequality, we use that the function $\gamma(\omega)$ is largely supported in $\labs{\omega}=\CO(\Delta_{B})$
\begin{align}
    \frac{\int_0^{\Delta_{B} } \gamma(-\omega)\labs{ f_{\omega}}^2 d\omega }{ \int_{\infty}^{\infty } \gamma(\omega)\labs{ f_{\omega}}^2 d\omega }     & =  \Omega\L(\e^{-\beta^2\Delta_{B}^2/8-2\beta \Delta_{B}}\R)\\
    &= \Omega\L(\e^{-\beta^2\Delta_{RMT}^2/8-2\beta \Delta_{RMT}}\R).
\end{align}
The rest follows from the existing calculation (Proposition~\ref{prop:Gaussian_density_lindbladian_recap}).
\end{proof}

\section{Numerical tests of expander properties}\label{sec:numerics}
We exactly diagonalize \footnote{The jupyter notebook code is available at \url{https://github.com/Shawnger/ETH_expander}. } a chaotic spin chain 
\begin{align}
    \vH = g\sum_{i=1}^L \vsigma^x_{i}+h\sum_{i=1}^L \vsigma^z_{i}+J\sum_{i=1}^L \vsigma^z_{i}\vsigma^z_{i+1}
\end{align}
up to $L=12$ qubits with periodic boundary condition $L+1\sim 1$. This is intended for a quick sanity check, and we expect to carry out larger scale numerics in follow-up work. At parameter $g= 0.9045, h=0.8090, J=1$, this model is presumably robustly chaotic~\cite{2014_Huse_ETH_numerics} and numerically tested to satisfy the traditional ETH, i.e., most eigenstates give thermal expectations~\cite{2014_Huse_ETH_numerics}. 

We wish to check the expander properties from the prediction of ETH in the sense that
\begin{align}
    \lambda_2\L( \CL_{\bnu_1,\bnu_2} \R ) &= \L(1 - \CO(\frac{1}{\sqrt{\labs{a}}})\R)   \lambda_2\L( \BE\CL_{\bnu_1,\bnu_2} \R ) = - \Omega\L( \labs{a} \cdot \labs{f(\bomega)}^2 \bnu_0 \R),\\ 
    \lambda_1\L( \CL_{\bnu_1,\bnu_2,\bomega'} \R ) &= \L(1 - \CO(\frac{1}{\sqrt{\labs{a}}})\R)   \lambda_1\L( \BE\CL_{\bnu_1,\bnu_2,\bomega'} \R ) = - \Omega\L( \labs{a} \cdot \labs{f(\bomega)}^2 \bnu_0 \R), 
\end{align}
where we set $\gamma(\bomega)=const$. The two equalities are the two checkable quantities we focus on: the scaling of the gap and the deviation from expectation.
\subsection{The gap of the Lindbladian}
At large $\labs{a}$, we want to check whether the eigenvalues scale linearly with the number of interactions $\labs{a}$
\begin{align}
    \lambda_2\L( \CL_{\bnu_1,\bnu_2} \R ) &\stackrel{?}{=} - \Omega\L(  \labs{a} \cdot \labs{f(\bomega)}^2 \bnu_0 \R),\\ 
    \lambda_1\L( \CL_{\bnu_1,\bnu_2,\bomega'} \R ) &\stackrel{?}{=}  - \Omega\L(  \labs{a} \cdot \labs{f(\bomega)}^2 \bnu_0 \R).
\end{align}
This is, after all, what we need for the proof (to feed into approximate tensorization). 
Indeed, we observe a roughly linear trend for the diagonal inputs (Figure~\ref{fig:diag_dependence}) and for the off-block-diagonal inputs (Figure~\ref{fig:off_block_diag}).
\begin{figure}[t]
    \centering
    \includegraphics[width=1.0\textwidth]{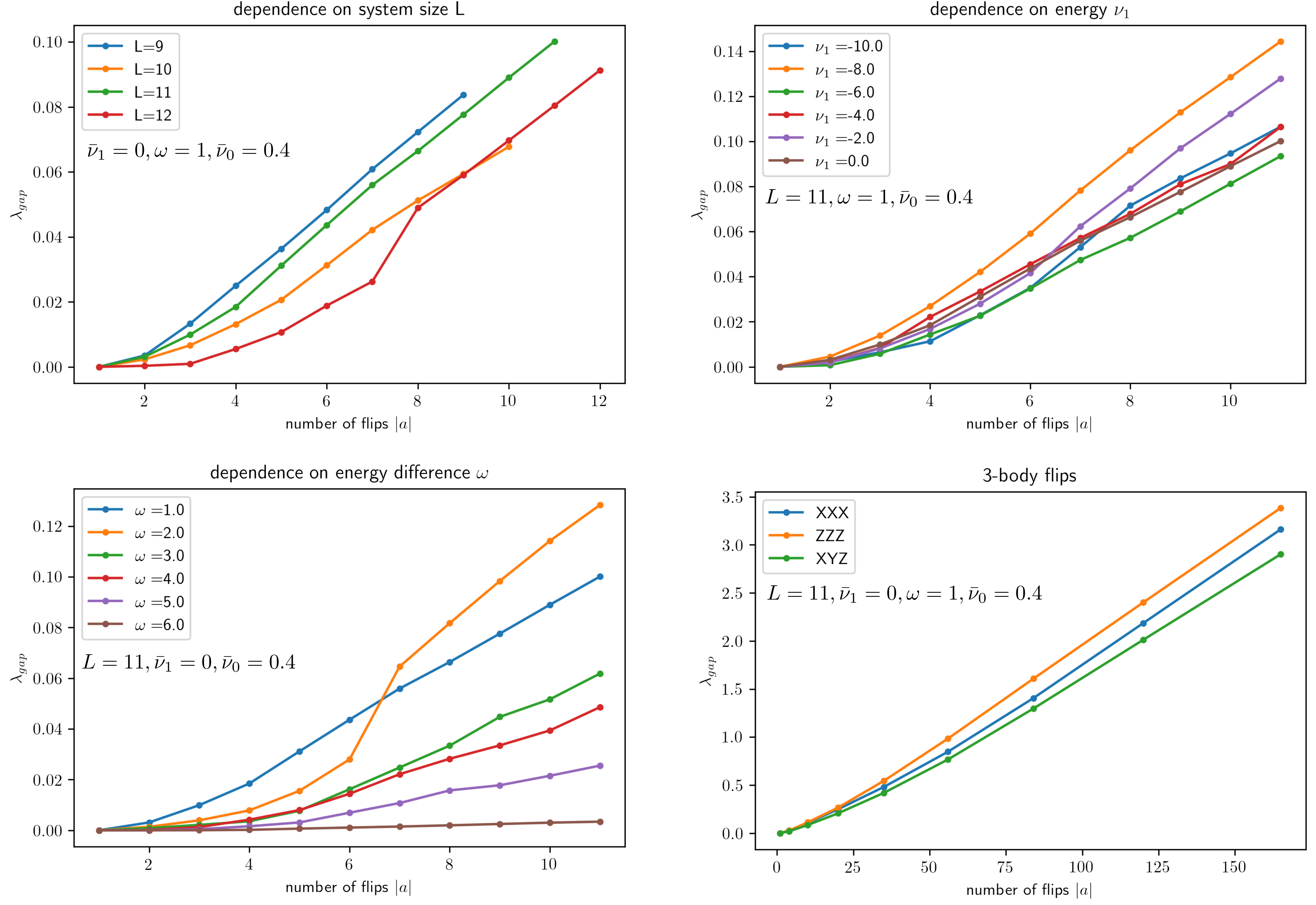}
    \caption{ The gap of the Lindbladian $\lambda_{gap}(\CL_{\bnu_1,\bnu_2})$ for block-diagonal inputs at $\bnu_1 = [\nu_1, \nu_1+\bnu_0], \bnu_2 = [\nu_1+\omega, \nu_1+\omega+\bnu_0]$, for ranges of parameters and interactions $\vsigma^x_i$. (a) The dependence on the system size $L$. (b) The dependence on energy $\bnu_1$. (c) The dependence on the energy difference $\omega$. (d) 3-body interactions for $\vsigma^x_i\vsigma^x_j\vsigma^x_k$, $\vsigma^z_i\vsigma^z_j\vsigma^z_k$, $\vsigma^x_i\vsigma^y_j\vsigma^z_k$.
    }
    \label{fig:diag_dependence}
\end{figure}

\begin{figure}[t]
    \centering
    \includegraphics[width=1.0\textwidth]{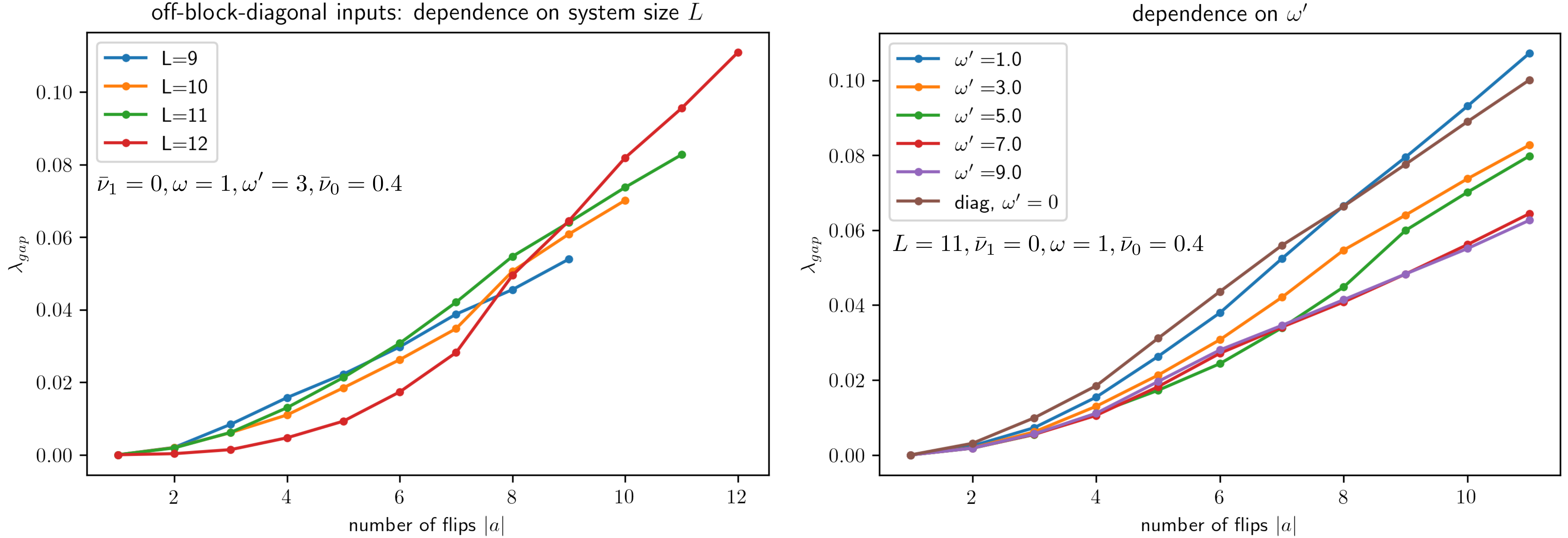}
    \caption{ The gap of the Lindbladian $\lambda_{gap}(\CL_{\bnu_1,\bnu_2,\omega'})$ for off-block-diagonal inputs at $\bnu_1 = [\nu_1, \nu_1+\bnu_0], \bnu_2 = [\nu_1+\omega, \nu_1+\omega+\bnu_0], \bnu_1' = [\nu_1+\omega', \nu_1+\bnu_0+\omega'], \bnu_2' = [\nu_1+\omega+\omega', \nu_1+\omega+\bnu_0+\omega']$, for ranges of parameters and interactions $\vsigma^x_i$. (a) The dependence on the system size $L$. (b) The dependence on the energy difference $\omega'$. Indeed, for a wide range of $\omega'$, the gap is within a ratio of $\sim 2$ from each other.
    }
    \label{fig:off_block_diag}
\end{figure}

\subsection{The deviation}
Next, we check something stronger than we need for the proof as a more refined probe to RMT: whether the Lindbladian is close to the ``expectation''. 
\begin{align}
    \norm{ \CL_{\bnu_1,\bnu_2} - \hat{\BE}\CL_{\bnu_1,\bnu_2}  } &\stackrel{?}{=}  \CO\L(\frac{1}{\sqrt{\labs{a}}}   \lambda_2\L(\hat{\BE}\CL_{\bnu_1,\bnu_2} \R )\R)\\
    \norm{ \CL_{\bnu_1,\bnu_2,\bomega'} - \hat{\BE}\CL_{\bnu_1,\bnu_2,\bomega'} } &\stackrel{?}{=}  \CO\L(\frac{1}{\sqrt{\labs{a}}}   \lambda_1\L( \hat{\BE} \CL_{\bnu_1,\bnu_2,\bomega'} \R )\R).
\end{align}
Note that numerically, we do not have access to the idealized expectation, and we just manually drop all the cross-terms that supposedly have zero-mean, denoted by $\hat{\BE}$ 
\begin{align}
    \hat{\BE} \CL_{\bnu_1,\bnu_2}[\vX] := \sum_{a}\bigg[\frac{\gamma_{a}(\bomega)}{2} \left( \vA^{a}_{\nu_1\nu_2}  \vX  \vA^{a}_{\nu_2\nu_1}-\frac{1}{2}\{\vA^{a}_{\nu_1\nu_2} \vA^a_{\nu_2\nu_1},\vX_{\nu_1\nu_1} \} \right)
     + \frac{\gamma_{a}(-\bomega)}{2} (\nu_1\leftrightarrow\nu_2 )\bigg]
\end{align}
where $\vA_{\nu_1\nu_2}:= \ket{\nu_1}\bra{\nu_1}\vA_{\bnu_1\bnu_2} \ket{\nu_2}\bra{\nu_2}$\footnote{In other words, the energies are disentangled by taking $\hat{\BE}$; this also resembles the original Davies' generator at infinite time~\eqref{eq:davies_WCL}}. Indeed, we observe that the deviation scales slower than the eigenvalues (Figure~\ref{fig:deviation}).
\begin{figure}[t]
    \centering
    \includegraphics[width=1.0\textwidth]{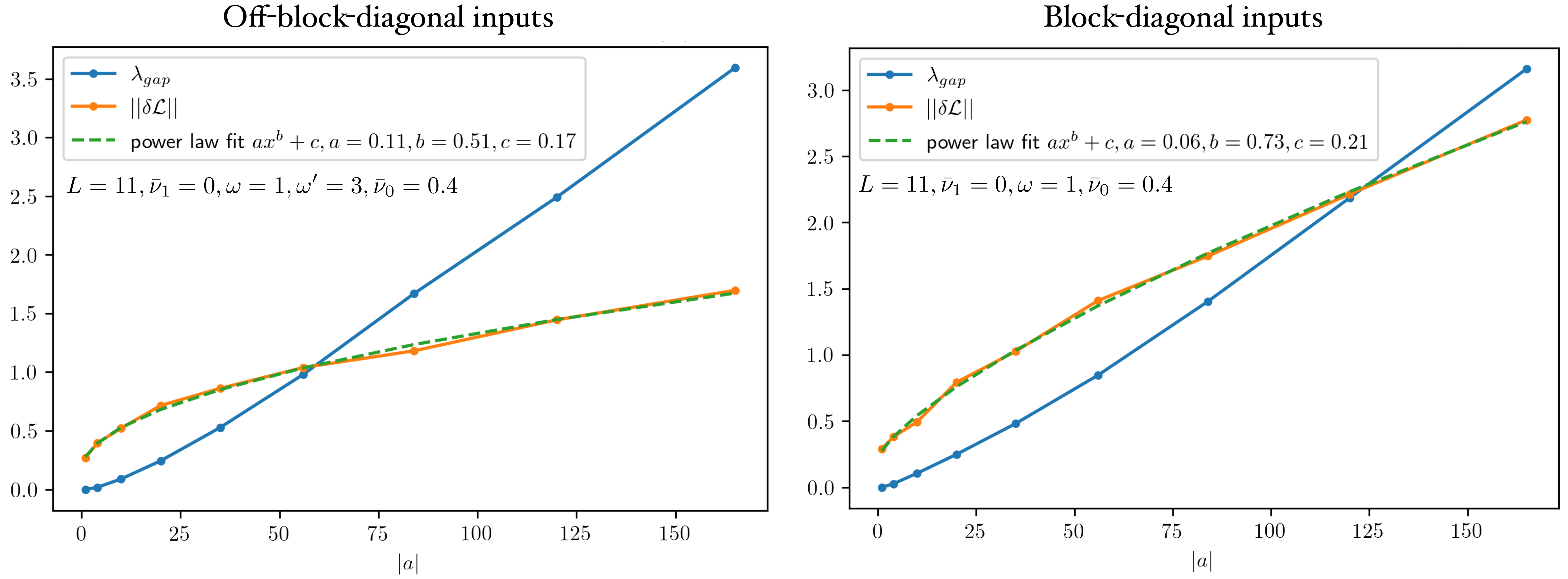}
    \caption{ The deviation from the expectation for block-diagonal inputs $\norm{\CL_{\bnu_1,\bnu_2}-\hat{\BE}[\CL_{\bnu_1,\bnu_2}]}$ and off-block-diagonal inputs $\norm{\CL_{\bnu_1,\bnu_2,\omega'}-\hat{\BE}[\CL_{\bnu_1,\bnu_2,\omega'}]}$. The interactions are 3-body $\vsigma^x_i\vsigma^x_j\vsigma^x_k$, and the energy windows are parameterized by $\bnu_1 = [\nu_1, \nu_1+\bnu_0], \bnu_2 = [\nu_1+\omega, \nu_1+\omega+\bnu_0], \bnu_1' = [\nu_1+\omega', \nu_1+\bnu_0+\omega'], \bnu_2' = [\nu_1+\omega+\omega', \nu_1+\omega+\bnu_0+\omega']$. For both, we see the deviation scales slower than the gap. 
    }
    \label{fig:deviation}
\end{figure}
\end{document}